\PassOptionsToPackage{dvipsnames}{xcolor}
\documentclass[acmsmall,screen]{acmart}
\usepackage[utf8]{inputenc} 
\usepackage[dvipsnames]{xcolor}
\usepackage{amsmath, amsthm}
\usepackage{thmtools}
\makeatletter
\@ifundefined{newcounteralias}{}{%
  \renewcommand\thmt@autorefsetup{\@xa\def\csname\thmt@envname autorefname\@xa\endcsname\@xa{\thmt@thmname}}%
}
\makeatother
\usepackage{mathtools}
\usepackage[english]{babel}
\usepackage{graphicx}
\usepackage{makecell}
\usepackage{parskip}
\usepackage{array}
\usepackage{tabularx}
\usepackage{tikz} 
\usepackage{lineno}
\usepackage{csquotes}
\usepackage{stackengine}
\usepackage{bm}
\usepackage{diagbox}
\usepackage{verbatim}
\usepackage{multirow}
\usepackage[inline]{enumitem}
\usepackage{mdframed}
\usepackage{listings}
\usepackage{subcaption}
\usepackage{hyperref}
\usepackage{etoolbox}

\usepackage{pgfplots}

\usepackage[ruled,linesnumbered,algo2e]{algorithm2e}

\usepackage{color,soul}

\usepackage[capitalise]{cleveref}
\usepackage{todonotes}
\usepackage{booktabs}
\usepackage{pgfplotstable}
\usepackage{mathpartir}
\usepackage{colortbl}
\usepackage[group-separator={,}, group-minimum-digits=4, text-series-to-math, propagate-math-font]{siunitx}
\usetikzlibrary {arrows.meta}
\usetikzlibrary{shapes,snakes}
\usetikzlibrary {decorations.pathmorphing}
\usetikzlibrary{fadings} 

\pgfdeclarelayer{background layer}
\pgfdeclarelayer{foreground layer}
\pgfsetlayers{background layer,main,foreground layer}
\input{macros}

\title{On the Decidability of Verification under Release/Acquire}

\author{Giovanna Kobus Conrado}
\orcid{0000-0001-9474-6505}
\affiliation{%
  \institution{Aarhus University}
  \city{Aarhus}
  \country{Denmark}
}
\email{gkc@cs.au.dk}

\author{Andreas Pavlogiannis}
\orcid{0000-0002-8943-0722}
\affiliation{%
  \institution{Aarhus University}
  \city{Aarhus}
  \country{Denmark}
}
\email{pavlogiannis@cs.au.dk}

\setcopyright{acmlicensed}
\acmJournal{PACMPL}
\acmVolume{X}
\acmNumber{POPL}
\acmArticle{X}
\acmYear{2027}
\acmMonth{1}
\acmDOI{10.1145/nnnnnnn.nnnnnnn}

\begin{document}

\setlength{\textfloatsep}{8pt plus 1.0pt minus 2.0pt}
\setlength{\intextsep}{8pt plus 1.0pt minus 2.0pt}
\setlength{\abovecaptionskip}{4pt plus 1pt minus 1pt}
\setlength{\abovedisplayskip}{4pt}
\setlength{\belowdisplayskip}{4pt}
\setlength{\abovedisplayshortskip}{4pt}
\setlength{\belowdisplayshortskip}{4pt}

\begin{abstract}
The verification of concurrent programs under weak-memory models is a burgeoning effort, owing to the increasing adoption of weak memory in concurrent software and hardware.
Release/Acquire has become the standard model for high-performance concurrent programming, adopted by common mainstream languages and computer architectures.
In a surprising result, Abdulla et al. (PLDI'19) proved that reachability in this model is undecidable when programs have access to atomic Read-Modify-Write (RMW) operations.
Moreover, undecidability holds even for executions with just 4 contexts, and is thus immune to underapproximations based on bounded context switching.
The pure RMW-free case was left as a challenging question, and despite the attention it has received, it has remained open for the past seven years.

In this paper, we settle the above open question by proving that reachability is undecidable even in the RMW-free fragment of Release/Acquire, thereby characterizing the simplest communication primitives that give rise to undecidability.
Moreover, we prove that bounding both the number of context switches and the number of RMWs recovers decidability, thus fully characterizing the boundary of decidability along the dimensions of RMW-bounding and context-bounding.
\end{abstract}

\ccsdesc[500]{Software and its engineering~Formal software verification}
\ccsdesc[500]{Theory of computation~Logic and verification}

\keywords{concurrency, weak-memory models, bounded context switches}

\setcopyright{cc}
\setcctype{by}
\acmJournal{PACMPL}
\acmYear{2027} \acmVolume{11} \acmNumber{POPL} \acmArticle{6}
\acmMonth{1} \acmDOI{10.1145/3871708}

\maketitle

\section{Introduction}\label{SEC:INTRO}

The standard sequential view of concurrent programs~\cite{Lamport1979} allows to reason about program executions in terms of concrete interleavings between different threads.
Although this view is intuitive and reflects the mental model of most practitioners, it fails to faithfully account for modern  program executions, where compiler and architectural optimizations may cause threads to hold divergent views of the shared memory.
Weak memory models~\cite{Sarkar2009,Alglave2012,Batty2011,Vafeiadis2015,Pulte2017,Raad2019} are formal specifications of such weak executions that a program may exhibit due to such divergence, and are becoming a robust and standard means for rigorous high-performance concurrent programming.

One of the most fundamental modern concurrency models follows the Release/Acquire paradigm, which constitutes the realization of message-passing-style communication over shared memory.
In high level,
\begin{enumerate*}[label=(\roman*)]
\item interthread communication is causally consistent, as per standard in message passing, and
\item the shared memory further imposes a coherence order on each memory location, which requires that all communication via each location appears sequentially consistent.
\end{enumerate*}
The Release/Acquire model is found in the concurrency semantics of several modern programming languages and hardware platforms~\cite{Batty2011,Pulte2017}.

As the behavior of a concurrent program depends on the underlying memory model, program-verification algorithms are memory-model dependent.
The most common verification question is state reachability:~\emph{``can the program reach a particular state, e.g., one that violates an assertion?''}.
Although reachability under Sequential Consistency has been long known to be $\PSPACE$-complete~\cite{Kozen1977,Sistla1985}, its study under weak memory is an ongoing effort in the programming languages community.
The intricacies of weak interthread interaction make the design of verification algorithms particularly challenging, while the problem itself is typically of much higher complexity (e.g., non-primitive recursive for $\tsomm$/$\psomm$~\cite{Atig2010}).

Verification under Release/Acquire was studied in~\cite{Abdulla19}, where the problem was shown to be undecidable.
The authors found this \emph{``[...] quite a surprising result considering the simplicity and intuitiveness [of Release/Acquire semantics]''}.
This undecidability becomes even more startling when contrasted to close variants of Release/Acquire, namely Strong/Weak/Localized Release/Acquire, which enjoy decidable reachability~\cite{Lahav2022,Singh2024}.

Interestingly, the undecidability proof of~\cite{Abdulla19} makes intricate use of the somewhat complex atomic Read-Modify-Write (RMW) operations, and does not hold for programs that only access the shared memory via traditional reads and writes.
The decidability over RMW-free programs was left as a challenging open problem, establishing a non-primitive recursive lower bound as a first step, via a reduction from lossy channel systems~\cite{Abdulla1993,Schnoebelen2002}.
This question has remained open for the past seven years, also resurfacing in other works in the meantime~\cite{Lahav2022,Singh2024}.

\Paragraph{Our contributions.}
Our first contribution settles the above open question.

\begin{restatable}{theorem}{thmundecidability}\label{thm:undecidability}
State reachability for RMW-free programs under Release/Acquire is undecidable\footnote{\Revision{\cref{thm:undecidability} has also been formalized in Lean 4, available at }\url{https://github.com/kobusgiovanna/decidability-ra}.}.
\end{restatable}

\cref{thm:undecidability} strengthens the undecidability of~\cite{Abdulla19} to the minimal setting of programs with only read and write memory operations.
This is, perhaps, even more surprising, considering that RMW-free programs have fundamental limitations under Release/Acquire~\cite{Castaneda24}, e.g., they cannot implement lock-free mutual exclusion in the style of the classical protocols by Petersen~\cite{Peterson1981} and Dekker~\cite{Dijkstra1962}.

Our proof is via a reduction from Post's Correspondence Problem (PCP), following the high-level approach of~\cite{Abdulla19}, where a set of guesser threads guesses a solution to PCP and communicates it to a set of \Revision{checker} threads.
The crux of that proof is a ``no-skipping'' property, where the \Revision{checker} threads make heavy use of RMWs in order to not skip any part of the solution communicated by the guesser threads, \Revision{essentially simulating a perfect, asynchronous, one-way channel.}
\Revision{In the RMW-free case, simulating such a channel is considerably more challenging, and it has been open whether it can be established with plain read and write operations.
}

One standard approach to tractable verification of concurrent programs is by under-approximating program behavior by bounding the number of context switches between threads~\cite{Qadeer2005}.
Unfortunately, the proof of~\cite{Abdulla19} shows that, under Release/Acquire with RMWs, undecidability holds even with just 4-context executions, and thus bounding context switches is not sufficient to make the problem decidable.
On the other hand, the no-skipping property in our reduction behind \cref{thm:undecidability} uses an unbounded number of context switches.
It is thus natural to ask:~do bounded context switches recover decidability for RMW-free programs?
We show that this is indeed the case.

\begin{restatable}{theorem}{thmboundedcontextswitches}\label{thm:bounded_context_switches}
State reachability for RMW-free programs under Release/Acquire is decidable when restricted to executions with a bounded number of context switches.  
Moreover, the problem remains decidable for general programs when the number of RMWs in an execution is also bounded. 
\end{restatable}

\cref{thm:undecidability} and \cref{thm:bounded_context_switches}, along with the undecidability of~\cite{Abdulla19}, fully characterize the decidability landscape for reachability under Release/Acquire, with respect to context switches and RMWs:~bounding the number of each is both sufficient and necessary for decidability.

\section{Preliminaries}\label{SEC:PRELIM}

In this section we define concurrent programs and their Release/Acquire semantics.

\Paragraph{General notation.}
Given an integer $i\in \mathbb{N}^+$, we let $\IntSet{i}=\{1,\cdots, i\}$. 
For a binary relation $B$, we write $B^{?}$,  $B^{+}$ and $B^*$ for the reflexive, transitive and reflexive-transitive closure of $B$, respectively, while $B^{-1}$ denotes the inverse of $B$.
We often write $\event_1\LTo{B}\event_2$ to denote that $(\event_1, \event_2)\in B$, and call this a $B$-edge.
We also write $\event_1\nLTo{B}\event_2$ to denote that $(\event_1, \event_2)\not\in B$.
Given two binary relations $A$, $B$, we denote by $A;B$ their composition.
Given a set $X$, we let $[X]$ be the identity relation over $X$.

We take $\ValueDom \subseteq \mathbb{N}$, $\LocationDom \subseteq \{x,y,\dots\}$, $\ThreadDom \subseteq \{t_1,t_2, \dots\}$  to be finite sets of values, shared memory locations and thread identifiers, respectively.

\subsection{Concurrent Programs}\label{SUBSEC:PRELIM_PROG}

We define sequential programs as labeled transitions systems (LTS). 
At a high level, a sequential program $\Prog_t$ of some thread $t\in\ThreadDom$ is an LTS in which each of its transitions emits a letter indicating an interaction with the shared memory via writes, reads and read-modify-writes. 
A concurrent program $\ConcProg$ is the composition of its corresponding sequential programs.

\Paragraph{Labeled Transition Systems.}
A \emph{Labeled Transition System} (LTS) over an alphabet $\Alphabet$ is a tuple $\LTS=\tuple{\States, \Trans, \InitState, \FinalState}$, where
\begin{enumerate*}[label=(\roman*)]
\item $\States$ is a finite set of states,
\item $\Trans\subseteq \States\times \Alphabet \times \States$ is a set of transitions,
\item $\InitState\in \States$ is the unique initial state, and
\item $\FinalState\in \States$ is the unique final state.
\end{enumerate*}
We use $\LTS.\States$, $\LTS.\Trans$, $\LTS.\InitState$ and $\LTS.\FinalState$  to refer to the respective components of $\LTS$. We often write $p\LTransTo{\alpha}q$ to denote that $(p,\alpha, q)\in \Trans$. A state $q\in \LTS.\States$ is reachable in $\LTS$ if $\InitState \LTransTo{\alpha}^* q$.
A sequence $\sigma=\alpha_1,\alpha_2,\dots,\alpha_n$ such that
$\InitState\LTransTo{\alpha_1}p_1\LTransTo{\alpha_2}\dots\LTransTo{\alpha_n}p_{n}$, where $p_i \in \States$ is a \emph{word} of $\LTS$. 
A word $\sigma=\alpha_1,\alpha_2,\dots,\alpha_n$ \emph{reaches} some state $q\in \LTS.\States$ if 
$\InitState\LTransTo{\alpha_1}p_1\LTransTo{\alpha_2}\dots p_{n-1}\LTransTo{\alpha_n}q$, where  $p_i \in \States$.

\Paragraph{Event labels.}
An \emph{event label} is either a \emph{read label} $\rd(t,  x, v_r)$, a \emph{write label} $\wt(t, x, v_w)$, or a \emph{read-modify-write label} $\rmw(t,x,v_r,v_w)$, where $t\in \ThreadDom$, $x\in\LocationDom$, and $v_r,v_w \in \ValueDom$.
For example, $\rd(t, x, v_r)$ denotes the action of thread $t$ reading the value $v_r$ from the shared memory location $x$.
We define $\LabelDom$ to be the set of all event labels and $\LabelDom^t$ to be the set of all labels with thread id $t$.

\Paragraph{Sequential programs.}
A \emph{sequential program} $\Prog_t$ of thread $t$ is represented as an LTS over the alphabet $\LabelDom^t$.
A transition $p\LTransTo{\alpha}q$ of $\Prog_t$ captures the fact that $\Prog_{t}$ executes an instruction that writes to/reads from the shared memory, where $\alpha$ labels the respective operation.
In doing so, $\Prog_t$ transitions from its state $p$ to state $q$.

\Paragraph{Concurrent programs.} 
A \emph{concurrent program} (or simply a \emph{program}) $\ConcProg = \tuple{\Prog_{t_1}, \Prog_{t_2}, \dots, \Prog_{t_\numThreads}}$  is the composition of sequential programs $\Prog_{t_i}$ for $t_i\in \ThreadDom$. 
The state of a concurrent program is a sequence $\tuple{p_1,p_2,\cdots, p_{\numThreads}}$, where $p_i \in \Prog_{t_i}.\States$. 
\Revision{We write $|\ConcProg.\States|=\max_{i\in\IntSet{\numThreads}}|\Prog_{t_i}.\States|$ for the maximum number of states of a sequential program in $\ConcProg$.} 
The initial and final states of $\ConcProg$ are $\tuple{\Prog_{t_1}.\InitState, \Prog_{t_2}.\InitState, \dots, \Prog_{t_{\numThreads}}.\InitState}$ and $\tuple{\Prog_{t_1}.\FinalState, \Prog_{t_2}.\FinalState, \dots, \Prog_{t_{\numThreads}}.\FinalState}$, respectively.
\subsection{Release/Acquire Semantics}\label{SUBSEC:PRELIM_EG} 

We now present the Release/Acquire memory model.

\Paragraph{Events.}
Events represent instructions of a program that interact with shared memory, and they are simply event labels (as emitted by $\ConcProg$) attached to some arbitrary unique identifier. 
Formally, an \emph{event} $\event$ is a pair $\tuple{\id, \llab}$, where $\id\in \Nats$ is an identifier and $\llab\in \LabelDom$ is a label. 
We write $\llab(\event)$, $\tid(\event)$, $\lloc(\event)$ and  $\op(\event)$ to refer to the label, thread id, location and operation, respectively, of $\event$. 
If $\op(\event)\in \{\wt,\rmw\}$, then $\event$ is a \emph{writing} event, while if $\op(\event)\in \{\rd,\rmw\}$, then $\event$ is a \emph{reading} event.
We write $\val_{\rd}(\event)$ and $\val_{\wt}(\event)$ to refer to the value read by a reading event $\event$ and the value written by writing event $\event$, respectively. 
We often identify an event only by its label, or even a partial label, when some components of the label are irrelevant or clear from the context. 
For example, we may refer to an event $\rd(t,x,v_r)$, $\rd(t,x)$ or $\rd(x)$. 
We let $\EventDom$ be the domain of events; $\ReadDom$, $\WriteDom$ and $\RMWDom$ are its subsets containing only read, write and RMW events, respectively. 

Given a set of events $X$, we let $\locx{X}$ denote the set of events of $X$ on location $x$, i.e, $\locx{X}=\{\event \in X\colon \lloc(\event)=x\}$. 
Similarly,  $X^t$ is the set of events of $X$ on thread $t$, i.e, $X^t=\{\event\in X\colon \tid(\event)=t \}$. 
We extend this notation to binary relations $B$ over $\EventDom$, e.g., $\locx{B}=\{(\event_1,\event_2) \in B\colon \lloc(\event_1)=\lloc(\event_2)=x\}$. 

\Paragraph{Execution graphs.} 
As per standard in Release/Acquire and related memory models, we represent program executions as \emph{execution graphs} $G=\tuple{\Events, \po, \rf, \mo}$.
We often write $G.\Events$, $G.\po$, $G.\rf$ and $G.\mo$ for the respective components of $G$.
$G.\Events$ denotes the set of events that participate in the execution.
The \emph{program order} $G.\po\subseteq (\bigcup_{t\in \ThreadDom}{\EventDom^t} \times {\EventDom^t})$ captures the total order in which events of the same thread executed.
The \emph{reads-from relation}
$G.\rf\subseteq (\bigcup_{x\in \LocationDom}(\locx{\WriteDom} \cup \locx{\RMWDom})\times (\locx{\ReadDom}\cup \locx{\RMWDom}))$ connects every reading event to a writing event that the former obtains its value from.
Hence, we require that
\begin{enumerate*}[label=(\roman*)]
\item the connected events are distinct, i.e., $(\event_w, \event_r)\in \rf \implies (\event_w \neq \event_r)$,
\item the values of the connected events match, i.e., $(\event_w, \event_r)\in \rf \implies (\val_{\wt}(\event_w)=\val_{\rd}(\event_r))$, and 
\item a reading event has exactly one incoming edge, i.e., $(\event_{w_1}, \event_r), (\event_{w_2}, \event_r)\in \rf\implies (\event_{w_1}=\event_{w_2})$  and $\forall (\event_r \in (\ReadDom\cup \RMWDom)), \exists \event_w \in \Events \colon (\event_w, \event_r)\in \rf$.
\end{enumerate*}
Finally, the \emph{modification order} $G.\mo\subseteq \bigcup_{x\in \LocationDom} ((\locx{\WriteDom} \cup \locx{\RMWDom}) \times (\locx{\WriteDom} \cup \locx{\RMWDom}))$ represents a total order on the writing events in each location.

\newcommand\tikzaxX{1}
\newcommand\tikzaxY{1}

\begin{figure}
\begin{subfigure}[b]{.23\textwidth}
\centering
    \begin{tikzpicture}[ line width=1pt, main/.style = {rectangle, inner sep=2pt}] 

    \node[main, fill=white] at (0*\tikzaxX,0*\tikzaxY) (e1){$\event$}; 
    \node[main, fill=white] at (0*\tikzaxX,-0.5*\tikzaxY) (e2){}; 
    \path[hb] (e1) edge [loop right] node {$\hb$} (e1);
    
\end{tikzpicture}
\caption{\customlabel{eq:hbirr}{irr-hb}}
\label{subfig:axioms_ra_hbirr}
\end{subfigure}%
\begin{subfigure}[b]{.23\textwidth}
\centering
\begin{tikzpicture}[ line width=1pt, main/.style = {rectangle, inner sep=2pt}] 

\node[main, fill=white] at (0*\tikzaxX,0*\tikzaxY) (e1){$\wt$}; 
\node[main, fill=white] at (2*\tikzaxX,0*\tikzaxY) (e2){$\wt'$}; 
\draw [mo] (e1) to[bend left=20] node [midway, above] {$\mo$} (e2);
\draw [hb] (e2) to[bend left=20] node [midway, below] {$\hb$} (e1);
    
\end{tikzpicture}
\caption{\customlabel{eq:wc}{write-coherence}}
\label{subfig:axioms_ra_wc}
\end{subfigure}%
\begin{subfigure}[b]{.23\textwidth}
\centering
\begin{tikzpicture}[ line width=1pt, main/.style = {rectangle, inner sep=2pt}] 

\node[main, fill=white] at (0*\tikzaxX,0*\tikzaxY) (e1){$\wt$}; 
\node[main, fill=white] at (2.5*\tikzaxX,0*\tikzaxY) (e2){$\wt'$}; 
\node[main, fill=white] at (1.25*\tikzaxX,0*\tikzaxY) (e3){$\rd$};
\draw [mo] (e1.north east) to[bend left=20] node [midway, above] {$\mo$} (e2.north west);
\draw [hb] (e2) to node [midway, below] {$\hb$} (e3);
\draw [rf] (e1) to node [midway, below] {$\rf$} (e3);

\end{tikzpicture}
\caption{\customlabel{eq:rc}{read-coherence}}
\label{subfig:axioms_ra_rc}
\end{subfigure}
\begin{subfigure}[b]{.23\textwidth}
\centering
\begin{tikzpicture}[ line width=1pt, main/.style = {rectangle, inner sep=2pt}] 

\node[main, fill=white] at (0*\tikzaxX,0*\tikzaxY) (e1){$\wt$}; 
\node[main, fill=white] at (2.5*\tikzaxX,0*\tikzaxY) (e2){$\wt'$}; 
\node[main, fill=white] at (1.25*\tikzaxX,0*\tikzaxY) (e3){$\rmw$};
\draw [mo] (e1.north east) to[bend left=20] node [midway, above] {$\mo$} (e2.north west);
\draw [mo] (e2) to node [midway, below] {$\mo$} (e3);
\draw [rf] (e1) to node [midway, below] {$\rf$} (e3);

\end{tikzpicture}
\caption{\customlabel{eq:at}{atomicity}}
\label{subfig:axioms_ra_at}
\end{subfigure}%
\caption{
The four axioms of Release/Acquire, phrased as forbidden patterns.
 }
\label{fig:axioms_ra}
\end{figure}

\Paragraph{Release/Acquire semantics.} 
We follow the standard (see, e.g.,~\cite{Lahav2022}) axiomatic definition of Release/Acquire over execution graphs $G$, phrased as axioms over the relations of $G$, as well as the derived \emph{happens-before} relation $\hb \triangleq (\po \cup \rf)^{+}$ (often interpreted as ``causality'').
In particular, $G$ is \emph{Release/Acquire-consistent} if it satisfies the following irreflexivity properties (see also \cref{fig:axioms_ra}):
\begin{compactitem}
\item \ref{eq:hbirr} states that $\hb$ is irreflexive, excluding the presence of causality cycles,
\item \ref{eq:wc} states that $\mo;\hb$ is irreflexive, i.e., $\mo$ follows causality,
\item \ref{eq:rc} states that $\mo;\hb;\rf^{-1}$ is irreflexive, i.e., a reading event obtains its value from the $\mo$-maximal writing event among those that are causally before the former, and
\item \ref{eq:at} states that $\mo;\mo;\rf^{-1}$ is irreflexive, i.e., an RMW obtains its value from the $\mo$-maximal writing event among those that are $\mo$-ordered before the former, capturing the fact that the reading and writing part of the RMW occur atomically.
\end{compactitem}

\Paragraph{Release/Acquire program semantics.}
Given an execution graph $G=\tuple{\Events, \po, \rf, \mo}$ and a thread $t_i\in \ThreadDom$, there is a unique topological ordering $\event_1, \event_2,\dots, \event_n$ of $G.\Events^{t_i}$ conforming to $G.\po^{t_i}$.
For such a case, we let $\Word{G}{t_i}=\llab(\event_1), \llab(\event_2),\dots, \llab(\event_n)$.
Given a concurrent program $\ConcProg=\tuple{\Prog_{t_1}, \Prog_{t_2}, \cdots, \Prog_{t_{\numThreads}}}$, the \emph{Release/Acquire semantics} of $\ConcProg$, denoted $\ExecutionGraphsOf{\ConcProg}{\ramm}$, is the set of execution graphs $G$ such that
\begin{enumerate*}[label=(\roman*)]
\item $G$ is Release/Acquire-consistent, and 
\item for all $t_i \in \ThreadDom$, we have that $\Word{G}{t_i}$ is a word of $\Prog_{t_i}$.
\end{enumerate*}

\Paragraph{Release/Acquire reachability.}
Given a concurrent program $\ConcProg=\tuple{\Prog_{t_1}, \Prog_{t_2}, \cdots, \Prog_{t_{\numThreads}}}$, a state $\tuple{p_1, p_2, \dots, p_{\numThreads}}$ of $\ConcProg$, and an execution graph $G\in \ExecutionGraphsOf{\ConcProg}{\ramm}$, we say that $G$ \emph{reaches} $\tuple{p_1, p_2, \dots, p_{\numThreads}}$ if $\Word{G}{t_i}$ reaches $p_i$ in $\Prog_{t_i}$ for all $i\in \IntSet{\numThreads}$.
The reachability problem is to decide, given a concurrent program $\ConcProg$, whether there exists an execution graph $G\in \ExecutionGraphsOf{\ConcProg}{\ramm}$ that reaches the final state of $\ConcProg$, i.e., $\tuple{\Prog_{t_1}.\FinalState, \Prog_{t_2}.\FinalState, \dots, \Prog_{t_{\numThreads}}.\FinalState}$\footnote{The problem of requiring a \emph{single} thread to reach the final state is equally expressive, in the sense that the two problems trivially reduce to each other.}.

\section{Undecidability of RMW-free Reachability}\label{SEC:UNDECRA}

In this section we prove \cref{thm:undecidability}.
We note that the proof is somewhat combinatorially involved.
To benefit exposition, we begin with a proof overview\Revision{, which presents the key idea of our construction (\cref{SUBSEC:UNDECIDABILITY_INTUITION})}.
\Revision{We then describe how our construction realizes this idea} (\cref{SUBSEC:UNDECIDABILITY_OVERVIEW}), followed by the formal reduction (\cref{SUBSEC:UNDECIDABILITY_FORMAL}), a proof of its soundness (\cref{SUBSEC:UNDECIDABILITY_SOUNDNESS}) and its completeness (\cref{SUBSEC:UNDECIDABILITY_COMPLETENESS}).

\Paragraph{Event indexing.}
For notational convenience, given an execution $G$ of a concurrent program, we write $\wt_i(t, x)$ and $\rd_i(t,x)$ to refer to the $i$-th write and read event, respectively, of thread $t$ on location $x$ in $G$.
We call $i$ the \emph{index} of $\wt_i(t, x)$ or $\rd_i(t, x)$, and write $\event_i$ for an event of index $i$ of any operation type.
In our figures, these events are represented as $\wt_i(x,v)$ and $\rd_i(x,v)$, respectively, where $v$ is the value written or read; we omit the thread information since it is visually clear from context.

\Paragraph{Post's Correspondence Problem.}
Post's Correspondence Problem (PCP)~\cite{Post46} takes as input two \Revision{non-empty} tuples $\tuple{\alpha_1, \cdots, \alpha_n}$ and $\tuple{\beta_1, \cdots, \beta_n}$ of non-empty words over some alphabet $\PCPAlphabet$, i.e., $\alpha_i, \beta_i\in \PCPAlphabet^*$.
A solution to PCP is a non-empty sequence of indices $j_1, \cdots, j_k \in \{1, \cdots, n\}^*$ such that $\alpha_{j_1}\cdots \alpha_{j_k}=\beta_{j_1}\cdots \beta_{j_k}$, where $\cdot$ denotes word concatenation.
PCP is well-known to be undecidable.

\subsection{\Revision{Proof Overview}}\label{SUBSEC:UNDECIDABILITY_INTUITION}

\Paragraph{The general construction.}
The undecidability of reachability with RMWs was proven in~\cite{Abdulla19} via a reduction from PCP.
Given an instance $\PCPInstance$ of PCP, the proof constructs a concurrent program $\ConcProg$ such that there exists an execution in $\ExecutionGraphsOf{\ConcProg}{\ramm}$ in which every thread reaches a particular $\term$ state iff $\PCPInstance$ has a solution.
We follow a similar approach here, though not having RMWs at our disposal appears to require a considerably more intricate construction.

\Revision{At a high level, one set of threads guesses the $\alpha$-component of a solution to the PCP instance, and another set guesses its $\beta$-component. 
We call these the \emph{guesser threads}.
Each set sends two sequences: the indices of its guessed solution, and the corresponding letters.
For example, if the $\alpha$-set guesses the solution $j_1, \dots, j_k$, it sends the indices $j_1, \dots, j_k$ and the letters of the word $\alpha_{j_1}\cdots \alpha_{j_k}$.
One thread receives the two index sequences and checks that they are equal, and another receives the two words and checks that they are equal. 
We call these the \emph{checker threads}.}

\Revision{Our thread names follow these roles: we denote guesser threads by $g$ and checker threads by $c$. 
The superscript indicates whether a thread sends or receives letters ($x$) or indices ($y$), and the subscript of a guesser thread indicates the component ($\alpha$ or $\beta$) it guesses. 
For example, $\tax$ sends the letters of the guessed $\alpha$-word, and $\tx$ receives letters.
In more detail, the threads behave as follows (see \cref{fig:pcp_example} for an illustration).}

\begin{compactitem}
\item \Revision{Threads $\tax$ and $\tay$ independently and non-deterministically guess a solution to PCP, one index at a time.
We use an auxiliary location $\vzaxyp$ to ensure that they guess the same solution $j_1, \dots, j_k$. 
Thread $\tay$ writes the indices $j_1, \dots, j_k$ to $\vya$, one by one, while $\tax$ writes the corresponding word $\alpha_{j_1}\cdots \alpha_{j_k}$ to $\vxa$, letter by letter.}

\item \Revision{Threads $\tbx$ and $\tby$ work the same as $\tax$ and $\tay$, with each instance of $\alpha$ replaced by $\beta$.
We use an auxiliary location $\vzbxyp$ to ensure that they guess the same solution $j'_1, \dots, j'_{k'}$.
Thread $\tby$ writes the indices $j'_1, \dots, j'_{k'}$ to $\vyb$, one by one, while $\tbx$ writes the corresponding word $\beta_{j'_1}\cdots \beta_{j'_{k'}}$ to $\vxb$, letter by letter.}

\item \Revision{Thread $\tx$ reads the letter sequences written to $\vxa$ and $\vxb$, one letter at a time, and checks that they are equal, i.e., that $\alpha_{j_1}\cdots \alpha_{j_k}=\beta_{j'_1}\cdots \beta_{j'_{k'}}$.}

\item \Revision{Thread $\ty$ reads the index sequences written to $\vya$ and $\vyb$, one index at a time, and checks that they are equal, i.e., that $j_1,\dots,j_k = j'_1,\dots ,j'_{k'}$.}
\end{compactitem}

\newcommand\tikzeX{0.5}
\newcommand\tikzeY{1}

\begin{figure}
\centering
\begin{tikzpicture}[ line width=1pt, main/.style = {inner sep=1pt}]

\draw [->] (-2*\tikzeX,1.5*\tikzeY) -- (-2*\tikzeX,-4.5*\tikzeY) node [near start, above=1.5cm]{$\tay$}; 

\draw [->] (-5.5*\tikzeX,1.5*\tikzeY) -- (-5.5*\tikzeX,-4.5*\tikzeY) node [near start, above=1.5cm]{$\tax$}; 

\draw [->] (1.5*\tikzeX,1.5*\tikzeY) -- (1.5*\tikzeX,-4.5*\tikzeY) node [near start, above=1.5cm]{$\tx$}; 

\draw [->] (4.5*\tikzeX,1.5*\tikzeY) -- (4.5*\tikzeX,-4.5*\tikzeY) node [near start, above=1.5cm]{$\ty$}; 

\draw [->] (8*\tikzeX,1.5*\tikzeY) -- (8*\tikzeX,-4.5*\tikzeY) node [near start, above=1.5cm]{$\tbx$}; 

\draw [->] (11.5*\tikzeX,1.5*\tikzeY) -- (11.5*\tikzeX,-4.5*\tikzeY) node [near start, above=1.5cm]{$\tby$};

\node[main, fill=white] at (-2*\tikzeX,1*\tikzeY) (w1i1){{\footnotesize$\wt_1(\vya, 2)$}}; 

\node[main, fill=white] at (-2*\tikzeX,0*\tikzeY) (w1al){{\footnotesize$\wt_1(\vzaxyp, 2)$}}; 

\node[main, fill=white] at (-5.5*\tikzeX,0.5*\tikzeY) (w1c1){{\footnotesize$\wt_1(\vxa, a)$}}; 

\node[main, fill=white] at (-5.5*\tikzeX,-0.5*\tikzeY) (w1c2){{\footnotesize$\wt_2(\vxa, a)$}}; 

\node[main, fill=white] at (-5.5*\tikzeX,-1*\tikzeY) (r1al){{\footnotesize$\rd_1(\vzaxyp, 2)$}}; 

\node[main, fill=white] at (-2*\tikzeX,-1.5*\tikzeY) (w1i2){{\footnotesize$\wt_2(\vya, 1)$}}; 

\node[main, fill=white] at (-2*\tikzeX,-2.5*\tikzeY) (w2al){{\footnotesize$\wt_2(\vzaxyp, 1)$}}; 

\node[main, fill=white] at (-5.5*\tikzeX,-2*\tikzeY) (w1c3){{\footnotesize$\wt_3(\vxa, b)$}}; 

\node[main, fill=white] at (-5.5*\tikzeX,-3*\tikzeY) (w1c4){{\footnotesize$\wt_4(\vxa, a)$}}; 

\node[main, fill=white] at (-5.5*\tikzeX,-3.5*\tikzeY) (r2al){{\footnotesize$\rd_2(\vzaxyp, 1)$}};

\node[main, fill=white]  at (1.5*\tikzeX,0.5*\tikzeY) (r211){{\footnotesize$\rd_1(\vxa,a)$}}; 

\node[main, fill=white] at (1.5*\tikzeX,0*\tikzeY) (r212){{\footnotesize$\rd_1(\vxb,a)$}}; 

\node[main, fill=white] at (1.5*\tikzeX,-0.5*\tikzeY) (r221){{\footnotesize$\rd_2(\vxa,a)$}}; 

\node[main, fill=white] at (1.5*\tikzeX,-1.5*\tikzeY) (r222){{\footnotesize$\rd_2(\vxb,a)$}}; 

\node[main, fill=white] at (1.5*\tikzeX,-2*\tikzeY)(r231){{\footnotesize$\rd_3(\vxa,b)$}}; 

\node[main, fill=white] at (1.5*\tikzeX,-2.5*\tikzeY)(r232){{\footnotesize$\rd_3(\vxb,b)$}}; 

\node[main, fill=white] at (1.5*\tikzeX,-3*\tikzeY)(r241){{\footnotesize$\rd_4(\vxa,a)$}}; 

\node[main, fill=white]  at (1.5*\tikzeX,-3.5*\tikzeY) (r242){{\footnotesize$\rd_4(\vxb,a)$}};

\node[main, fill=white] (r311) at (4.5*\tikzeX,1*\tikzeY){{\footnotesize$\rd_1(\vya,2)$}}; 

\node[main, fill=white]  at (4.5*\tikzeX,0.5*\tikzeY) (r312){{\footnotesize$\rd_1(\vyb,2)$}}; 

\node[main, fill=white] at (4.5*\tikzeX,-0.5*\tikzeY) (r321){{\footnotesize$\rd_2(\vya,1)$}}; 

\node[main, fill=white] at (4.5*\tikzeX,-1*\tikzeY) (r322){{\footnotesize$\rd_2(\vyb,1)$}};

\node[main, fill=white] at (11.5*\tikzeX,0.5*\tikzeY) (w4i1){{\footnotesize$\wt_1(\vyb,2)$}}; 

\node[main, fill=white] at (11.5*\tikzeX,0*\tikzeY) (w1bl){{\footnotesize$\wt_1(\vzbxyp,2)$}}; 

\node[main, fill=white]  at (8*\tikzeX,0*\tikzeY)(w4c1){{\footnotesize$\wt_1(\vxb,a)$}}; 

\node[main, fill=white] at (8*\tikzeX,-0.5*\tikzeY) (r1bl){{\footnotesize$\rd_1(\vzbxyp,2)$}}; 

\node[main, fill=white]  at (11.5*\tikzeX,-1*\tikzeY) (w4i2){{\footnotesize$\wt_2(\vyb,1)$}};

\node[main, fill=white]  at (11.5*\tikzeX,-1.5*\tikzeY) (w2bl){{\footnotesize$\wt_2(\vzbxyp,1)$}};

\node[main, fill=white]  at (8*\tikzeX,-1.5*\tikzeY) (w4c2){{\footnotesize$\wt_2(\vxb,a)$}};  

\node[main, fill=white] at (8*\tikzeX,-2.5*\tikzeY) (w4c3){{\footnotesize$\wt_3(\vxb,b)$}}; 

\node[main, fill=white] at (8*\tikzeX,-3.5*\tikzeY) (w4c4){{\footnotesize$\wt_4(\vxb,a)$}}; 

\node[main, fill=white] at (8*\tikzeX,-4*\tikzeY) (r2bl){{\footnotesize$\rd_2(\vzbxyp,1)$}}; 

\begin{pgfonlayer}{foreground layer}

    \draw [rf] (w1i1) -- (r311);
    \draw [rf] (w1i2.east) -- (r321);

    \draw [rf] (w1c1) -- (r211);
    \draw [rf] (w1c2) -- (r221);
    \draw [rf] (w1c3) -- (r231);
    \draw [rf] (w1c4) -- (r241);

    \draw [rf] (w4i1) -- (r312);
    \draw [rf] (w4i2) -- (r322);

    \draw [rf] (w4c1) -- (r212);
    \draw [rf] (w4c2) -- (r222);
    \draw [rf] (w4c3) -- (r232);
    \draw [rf] (w4c4) -- (r242);

    \draw [rf] (w1al.south west) -- (r1al.north east); 
    \draw [rf] (w2al.south west) -- (r2al.north east); 

    \draw [rf] (w1bl) -- (r1bl);
    \draw [rf] (w2bl.south west) -- (r2bl.north east);

\end{pgfonlayer}

\fill[ Blue, opacity = 0.15] (0.1*\tikzeX,0.7*\tikzeY) rectangle (2.9*\tikzeX,-0.2*\tikzeY);

\fill[ Blue, opacity = 0.15] (0.1*\tikzeX,-0.3*\tikzeY) rectangle (2.9*\tikzeX,-1.7*\tikzeY);

\fill[ Blue, opacity = 0.15] (0.1*\tikzeX,-1.8*\tikzeY) rectangle (2.9*\tikzeX,-2.7*\tikzeY);

\fill[ Blue, opacity = 0.15] (0.1*\tikzeX,-2.8*\tikzeY) rectangle (2.9*\tikzeX,-3.7*\tikzeY);

\fill[ RubineRed, opacity = 0.15] (3.1*\tikzeX,1.2*\tikzeY) rectangle (5.9*\tikzeX,0.3*\tikzeY);

\fill[ RubineRed, opacity = 0.15] (3.1*\tikzeX,-0.3*\tikzeY) rectangle (5.9*\tikzeX,-1.2*\tikzeY);

\draw [decorate,decoration={brace,amplitude=4},thin, white] (-7*\tikzeX, -1.1*\tikzeY) -- (-7*\tikzeX, 1.3*\tikzeY) node (bx1) [black,midway] {};
        \node[left = 0cm of bx1,align=right] {writes \\
        $\alpha_2 = aa$};

\fill[ black, opacity = 0.08] (r1al.west |- w1i1.north) rectangle (w1al.east |- r1al.south); 

\draw [decorate,decoration={brace,amplitude=4},thin, white] (-7*\tikzeX, -3.8*\tikzeY) -- (-7*\tikzeX, -1.3*\tikzeY) node (bx1) [black,midway] {};
        \node[left = 0cm of bx1,align=right] {writes\\
        $\alpha_1 = ba$};

\fill[ black, opacity = 0.08] (r2al.west |- w1i2.north) rectangle (w2al.east |- r2al.south); 

\draw [decorate,decoration={brace,amplitude=4},thin, white] (13*\tikzeX, 0.7*\tikzeY) -- (13*\tikzeX, -0.6*\tikzeY) node (bx1) [black,midway] {};
        \node[right = 0cm of bx1,align=left] {writes\\
        $\beta_2 = a$};

\fill[ black, opacity = 0.08] (w4i1.north east) rectangle (r1bl.south west);

\draw [decorate,decoration={brace,amplitude=4},thin, white] (13*\tikzeX, -0.8*\tikzeY) -- (13*\tikzeX, -4.2*\tikzeY) node (bx1) [black,midway,align=left] {};
        \node[right = 0cm of bx1,align=left] {writes\\
        $\beta_1 = aba$};

\fill[ black, opacity = 0.08] (w4i2.north east) rectangle (r2bl.south west);

\end{tikzpicture}
\caption{
The high level behavior of $\tax$, $\tay$, $\tbx$, $\tby$, $\tx$, and $\ty$ on a PCP instance with the tuples $\tuple{ba,aa}$ and $\tuple{aba,a}$ as input.
The solution illustrated is the sequence $2,1$, which generates the string $aaba$. 
Words guessed by $\tax$ and $\tay$ are shown on the left; those guessed by $\tbx$ and $\tby$, on the right. 
The highlighted pairs in $\tx$ and $\ty$ mark operations that must read the same value.
}
\label{fig:pcp_example}
\end{figure}

\Revision{All threads can reach their final instruction $\term$ if and only if both checks succeed, i.e., if and only if the guessed sequence is a solution to $\PCPInstance$. }

\Paragraph{Understanding the challenge.}
\Revision{Our construction requires the guesser threads to send sequences of values to the checker threads one-by-one. 
At the same time, the checker threads must be able to lag arbitrarily far behind the guesser threads:~since $|\alpha_j|$ and $|\beta_j|$ may differ, the $\alpha$-threads and the $\beta$-threads progress through the solution at different rates, so the checker threads cannot stay synchronized with both groups of guesser threads.}

\Revision{Conceptually, the above setting requires that guesser threads communicate to checker threads via \emph{perfect, asynchronous, one-way channels}. 
In general, an \emph{asynchronous} channel allows a sender $\tsnd$ to send a sequences of values (messages) to a receiver $\trcv$, where $\trcv$ may lag arbitrarily behind $\tsnd$, i.e., $\trcv$ may receive a message long after $\tsnd$ deposited it on the channel. 
Furthermore, the channel being \emph{perfect} means that every message sent by $\tsnd$ is received exactly once by $\trcv$, i.e., no message is received twice and no message is skipped over.
In \cref{fig:pcp_example}, for example, nothing prevents $\rd_2(\tx, \vxa)$ from skipping $\wt_2(\tax, \vxa)$ and reading from $\wt_3(\tax, \vxa)$ instead. 
Without no-skipping, we have a \emph{lossy} channel rather than a perfect one.
Since lossy channel systems admit decidable reachability~\cite{Abdulla1993}, however, this would be insufficient for our undecidability proof.}

\Revision{Common channel implementations achieve no-skipping via handshaking:~the receiving end acknowledges each received message, at which point the channel makes the next message available for reading.
However, over shared memory, an acknowledgment in the form a value communicated from the receiving end to the sending end would break asynchrony. 
\emph{How can we simulate a perfect, asynchronous channel over Release/Acquire shared memory, using only reads and writes?}}

\Paragraph{Simulating a perfect channel.}
\Revision{To build intuition, we illustrate here a mechanism for simulating a perfect asynchronous channel under Release/Acquire. 
Later, we employ this mechanism in our undecidability proof.
The mechanism is illustrated in \cref{fig:channel_example}, whereby a sender thread $\tsnd$ communicates to a receiver thread $\trcv$ over location $\vch$, effectively achieving channel semantics.}

\Revision{Before we analyze the figure in detail, we start with a couple of high-level observations that serve as guiding principles behind our construction.
First, the semantics of Release/Acquire are axioms over $\po$, $\rf$, and $\mo$.
As we have already argued, for each received value, the acknowledgment from the receiver to the sender cannot be in the form of an $\rf$-edge, as this would break the asynchrony requirement.
Thus, such an acknowledgment can only be in the form of an $\mo$-edge. 
Moreover, the mere requirement of \ref{eq:wc} that $\mo;\hb$ is irreflexive is promising for guaranteeing asynchrony:~in the bigger picture, we must use several such channels, which will lead to interleaving their corresponding acknowledgment edges.
This will inevitably lead to $(\hb\cup \mo)$-cycles; however, as long as each channel is simulated over distinct memory locations $\vch$, each $\mo_{\vch};\hb$ will still be irreflexive}\footnote{\Revision{For comparison, Strong Release/Acquire requires the irreflexivity of $(\hb\cup \mo)^+$, and thus our mechanism for simulating channels does not work --- this is expected, as Strong Release/Acquire has decidable reachability.}}.
\Revision{Second, as stated in \cref{thm:bounded_context_switches} (and proven in \cref{SEC:DECBCS}), reachability under bounded context switches is decidable. 
This implies that any mechanism for simulating a channel has to involve more threads than merely $\tsnd$ and $\trcv$, so that the additional threads can introduce the required context-switching.}

\Revision{We now discuss the mechanism of \cref{fig:channel_example}.
First, observe that, besides the two main threads $\tsnd$ and $\trcv$, we also have two additional threads, $\tsndp$ and $\trcvp$.
The two thread pairs communicate over $\vch$ and $\vchp$, respectively.
Moreover, the $\tsnd$ communicates with $\tsndp$ over $\vchz$ and $\vchzp$, while $\trcv$ communicates with $\trcvp$ over $\vch$ and $\vchp$.
We further make three assumptions:}
\begin{compactenum}[label=(A\arabic*)]
\item\label{item:ch_one_thread} \Revision{Although $\vch$ and $\vchp$ is each written by two threads,
each read event can only read from one of the writing threads.
For example, all reads $\rd_i(\trcv,\vch)$ can only read from $\tsnd$.
This can be simply enforced by requiring the values read/written to also specify the writing thread.}
\item\label{item:ch_first_read} \Revision{$\rd_1(\trcv, \vch)$ always reads from $\wt_1(\tsnd, \vch)$.
This can be simply enforced by marking the value written by $\wt_1(\tsnd, \vch)$ and read by $\rd_1(\trcv, \vch)$ in a unique way.}
\item\label{item:ch_sync}
\Revision{There is a mechanism establishing the $\rf$-pattern between $\tsndp$ and $\tsnd$ (i.e., $\rf$-edges such as \squaredsmall{A} and \squaredsmall{C}), as well as the $\rf$-pattern between $\trcv$ and $\trcvp$ (i.e., $\rf$-edges such as \squaredsmall{B} and \squaredsmall{D}).}
\end{compactenum}

\Revision{We now show that, under these assumptions, $\trcv$ and $\trcvp$ cannot skip any write of $\tsnd$ and $\tsndp$, respectively, on locations $\vch$ and $\vchp$.
Perhaps surprisingly, simulating a single perfect channel appears to require simulating two channels simultaneously (on both $\vch$ and $\vchp$).}

\newcommand\tikzchX{0.7}
\newcommand\tikzchY{0.9}
\newcommand\tikzchS{-2.95}
\newcommand\tikzchE{-3.6}
\newcommand\chcut[3]{%
  \coordinate (cutS) at (intersection cs: first line={(#2)--(#3)}, second line={(-5,\tikzchS*\tikzchY)--(15,\tikzchS*\tikzchY)});
  \coordinate (cutE) at (intersection cs: first line={(#2)--(#3)}, second line={(-5,\tikzchE*\tikzchY)--(15,\tikzchE*\tikzchY)});
  \draw [#1, -] (#2) -- (cutS);
  \draw [#1, -, path fading=south] (cutS) -- (cutE);
}

\begin{figure}
\centering
\begin{tikzpicture}[ line width=1pt, main/.style = {draw,inner sep=1pt}, font=\normalsize]

\draw [->] (-2*\tikzchX,2*\tikzchY) -- (-2*\tikzchX,-4.0*\tikzchY);
\node[] (aa) at (-2*\tikzchX,2.4*\tikzchY){$\tsndp$};

\draw [->] (2*\tikzchX,2*\tikzchY) -- (2*\tikzchX,-4.0*\tikzchY);
\node[] (aa) at (2*\tikzchX,2.4*\tikzchY){$\tsnd$};

\draw [->] (6*\tikzchX,2*\tikzchY) -- (6*\tikzchX,-4.0*\tikzchY);
\node[] (aa) at (6*\tikzchX,2.4*\tikzchY){$\trcv$};

\draw [->] (10*\tikzchX,2*\tikzchY) -- (10*\tikzchX,-4.0*\tikzchY);
\node[] (aa) at (10*\tikzchX,2.4*\tikzchY){$\trcvp$};

\node[main, fill=white, draw=white, ] (w1g) at (-2*\tikzchX,0*\tikzchY){\footnotesize$\wt_1(\vchp)$};
\node[main, fill=white, draw=white, ] (w1y) at (-2*\tikzchX,-.7*\tikzchY){\footnotesize$\wt_1(\vchz)$};
\node[main, fill=white, draw=white, ] (r1p) at (-2*\tikzchX,-1.3*\tikzchY){\footnotesize$\rd_1(\vchzp)$};
\node[main, fill=white, draw=white, ] (w2g) at (-2*\tikzchX,-2*\tikzchY){\footnotesize$\wt_2(\vchp)$};
\node[main, fill=white, draw=white, ] (w2y) at (-2*\tikzchX,-2.7*\tikzchY){\footnotesize$\wt_2(\vchz)$};
\node[main, fill=white, draw=white, inner sep=3pt] (d1) at (-2*\tikzchX,-3.5*\tikzchY){$\dots$};

\node[main, fill=white, draw=white, ] (w1b) at (2*\tikzchX,1*\tikzchY){\footnotesize$\wt_1(\vch)$};
\node[main, fill=white, draw=white, ] (w1p) at (2*\tikzchX,0.3*\tikzchY){\footnotesize$\wt_1(\vchzp)$};
\node[main, fill=white, draw=white, ] (w2b) at (2*\tikzchX,-1*\tikzchY){\footnotesize$\wt_2(\vch)$};
\node[main, fill=white, draw=white, ] (w2p) at (2*\tikzchX,-1.7*\tikzchY){\footnotesize$\wt_2(\vchzp)$};
\node[main, fill=white, draw=white, ] (r1y) at (2*\tikzchX,-2.3*\tikzchY){\footnotesize$\rd_1(\vchz)$};
\node[main, fill=white, draw=white, ] (w3b) at (2*\tikzchX,-3*\tikzchY){\footnotesize$\wt_3(\vch)$};
\node[main, fill=white, draw=white, inner sep=3pt] (d2) at (2*\tikzchX,-3.5*\tikzchY){$\dots$};

\node[main, fill=white, draw=white, ] (w1ba) at (6*\tikzchX,1.7*\tikzchY){\footnotesize$\wt_1(\vch)$};
\node[main, fill=white, draw=white, ] (r1b) at (6*\tikzchX,1*\tikzchY){\footnotesize$\rd_1(\vch)$};
\node[main, fill=white, draw=white, ] (w2ba) at (6*\tikzchX,-0.3*\tikzchY){\footnotesize$\wt_2(\vch)$};
\node[main, fill=white, draw=white, ] (r2b) at (6*\tikzchX,-1*\tikzchY){\footnotesize$\rd_2(\vch)$};
\node[main, fill=white, draw=white, ] (r1ga) at (6*\tikzchX,-1.7*\tikzchY){\footnotesize$\rd_1(\vchp)$};
\node[main, fill=white, draw=white, ] (w3ba) at (6*\tikzchX,-2.3*\tikzchY){\footnotesize$\wt_3(\vch)$};
\node[main, fill=white, draw=white, ] (r3b) at (6*\tikzchX,-3*\tikzchY){\footnotesize$\rd_3(\vch)$};
\node[main, fill=white, draw=white, inner sep=3pt] (d3) at (6*\tikzchX,-3.5*\tikzchY){$\dots$};

\node[main, fill=white, draw=white, ] (w1ga) at (10*\tikzchX,0.7*\tikzchY){\footnotesize$\wt_1(\vchp)$};
\node[main, fill=white, draw=white, ] (r1g) at (10*\tikzchX,0*\tikzchY){\footnotesize$\rd_1(\vchp)$};
\node[main, fill=white, draw=white, ] (r1ba) at (10*\tikzchX,-.7*\tikzchY){\footnotesize$\rd_1(\vch)$};
\node[main, fill=white, draw=white, ] (w2ga) at (10*\tikzchX,-1.3*\tikzchY){\footnotesize$\wt_2(\vchp)$};
\node[main, fill=white, draw=white, ] (r2g) at (10*\tikzchX,-2*\tikzchY){\footnotesize$\rd_2(\vchp)$};
\node[main, fill=white, draw=white, ] (r2ba) at (10*\tikzchX,-2.7*\tikzchY){\footnotesize$\rd_2(\vch)$};
\node[main, fill=white, draw=white, inner sep=3pt] (d4) at (10*\tikzchX,-3.5*\tikzchY){$\dots$};

\node[main, draw=none, text opacity=0] (r2p) at (-2*\tikzchX,-3.3*\tikzchY){\footnotesize$\rd_2(\vchzp)$};
\node[main, draw=none, text opacity=0] (r2y) at (2*\tikzchX,-4.3*\tikzchY){\footnotesize$\rd_2(\vchz)$};
\node[main, draw=none, text opacity=0] (r2ga) at (6*\tikzchX,-3.7*\tikzchY){\footnotesize$\rd_2(\vchp)$};
\node[main, draw=none, text opacity=0] (r3ba) at (10*\tikzchX,-4.7*\tikzchY){\footnotesize$\rd_3(\vch)$};

\begin{pgfonlayer}{background layer}

\draw [mo] (w1ba.180) to node[]{\circledsmall{2}} (w1b.10);
\draw [mo] (w2ba.180) to[bend right=0] (w2b.10);
\draw [mo] (w3ba.180) to[bend right=0] (w3b.10);

\draw [mo] (w1ga.180) to node[]{\circledsmall{4}} (w1g.5);
\draw [mo] (w2ga.180) to[bend right=0] (w2g.5);

\draw [rf] (w1p.west) -- (r1p.east) node [midway]{\squaredsmall{A}};
\draw [rf] (w1y.east) -- (r1y.west) node [midway]{\squaredsmall{C}};

\draw [rf] (w1b.0) to node[]{\circledsmall{1}} (r1b.180);
\draw [rf] (w2b.0) to node[]{\circledsmall{5}} (r2b.180);
\draw [rf] (w3b.0) -- (r3b.180);

\draw [rf] (w1ba.east) -- (r1ba.west) node [near start]{\squaredsmall{B}};
\draw [rf] (w2ba.east) -- (r2ba.west);

\draw [rf] (w1ga.west) -- (r1ga.east) node [midway]{\squaredsmall{D}};

\draw [rf] (w1g.0) to node[]{\circledsmall{3}} (r1g.180);
\draw [rf] (w2g.0) -- (r2g.180);

\end{pgfonlayer}

\chcut{rf}{w2p.west}{r2p.east}
\chcut{rf}{w2y.east}{r2y.west}
\chcut{rf}{w2ga.west}{r2ga.east}
\chcut{rf}{w3ba.east}{r3ba.west}

\end{tikzpicture}
\caption{
\Revision{Simulating a perfect, asynchronous one-way channel from $\tsnd$ to $\trcv$ on location $\vch$, using the auxiliary threads $\tsndp$ and $\trcvp$ and the locations $\vchp$, $\vchz$, and $\vchzp$.}
\Revision{Values are omitted.}
}
\label{fig:channel_example}
\end{figure}

\Revision{By  \ref{item:ch_first_read}, we have $\wt_1(\tsnd, \vch) \LTo{\rf} \rd_1(\trcv, \vch)$ (edge \circledsmall{1}).
Since also $\wt_1(\trcv, \vch)\LTo{\po} \rd_1(\trcv, \vch)$, \ref{eq:rc} forces the acknowledgment $\wt_1(\trcv, \vch) \LTo{\mo} \wt_1(\tsnd, \vch)$ (edge \circledsmall{2}).
Next, edge \squaredsmall{A}, given by  \ref{item:ch_sync}, places $\wt_1(\tsnd, \vch)$ $\hb$-before every write of $\tsndp$ on $\vchp$ except the first.
By  \ref{item:ch_one_thread}, $\rd_1(\trcvp, \vchp)$ reads from a write of $\tsndp$; if it read from any write but the first (for example, $\wt_2(\tsndp, \vchp)$), then $\rd_1(\trcvp, \vch)$, which succeeds $\rd_1(\trcvp, \vchp)$ in $\po$, would satisfy $\wt_1(\trcv, \vch) \LTo{\mo} \wt_1(\tsnd, \vch) \LTo{\hb} \rd_1(\trcvp, \vch)$ and $\wt_1(\trcv, \vch) \LTo{\rf} \rd_1(\trcvp, \vch)$ 
(edge \squaredsmall{B}, by  \ref{item:ch_sync}), violating \ref{eq:rc}.
Hence $\wt_1(\tsndp, \vchp) \LTo{\rf} \rd_1(\trcvp, \vchp)$ (edge \circledsmall{3}), which, as before, forces the acknowledgment $\wt_1(\trcvp, \vchp) \LTo{\mo} \wt_1(\tsndp, \vchp)$ (edge \circledsmall{4}).}

\Revision{Symmetrically, edges \circledsmall{4}, \squaredsmall{C}, and \squaredsmall{D} prevent $\rd_2(\trcv, \vch)$ from reading any write of $\tsnd$ after $\wt_2(\tsnd, \vch)$, and thus $\wt_2(\tsnd, \vch) \LTo{\rf} \rd_2(\trcv, \vch)$ (edge \circledsmall{5}).
Repeating this argument inductively yields the no-skipping invariant}
\[
\forall i \quad \wt_i(\tsnd, \vch) \LTo{\rf} \rd_i(\trcv, \vch) \quad \text{and} \quad \wt_i(\tsndp, \vchp) \LTo{\rf} \rd_i(\trcvp, \vchp).
\]
\Revision{This mechanism for simulating a perfect, asynchronous channel is central to our proof of undecidability, but the construction requires tackling a few more challenges.
In particular, as we have multiple such channels used simultaneously, this will result in additional $\hb$-paths between the events simulating the channel.
Thus, besides \ref{eq:hbirr}, we have to ensure that, for each channel location $\vch$, the relations $\mo_{\vch};\hb_{\vch}$ (\ref{eq:wc}) and $\mo_{\vch};\hb_{\vch};\rf^{-1}_{\vch}$ remain irreflexive, in the presence of the additional $\hb$-paths.}

\subsection{\Revision{Realizing the Construction}}\label{SUBSEC:UNDECIDABILITY_OVERVIEW} 
\edef\savedpredisplaypenalty{\the\predisplaypenalty}\predisplaypenalty=0

\Revision{We now describe in more detail how our reduction realizes the construction outlined in \cref{SUBSEC:UNDECIDABILITY_INTUITION}, building upon the pattern of \cref{fig:channel_example}.}
\Revision{Recall that $\tax$ and $\tbx$ send their words to $\tx$ via $\vxa$ and $\vxb$, and $\tay$ and $\tby$ send their index sequences to $\ty$ via $\vya$ and $\vyb$.}
\Revision{Following \cref{fig:channel_example}, each of these locations has a primed counterpart ($\vxap$, $\vyap$, $\vxbp$, $\vybp$), each guesser thread has an auxiliary guesser thread ($\taxp$, $\tayp$, $\tbxp$, $\tbyp$), and each checker thread has an auxiliary checker thread ($\txp$, $\typ$).}
\Revision{We first show how to enforce \ref{item:ch_one_thread} -- \ref{item:ch_sync}, and then conclude that the pattern of \cref{fig:channel_example} yields no-skipping in our setting.}
\cref{fig:pcp_top} illustrates the full communication topology of our construction.

\newcommand\tikztgX{2.5}
\newcommand\tikztgY{1.5}

\begin{figure}
\centering
\scalebox{0.9}{
\begin{tikzpicture}[ line width=1pt, main/.style = {draw,}] 

            \draw[draw=none, fill=blue, opacity=0.1, thick, rounded corners] (-0.8, 2*\tikztgY-0.8) rectangle (3*\tikztgX+0.8, 2*\tikztgY+0.8);
            \node[blue] at (3*\tikztgX/2, 2*\tikztgY+1.2) {\cref{FIG:RAEX_AXY}};
            
            \draw[draw=none, fill=red, opacity=0.1, thick, rounded corners] (-0.8, 2*\tikztgY+0.8) rectangle (\tikztgX+0.8, 1*\tikztgY-0.8);
            \node[red!50!black] at (\tikztgX*-0.55, 1.45*\tikztgY) {\cref{FIG:RAEX}};

            \node[main, fill=white, draw=black, text=black, minimum size=0.8cm] (tayp) at (\tikztgX*3,2*\tikztgY){\large$\tayp$};
            \node[main, fill=white, draw=black, text=black, minimum size=0.8cm] (tay) at (\tikztgX*2,2*\tikztgY){\large$\tay$};
            \node[main, fill=white, draw=black, text=black, minimum size=0.8cm] (tax) at (\tikztgX*1,2*\tikztgY){\large$\tax$};
            \node[main, fill=white, draw=black, text=black, minimum size=0.8cm] (taxp) at (\tikztgX*0,2*\tikztgY){\large$\taxp$};

            \node[main, fill=white, draw=black, text=black, minimum size=0.8cm] (typ) at (\tikztgX*3,1*\tikztgY){\large$\typ$};
            \node[main, fill=white, draw=black, text=black, minimum size=0.8cm] (ty) at (\tikztgX*2,1*\tikztgY){\large$\ty$};
            \node[main, fill=white, draw=black, text=black, minimum size=0.8cm] (tx) at (\tikztgX*1,1*\tikztgY){\large$\tx$};
            \node[main, fill=white, draw=black, text=black, minimum size=0.8cm] (txp) at (\tikztgX*0,1*\tikztgY){\large$\txp$};

            \node[main, fill=white, draw=black, text=black, minimum size=0.8cm] (tbyp) at (\tikztgX*3,0*\tikztgY){\large$\tbyp$};
            \node[main, fill=white, draw=black, text=black, minimum size=0.8cm] (tby) at (\tikztgX*2,0*\tikztgY){\large$\tby$};
            \node[main, fill=white, draw=black, text=black, minimum size=0.8cm] (tbx) at (\tikztgX*1,0*\tikztgY){\large$\tbx$};
            \node[main, fill=white, draw=black, text=black, minimum size=0.8cm] (tbxp) at (\tikztgX*0,0*\tikztgY){\large$\tbxp$};


            \draw [dashed, ->, color=\zcolor] (tay.15) to node [above=0pt]{\small$\vzay$} (tayp.165);
            \draw [dashed, ->, color=\zpcolor] (tayp.-165) to node [below=0pt]{\small$\vzayp$} (tay.-15);

            \draw [dashed, ->, color=\zcolor] (tax.15) to node [above=0pt]{\small$\vzaxy$} (tay.165);
            \draw [dashed, ->, color=\lcolor] (tay.-165) to node [below=0pt]{\small$\vzaxyp$} (tax.-15);

            \draw [dashed, ->, color=\zcolor] (tax.165) to node [above=0pt]{\small$\vzax$} (taxp.15);
            \draw [dashed, ->, color=\zpcolor] (taxp.-15) to node [below=0pt]{\small$\vzaxp$} (tax.-165);

            \draw [dashed, ->, color=\ycolor] (ty.15) to node [above=0pt]{\small$\vya , \vyb$} (typ.165);
            \draw [dashed, ->, color=\ypcolor] (typ.-165) to node [below=0pt]{\small$\vyap, \vybp$} (ty.-15);

            \draw [dashed, ->, color=\xcolor] (tx.165) to node [above=0pt]{\small$\vxa, \vxb$} (txp.15);
            \draw [dashed, ->, color=\xpcolor] (txp.-15) to node [below=0pt]{\small$\vxap, \vxbp$} (tx.-165);

            \draw [dashed, ->, color=\zcolor] (tby.15) to node [above=0pt]{\small$\vzby$} (tbyp.165);
            \draw [dashed, ->, color=\zpcolor] (tbyp.-165) to node [below=0pt]{\small$\vzbyp$} (tby.-15);

            \draw [dashed, ->, color=\zcolor] (tbx.15) to node [above=0pt]{\small$\vzbxy$} (tby.165);
            \draw [dashed, ->, color=\lcolor] (tby.-165) to node [below=0pt]{\small$\vzbxyp$} (tbx.-15);

            \draw [dashed, ->, color=\zcolor] (tbx.165) to node [above=0pt]{\small$\vzbx$} (tbxp.15);
            \draw [dashed, ->, color=\zpcolor] (tbxp.-15) to node [below=0pt]{\small$\vzbxp$} (tbx.-165);


            \draw [double, ->, color=\ycolor] (tay) to node [left=0pt]{\small$\vya$} (ty);
            \draw [->, color=\ypcolor] (tayp) to node [left=0pt]{\small$\vyap$} (typ);
            \draw [double, ->, color=\xcolor] (tax) to node [right=0pt]{\small$\vxa$} (tx);
            \draw [->, color=\xpcolor] (taxp) to node [right=0pt]{\small$\vxap$} (txp);

            \draw [double, ->, color=\ycolor] (tby) to node [left=0pt]{\small$\vyb$} (ty);
            \draw [->, color=\ypcolor] (tbyp) to node [left=0pt]{\small$\vybp$} (typ);
            \draw [double, ->, color=\xcolor] (tbx) to node [right=0pt]{\small$\vxb$} (tx);
            \draw [->, color=\xpcolor] (tbxp) to node [right=0pt]{\small$\vxbp$} (txp);

            \draw [decorate,decoration={brace,amplitude=4},thin] (\tikztgX*3.4,2.3*\tikztgY) -- (\tikztgX*3.4,1.7*\tikztgY) node (d1) [black,midway] {};
            \node[right = 0.1cm of d1] {Guesser threads};
            
            \draw [decorate,decoration={brace,amplitude=4},thin] (\tikztgX*3.4,1.3*\tikztgY) -- (\tikztgX*3.4,0.7*\tikztgY) node (d2) [black,midway] {};
            \node[right = 0.1cm of d2] {Checker threads};

            \draw [decorate,decoration={brace,amplitude=4},thin] (\tikztgX*3.4,0.3*\tikztgY) -- (\tikztgX*3.4,-0.3*\tikztgY) node (d3) [black,midway] {};
            \node[right = 0.1cm of d3] {Guesser threads};

\end{tikzpicture}
}
\caption{
The communication topology between the threads of $\ConcProg$. 
Dashed edges indicate a $\rf$ relation between the threads.
Filled edges indicate a $\rf$ relation \emph{and} that both threads write to the respective location. 
\Revision{Double edges indicate a perfect one-way channel.}
\cref{FIG:RAEX_AXY} illustrates in more detail the communication between threads $\taxp$, $\tax$, $\tay$ and $\tayp$. 
\cref{FIG:RAEX} illustrates in more detail the communication between threads $\taxp$, $\tax$, $\txp$ and $\tx$. 
}
\label{fig:pcp_top}
\end{figure}

\Paragraph{\Revision{Constraining read events.}}
\Revision{Assumption \ref{item:ch_one_thread} holds for every location of our construction: each location is either written by only one thread, or it is one of the $x$-locations ($\vxa, \vxap, \vxb, \vxbp$) and $y$-locations ($\vya, \vyap, \vyb, \vybp$), whose values identify the thread to read from.
Events $\event_i$ on these locations write/read either $\tuple{t, v}$ or $\tuple{\ov{t}, v}$.
If $\event_i$ is a write event, the first field $t$ records $\tid(\event_i)$;
if $\event_i$ is a read event, it records the thread that $\event_i$ intends to read from.
Hence each read event $\rd_i$ reads from its intended thread, even though multiple threads write to $\lloc(\rd_i)$.}
\Revision{For example, the reads of $\txp$ on $\vxap$ read from $\taxp$ and never from $\txp$ itself, and its reads on $\vxa$ read from $\tx$ and never from $\tax$.}
\Revision{The overlined variant $\ov{t}$ is used only by the first write of each guesser thread $t \in \{\tax, \tbx, \tay, \tby\}$, and is expected only by the first read of the corresponding checker thread from $t$; hence, for example, $\rd_1(\tx, \vxa)$ reads from $\wt_1(\tax, \vxa)$, i.e., \ref{item:ch_first_read} holds.}
The second field $v$ depends on the value of its first field $t$.
For $x$-locations, $v$ is a letter $\gamma \in \PCPAlphabet$ \Revision{or the end marker $\bot$} if $t\in\{\tax,\tbx\}$, 0 if $t\in\{\taxp,\tbxp\}$, and $(i \bmod 4)$ if $t\in\{\tx,\txp\}$; 
for $y$-locations, $v$ is an index $j \in [n]$ \Revision{or the end marker $\bot$} if $t\in\{\tay,\tby\}$, 0 if $t\in\{\tayp,\tbyp\}$, and $(i \bmod 4)$ if $t\in\{\ty,\typ\}$. 

\Paragraph{Synchronization among guesser threads.}
\Revision{We enforce \ref{item:ch_sync}, i.e., that each thread executes in sync with its auxiliary thread, starting with the guesser threads.}
We introduce auxiliary locations $\vzax, \vzaxp, \vzay, \vzayp, \vzaxy$ and $\vzbx, \vzbxp, \vzby, \vzbyp, \vzbxy$; these, as well as $\vzaxyp$ and $\vzbxyp$, establish synchronization \Revision{among the guesser threads of each group} (see top and bottom rows of \cref{fig:pcp_top}). 
Overall, these locations enforce that each of the thread sets $\{\tax, \tay, \taxp, \tayp\}$ and $\{\tbx, \tby, \tbxp, \tbyp\}$ executes synchronously, meaning that the threads of each set, together, produce a sequentially-consistent execution if we ignore interactions with other threads. 

All events $\event_i$ on $z$-locations ($\vzax, \vzaxp, \vzay, \vzayp, \vzaxy, \vzbx, \vzbxp, \vzby, \vzbyp, \vzbxy$) write/read a single value $\val(\event_i) = (i \bmod 4)$. 
Events  $\event_i$ on $\ell$-locations ($\vzaxyp, \vzbxyp$) write/read $\tuple{\vc,v}$, where $\vc = (i \bmod 4)$ and $v$ is some index $j \in [n]$ of the PCP instance $\PCPInstance$. 
This distinction exists because events on $\vzaxyp$ (resp., $\vzbxyp$) have a double role:~besides synchronizing $\tax$ and $\tay$ (resp., $\tbx$ and $\tby$), they are also used for communicating the guessed PCP solution of $\tay$ to $\tax$ (resp., $\tby$ to $\tbx$), so that $\tax$ (resp., $\tbx$) can verify that it matches its own guess \Revision{(as outlined in \cref{SUBSEC:UNDECIDABILITY_INTUITION})}.

\newcommand\tikzwXa{0.8}
\newcommand\tikzwYa{0.8}

\begin{figure}
\centering
\begin{tikzpicture}[ line width=1pt, main/.style = {draw,inner sep=1pt}, font=\normalsize]

\draw [->] (10*\tikzwXa,0.5*\tikzwYa) -- (10*\tikzwXa,-13.3*\tikzwYa);
\node[] (aa) at (10*\tikzwXa,1*\tikzwYa){$\tayp$};
\draw [->] (6*\tikzwXa,0.5*\tikzwYa) -- (6*\tikzwXa,-13.3*\tikzwYa);
\node[] (aa) at (6*\tikzwXa,1*\tikzwYa){$\tay$};
\draw [->] (2*\tikzwXa,0.5*\tikzwYa) -- (2*\tikzwXa,-13.3*\tikzwYa);
\node[] (aa) at (2*\tikzwXa,1*\tikzwYa){$\tax$};
\draw [->] (-2*\tikzwXa,0.5*\tikzwYa) -- (-2*\tikzwXa,-13.3*\tikzwYa);
\node[] (aa) at (-2*\tikzwXa,1*\tikzwYa){$\taxp$};

\node[main, fill=white, draw=white, inner sep=3pt] (d1) at (-2*\tikzwXa,-12.8*\tikzwYa){$\dots$};
\node[main, fill=white, draw=white, inner sep=3pt] (d2) at (2*\tikzwXa,-12.8*\tikzwYa){$\dots$};
\node[main, fill=white, draw=white, inner sep=3pt] (d3) at (6*\tikzwXa,-12.8*\tikzwYa){$\dots$};
\node[main, fill=white, draw=white, inner sep=3pt] (d4) at (10*\tikzwXa,-12.8*\tikzwYa){$\dots$};

\foreach \i in {1,2,3} {
  \pgfmathsetmacro{\yoffset}{((\i-1)*2.5+2)}
  \node[main, fill=white, draw=white, inner sep=1pt] (w\i yap) at (10*\tikzwXa,-\yoffset*\tikzwYa)
    {\footnotesize$\wt_{\i}(\vyap,\tuple{\tayp,0})$};
  \node[main, fill=white, draw=white, inner sep=1pt] (w\i zay) at (10*\tikzwXa,{-(\yoffset+0.5)*\tikzwYa})
    {\footnotesize$\wt_{\i}(\vzay,\i)$};
  \node[main, fill=white, draw=white, inner sep=1pt] (r\i zayp) at (10*\tikzwXa,{-(\yoffset+1.0)*\tikzwYa})
    {\footnotesize$\rd_{\i}(\vzayp,\i)$};
}

\foreach \i in {1,2,3} {
  \pgfmathsetmacro{\yoffset}{((\i-1)*2.5)}
    \pgfmathtruncatemacro{\imod}{Mod(\i,4)}
  \pgfmathtruncatemacro{\iprev}{\i-1}
  \pgfmathtruncatemacro{\ipmod}{Mod(\iprev,4)}
  \ifnum\i>1
      \node[main, fill=white, draw=white, inner sep=1pt] (w\i ya) at (6*\tikzwXa,{-(\yoffset+0.0)*\tikzwYa})
        {\footnotesize$\wt_{\i}(\vya,\tuple{\tay, *})$};
    \else
        \node[main, fill=white, draw=white, inner sep=1pt] (w\i ya) at (6*\tikzwXa,{-(\yoffset+0.0)*\tikzwYa})
        {\footnotesize$\wt_{\i}(\vya,\tuple{\ov{\tay}, *})$};
    \fi
    \node[main, fill=gray!20, draw=white, inner sep=1pt] (w\i zaxyp) at (6*\tikzwXa,{-(\yoffset+0.5)*\tikzwYa})
    {\footnotesize$\wt_{\i}(\vzaxyp,\tuple{\i, *})$};
      \node[main, fill=white, draw=white, inner sep=1pt] (w\i zayp) at (6*\tikzwXa,{-(\yoffset+1.0)*\tikzwYa})
    {\footnotesize$\wt_{\i}(\vzayp,\i)$};
  \ifnum\i>1
    \node[main, fill=gray!20, draw=white, inner sep=1pt] (r\iprev zaxy) at (6*\tikzwXa,{-(\yoffset+1.5)*\tikzwYa})
      {\footnotesize$\rd_{\iprev}(\vzaxy,\ipmod)$};
    \node[main, fill=white, draw=white, inner sep=1pt] (r\iprev zay) at (6*\tikzwXa,{-(\yoffset+2.0)*\tikzwYa})
      {\footnotesize$\rd_{\iprev}(\vzay,\ipmod)$};
  \fi
}

\foreach \i in {1,2,3,4,5} {
  \pgfmathsetmacro{\yoffset}{(\i-1)*2.5}
  \pgfmathtruncatemacro{\imod}{Mod(\i,4)}
  \pgfmathtruncatemacro{\iprev}{\i-1}
  \pgfmathtruncatemacro{\ipmod}{Mod(\iprev,4)}
  \ifnum\i>1
  \node[main, fill=white, draw=white, inner sep=1pt] (w\i xa) at (2*\tikzwXa,{-(\yoffset+0.5)*\tikzwYa})
    {\footnotesize$\wt_{\i}(\vxa,\tuple{\tax,*})$};
    \else
\node[main, fill=white, draw=white, inner sep=1pt] (w\i xa) at (2*\tikzwXa,{-(\yoffset+0.5)*\tikzwYa})
    {\footnotesize$\wt_{\i}(\vxa,\tuple{\ov{\tax},*})$};
    \fi
    \ifnum\i=4
      \node[main, fill=white, draw=white, inner sep=1pt] (w\i zaxp) at (2*\tikzwXa,{-(\yoffset+1.0)*\tikzwYa})
        {\footnotesize$\wt_{\i}(\vzaxp,\imod)$\normalsize\squaredsmall{D}};
    \else
        \node[main, fill=white, draw=white, inner sep=1pt] (w\i zaxp) at (2*\tikzwXa,{-(\yoffset+1.0)*\tikzwYa})
            {\footnotesize$\wt_{\i}(\vzaxp,\imod)$};
    \fi
  \ifnum\i=2
    \node[main, fill=white, draw=white, inner sep=1pt] (r\iprev zax) at (2*\tikzwXa,{-(\yoffset+1.5)*\tikzwYa})
      {\footnotesize$\rd_{\iprev}(\vzax,\ipmod)$\normalsize\squaredsmall{A}};
  \fi
  \ifnum\i>2
    \node[main, fill=white, draw=white, inner sep=1pt] (r\iprev zax) at (2*\tikzwXa,{-(\yoffset+1.5)*\tikzwYa})
      {\footnotesize$\rd_{\iprev}(\vzax,\ipmod)$};
  \fi
}

\foreach \i in {1,3,4} {
  \pgfmathsetmacro{\yoffset}{(\i-1)*2.5}
  \pgfmathtruncatemacro{\idx}{(\i==1) ? 1 : ((\i==3) ? 2 : 3)}

    \node[main, fill=gray!20, draw=white, inner sep=1pt] (w\idx zaxy) at (2*\tikzwXa,{-(\yoffset+0)*\tikzwYa})
    {\footnotesize$\wt_{\idx}(\vzaxy,\idx)$};
}

\foreach \i in {2,3,5} {
  \pgfmathsetmacro{\yoffset}{(\i-1)*2.5}
  \pgfmathtruncatemacro{\idx}{(\i==2) ? 1 : ((\i==3) ? 2 : 3)}

  \node[main, fill=gray!20, draw=white, inner sep=1pt] (r\idx zaxyp) at (2*\tikzwXa,{-(\yoffset+2.0)*\tikzwYa})
    {\footnotesize$\rd_{\idx}(\vzaxyp,\tuple{\idx, *})$};
}

\foreach \i in {1,2,3,4,5} {
  \pgfmathsetmacro{\yoffset}{((\i-1)*2.5+1.5)}
    \pgfmathtruncatemacro{\imod}{Mod(\i,4)}
  \pgfmathtruncatemacro{\iprev}{\i-1}
  \pgfmathtruncatemacro{\ipmod}{Mod(\iprev,4)}
  \node[main, fill=white, draw=white, inner sep=1pt] (w\i axp) at (-2*\tikzwXa,-\yoffset*\tikzwYa)
    {\footnotesize$\wt_{\i}(\vxap,\tuple{\taxp,0})$};
    \ifnum\i<5
      \node[main, fill=white, draw=white, inner sep=1pt] (w\i zax) at (-2*\tikzwXa,{-(\yoffset+0.5)*\tikzwYa})
        {\footnotesize$\wt_{\i}(\vzax,\imod)$};
    \else
        \node[main, fill=white, draw=white, inner sep=1pt] (w\i zax) at (-2*\tikzwXa,{-(\yoffset+0.5)*\tikzwYa})
            {\footnotesize$\wt_{\i}(\vzax,\imod)$\normalsize\squaredsmall{B}};
    \fi
    \ifnum\i=4
      \node[main, fill=white, draw=white, inner sep=1pt] (r\i zaxp) at (-2*\tikzwXa,{-(\yoffset+1.0)*\tikzwYa})
    {\footnotesize$\rd_{\i}(\vzaxp,\imod)$\normalsize\squaredsmall{C}};
    \else
        \node[main, fill=white, draw=white, inner sep=1pt] (r\i zaxp) at (-2*\tikzwXa,{-(\yoffset+1.0)*\tikzwYa})
    {\footnotesize$\rd_{\i}(\vzaxp,\imod)$};
    \fi
}

\foreach \i in {1,2} {
  \draw [rf] (w\i zay.west) -- (r\i zay.east) node [midway,above left=-5pt]{};
}
\foreach \i in {1,2,3} {
  \draw [rf] (w\i zayp.east) -- (r\i zayp.west) node [midway,above left=-5pt]{};
}
\foreach \i in {2,3,4} {
  \draw [rf] (w\i zax.east) -- (r\i zax.west) node [midway,above left=-5pt]{};
}
\foreach \i in {2,3,4,5} {
  \draw [rf] (w\i zaxp.west) -- (r\i zaxp.east) node [midway,above left=-5pt]{};
}

\draw [rf] (w1zax.east) to node[]{\circledsmall{1}} (r1zax.west);
\draw [rf] (w1zaxp.west) to node[]{\circledsmall{2}} (r1zaxp.east);

\foreach \i in {1,2} {
  \draw [rf] (w\i zaxy.east) -- (r\i zaxy.west) node [midway,above left=-5pt]{};
}
\foreach \i in {1,2,3} {
  \draw [rf] (w\i zaxyp.west) -- (r\i zaxyp.east) node [midway,above left=-5pt]{};
}

\end{tikzpicture}
\caption{
An illustration of the synchronization between threads in $\{\tax, \tay, \taxp, \tayp\}$. 
Here $|\alpha_{j_1}|=2$, $|\alpha_{j_2}|=1$ and $|\alpha_{j_3}|=2$, and $\tax$ performs two write events to $\vxa$ corresponding to the characters of its first guessed word ($\wt_1(\tax,\vxa)$ and $\wt_2(\tax,\vxa)$), then one write event for the second ($\wt_3(\tax,\vxa)$), and two write events for the third ($\wt_4(\tax,\vxa)$ and $\wt_5(\tax,\vxa)$). 
We denote by $\AnyValue$ an unspecified value in an event, which corresponds to a value specific to the guessed solution of the PCP instance.
Highlighted events are relevant for our proof of completeness later.
}
\label{FIG:RAEX_AXY}
\end{figure}

\SubParagraph{Synchronization among $\tax$, $\taxp$, $\tay$, and $\tayp$.}
We now illustrate how $\tax$, $\taxp$, $\tay$, and $\tayp$ (the top row in \cref{fig:pcp_top}) use the aforementioned locations to interact, following the example in \cref{FIG:RAEX_AXY}.
The values shown for each event in the figure are chosen deterministically by the respective threads, with nondeterministic values relating to the guessed solution of the PCP instance shown as $*$.
Our analysis establishes why the $\rf$-edges on the deterministic events always follow this particular pattern.
\looseness=-1

First, observe that $\wt_1(\taxp, \vzax) \LTo{\rf} \rd_1(\tax, \vzax)$ (edge \circledsmall{1}), because $\rd_1(\tax, \vzax)$ (event \squaredsmall{A}) must read value $1$.
Any later write $\wt_i(\taxp, \vzax)$ that writes this value is such that $i\geq 5$ (for example, event \squaredsmall{B}).
However, any such $\wt_i(\taxp, \vzax)$ is preceded by $\rd_4(\taxp, \vzaxp)$ (event \squaredsmall{C}), which reads value $0$.
In turn, this value is only written by writes $\wt_j(\tax, \vzaxp)$ that appear after $\rd_1(\tax, \vzax)$ (for example, event \squaredsmall{D}).
Hence, if $\rd_1(\tax, \vzax)$ were to read from any write $\wt_i(\taxp, \vzax)$ with $i\geq 2$, this would cause an $\hb$-cycle 
$\wt_i(\taxp, \vzax)\LTo{\rf}\rd_1(\tax, \vzax) \LTo{\po} \wt_j(\tax, \vzaxp) \LTo{\rf} \rd_4(\taxp, \vzaxp) \LTo{\po} \wt_i(\taxp, \vzax)$, violating \ref{eq:hbirr}.
Symmetric reasoning on $\rd_1(\taxp, \vzaxp)$ establishes that $\wt_1(\tax,\vzaxp)\LTo{\rf} \rd_1(\taxp, \vzaxp)$ (edge \circledsmall{2}).
The argument repeats for the next two read events $\rd_2(\tax, \vzax)$ and $\rd_2(\taxp, \vzaxp)$,
inductively establishing the no-skipping invariant 
\[
\forall i \quad \wt_i(\taxp, \vzax) \LTo{\rf} \rd_i(\tax, \vzax)   \quad \text{and} \quad  \wt_i(\tax, \vzaxp) \LTo{\rf} \rd_i(\taxp,\vzaxp)
\numberthis\label{eq:pcp_inv_tax_taxp}
\]
\Revision{This is \ref{item:ch_sync} for $\tax$ and $\taxp$, with $\vzax$ and $\vzaxp$ in place of $\vchz$ and $\vchzp$ (edges \squaredsmall{C} and \squaredsmall{A} in \cref{fig:channel_example}).}
The same pattern holds for the interaction between $\tay$ and $\tayp$ over locations $\vzay$ and $\vzayp$, establishing the no-skipping invariant
\[
\forall i \quad \wt_i(\tayp, \vzay) \LTo{\rf} \rd_i(\tay, \vzay)   \quad \text{and} \quad  \wt_i(\tay, \vzayp) \LTo{\rf} \rd_i(\tayp,\vzayp)
\numberthis\label{eq:pcp_inv_tay_tayp}
\]
Finally, a similar pattern holds for the interaction between $\tax$ and $\tay$ over locations $\vzaxy$ and $\vzaxyp$, establishing the no-skipping invariant 
\[
\forall i \quad \wt_i(\tax, \vzaxy) \LTo{\rf} \rd_i(\tay, \vzaxy)   \quad \text{and} \quad  \wt_i(\tay, \vzaxyp) \LTo{\rf} \rd_i(\tax,\vzaxyp)
\numberthis\label{eq:pcp_inv_tax_tay}
\]
\SubParagraph{Synchronization among $\tbx$, $\tbxp$, $\tby$, and $\tbyp$.}
Threads $\tbx$, $\tbxp$, $\tby$, and $\tbyp$ are completely symmetric to $\tax$, $\taxp$, $\tay$, and $\tayp$.
They use locations $\vzbx$, $\vzbxp$, $\vzby$, $\vzbyp$, $\vzbxy$, and $\vzbxyp$ to form the same pattern as threads $\tax$, $\taxp$, $\tay$, and $\tayp$, and locations $\vzax$, $\vzaxp$, $\vzay$, $\vzayp$, $\vzaxy$, and $\vzaxyp$, respectively.
These establish the following no-skipping invariants, which are analogous to \cref{eq:pcp_inv_tax_taxp,eq:pcp_inv_tay_tayp,eq:pcp_inv_tax_tay}. 
\begin{align*}
    \forall i \quad &\wt_i(\tbxp, \vzbx) \LTo{\rf} \rd_i(\tbx, \vzbx) \quad \text{and} \quad \wt_i(\tbx, \vzbxp) \LTo{\rf} \rd_i(\tbxp,\vzbxp)
    \stepcounter{equation}\tag{\theequation}\label{eq:pcp_inv_tbx_tbxp}\\
    \forall i \quad &\wt_i(\tbyp, \vzby) \LTo{\rf} \rd_i(\tby, \vzby) \quad \text{and} \quad \wt_i(\tby, \vzbyp) \LTo{\rf} \rd_i(\tbyp,\vzbyp)
    \stepcounter{equation}\tag{\theequation}\label{eq:pcp_inv_tby_tbyp}\\
    \forall i \quad &\wt_i(\tbx, \vzbxy) \LTo{\rf} \rd_i(\tby, \vzbxy) \quad \text{and} \quad \wt_i(\tby, \vzbxyp) \LTo{\rf} \rd_i(\tbx,\vzbxyp)
    \stepcounter{equation}\tag{\theequation}\label{eq:pcp_inv_tbx_tby}
\end{align*}
\Paragraph{Synchronization among checker threads.} 
\Revision{Unlike the guesser threads, the checker threads do not need additional locations: they are synchronized with their auxiliary threads through the $x$- and $y$-locations.}
\cref{FIG:RAEX} illustrates the interaction between the threads $\tax$, $\taxp$, $\tx$ and $\txp$\Revision{, i.e., the pattern of \cref{fig:channel_example} with concrete values}.
Events that are not part of this pattern are omitted for readability.
\Revision{As in \cref{FIG:RAEX_AXY}, each thread executes the depicted sequence of operations deterministically, with $*$ marking non-deterministic values coming from the guessed solution of the PCP instance.}
\input{figures/undec_ra_example}
\SubParagraph{Synchronization between $\tx$ and $\txp$.}
As shown in \cref{FIG:RAEX}, $\tx$ and $\txp$ alternate reads and writes on $\vxa$ and $\vxap$, forming the same pattern as the guesser threads in \cref{FIG:RAEX_AXY}, and establishing the invariant
\[
\forall i \quad \wt_i(\tx, \vxa) \LTo{\rf} \rd_i(\txp, \vxa) \quad \text{and} \quad \wt_i(\txp, \vxap) \LTo{\rf} \rd_i(\tx, \vxap)
\numberthis\label{eq:pcp_inv_tx_txp}
\]
\Revision{This is \ref{item:ch_sync} for $\tx$ and $\txp$ (edges \squaredsmall{B} and \squaredsmall{D} in \cref{fig:channel_example}).}
The same pattern holds between $\tx$ and $\txp$ over locations $\vxb$ and $\vxbp$, establishing also that
\[
\forall i \quad \wt_i(\tx, \vxb) \LTo{\rf} \rd_i(\txp, \vxb) \quad \text{and} \quad \wt_i(\txp, \vxbp) \LTo{\rf} \rd_i(\tx, \vxbp)
\numberthis
\]
\SubParagraph{Synchronization between $\ty$ and $\typ$.} 
Similarly to above, the locations $\vya$, $\vyap$, $\vyb$, and $\vybp$ synchronize $\ty$ and $\typ$. 
The pattern between $\ty$ and $\typ$ over locations $\vya$ and $\vyap$, and the pattern between $\ty$ and $\typ$ over locations $\vyb$ and $\vybp$, establish the invariants
\[
\forall i \quad \wt_i(\ty, \vya) \LTo{\rf} \rd_i(\typ, \vya) \quad \text{and} \quad \wt_i(\typ, \vyap) \LTo{\rf} \rd_i(\ty, \vyap)
\numberthis\label{eq:pcp_inv_ty_typ}
\]
\[
\forall i \quad \wt_i(\ty, \vyb) \LTo{\rf} \rd_i(\typ, \vyb) \quad \text{and} \quad \wt_i(\typ, \vybp) \LTo{\rf} \rd_i(\ty, \vybp)
\numberthis
\]
\Paragraph{No-skipping between guesser and checker threads.} 
\Revision{With \ref{item:ch_one_thread}--\ref{item:ch_sync} in place, the argument of \cref{SUBSEC:UNDECIDABILITY_INTUITION} establishes the no-skipping property on locations $\vxa$, $\vya$, $\vxb$, and $\vyb$, between guesser and checker threads.}
\Revision{Equivalently, we obtain a perfect one-way channel from $\tax$ (resp., $\tay$, $\tbx$, $\tby$) to $\tx$ (resp., $\ty$, $\tx$, $\ty$) via $\vxa$ (resp., $\vya$, $\vxb$, $\vyb$).}
\Revision{\cref{FIG:RAEX} depicts the simulation of such a channel over $\vxa$, where $\tax$, $\taxp$, $\tx$, and $\txp$ take the roles of $\tsnd$, $\tsndp$, $\trcv$, and $\trcvp$, and $\vxa$ and $\vxap$ take the roles of $\vch$ and $\vchp$; its edge labels are analogous to those of \cref{fig:channel_example}.}
\Revision{We thus obtain the no-skipping invariants}
\begin{align*}
    \forall i \quad &\wt_i(\tax, \vxa) \LTo{\rf} \rd_i(\tx, \vxa) \quad \text{and} \quad \wt_i(\taxp, \vxap) \LTo{\rf} \rd_i(\txp, \vxap)
    \stepcounter{equation}\tag{\theequation}\label{eq:pcp_inv_no_skipping_xalpha}\\
    \forall i \quad &\wt_i(\tbx, \vxb) \LTo{\rf} \rd_i(\tx, \vxb) \quad \text{and} \quad \wt_i(\tbxp, \vxbp) \LTo{\rf} \rd_i(\txp, \vxbp)
    \stepcounter{equation}\tag{\theequation}\label{eq:pcp_inv_no_skipping_xbeta}\\
    \forall i \quad &\wt_i(\tay, \vya) \LTo{\rf} \rd_i(\ty, \vya) \quad \text{and} \quad \wt_i(\tayp, \vyap) \LTo{\rf} \rd_i(\typ, \vyap)
    \stepcounter{equation}\tag{\theequation}\label{eq:pcp_inv_no_skipping_yalpha}\\
    \forall i \quad &\wt_i(\tby, \vyb) \LTo{\rf} \rd_i(\ty, \vyb) \quad \text{and} \quad \wt_i(\tbyp, \vybp) \LTo{\rf} \rd_i(\typ, \vybp)
    \stepcounter{equation}\tag{\theequation}\label{eq:pcp_inv_no_skipping_ybeta}
\end{align*}
\Revision{Finally, the fact that all these patterns can be interleaved in a single consistent execution is shown in the completeness proof (\cref{SUBSEC:UNDECIDABILITY_COMPLETENESS}).}
\predisplaypenalty=\savedpredisplaypenalty

\subsection{Formal Reduction}\label{SUBSEC:UNDECIDABILITY_FORMAL} 

We now present the formal reduction.
For clarity of exposition,  we use common pseudocode language to represent the concurrent program.
Its translation to an LTS is intuitive and straightforward; \Revision{we refer to \cref{SEC:APP_SEMANTICS} for the formal semantics}.

\begin{figure*}
\centering
\begin{tabular}{|c|c|c|c|}
\hline
Thread $\tax$ & Thread  $\tay$ & Thread $\tbx$ & Thread $\tby$\\
\hline
\begin{lstlisting}[mathescape=true,]
$\alpha_{\bot} \gets \bot$;
$\vaux \gets\ndet{\{1,\dots, n\}}$;
while $\True$ do
  $\vcy \gets (\vcy + 1)\bmod 4$;
  $\wt(\vzaxy,\vcy)$;
  for $j \in 1,\dots,|\alpha_{\vaux}|$ do
    $prev \gets \vc$;
    $\vc \gets (\vc + 1)\bmod 4$;
    if $\firstit$ then
      $\wt(\vxa,\tuple{\ov{\tax}, \alpha_{\vaux}[j]})$;
      $\wt(\vzaxp,\vc)$;
    else
      $\wt(\vxa, \tuple{\tax, \alpha_{\vaux}[j]} )$;
      $\wt(\vzaxp,\vc)$;
      $\rd(\vzax,prev)$;
    endif
  done
  $\rd(\vzaxyp,\tuple{\vcy, \vaux})$;
  if $\vaux = \bot$ then
    break;
  endif
  $\vaux \gets \ndet{ \{1,\dots, n, \bot\}}$;
done
$\term$
\end{lstlisting} &
\begin{lstlisting}[mathescape=true,]
$\alpha_{\bot} \gets \bot$;
$\vaux \gets\ndet{\{1,\dots, n\}}$;
while $\True$ do
  $prev \gets \vc$;
  $\vc \gets (\vc + 1)\bmod 4$;
  if $\firstit$ then
    $\wt(\vya,\tuple{\ov{\tay}, \vaux})$;
    $\wt(\vzaxyp,\tuple{\vc,\vaux})$;
    $\wt(\vzayp,\vc)$;
  else
    $\wt(\vya,\tuple{\tay, \vaux})$;
    $\wt(\vzaxyp,\tuple{\vc,\vaux})$;
    $\wt(\vzayp,\vc)$;
    $\rd(\vzaxy,prev)$;
    $\rd(\vzay,prev)$;
  endif
  if $\vaux = \bot$ then
    break;
  endif
  $\vaux \gets \ndet{ \{1,\dots, n, \bot\}}$;
done
$\term$
\end{lstlisting} &
\begin{lstlisting}[mathescape=true,]
$\beta_{\bot} \gets \bot$;
$\vaux \gets\ndet{\{1,\dots, n\}}$;
while $\True$ do
  $\vcy \gets (\vcy + 1)\bmod 4$;
  $\wt(\vzbxy,\vcy)$;
  for $j \in 1,\dots,|\beta_{\vaux}|$ do
    $prev \gets \vc$;
    $\vc \gets (\vc + 1)\bmod 4$;
    if $\firstit$ then
      $\wt(\vxb,\tuple{\ov{\tbx}, \beta_{\vaux}[j]})$;
      $\wt(\vzbxp,\vc)$;
    else
      $\wt(\vxb, \tuple{\tbx, \beta_{\vaux}[j]} )$;
      $\wt(\vzbxp,\vc)$;
      $\rd(\vzbx,prev)$;
    endif
  done
  $\rd(\vzbxyp,\tuple{\vcy, \vaux})$;
  if $\vaux = \bot$ then
    break;
  endif
  $\vaux \gets \ndet{ \{1,\dots, n, \bot\}}$;
done
$\term$
\end{lstlisting} &
\begin{lstlisting}[mathescape=true,]
$\beta_{\bot} \gets \bot$;
$\vaux \gets\ndet{\{1,\dots, n\}}$;
while $\True$ do
  $prev \gets \vc$;
  $\vc \gets (\vc + 1)\bmod 4$;
  if $\firstit$ then
    $\wt(\vyb,\tuple{\ov{\tby}, \vaux})$;
    $\wt(\vzbxyp,\tuple{\vc,\vaux})$;
    $\wt(\vzbyp,\vc)$;
  else
    $\wt(\vyb,\tuple{\tby, \vaux})$;
    $\wt(\vzbxyp,\tuple{\vc,\vaux})$;
    $\wt(\vzbyp,\vc)$;
    $\rd(\vzbxy,prev)$;
    $\rd(\vzby,prev)$;
  endif
  if $\vaux = \bot$ then
    break;
  endif
  $\vaux \gets \ndet{ \{1,\dots, n, \bot\}}$;
done
$\term$
\end{lstlisting} \\
\hline
Thread $\taxp$ & Thread  $\tayp$ & Thread $\tbxp$ & Thread $\tbyp$\\
\hline
\begin{lstlisting}[mathescape=true]
while $\ndet{\{ \False, \True \}}$ do
  $\vc \gets (\vc + 1)\bmod 4$;
  $\wt(\vxap, \tuple{\taxp,0})$;
  $\wt(\vzax,\vc)$;
  $\rd(\vzaxp,\vc)$;
done
$\term$
\end{lstlisting} &
\begin{lstlisting}[mathescape=true,]
while $\ndet{\{ \False, \True \}}$ do
  $\vc \gets (\vc + 1)\bmod 4$;
  $\wt(\vyap, \tuple{\tayp,0})$;
  $\wt(\vzay,\vc)$;
  $\rd(\vzayp,\vc)$;
done
$\term$
\end{lstlisting} &
\begin{lstlisting}[mathescape=true,]
while $\ndet{\{ \False, \True \}}$ do
  $\vc \gets (\vc + 1)\bmod 4$;
  $\wt(\vxbp, \tuple{\tbxp,0})$;
  $\wt(\vzbx,\vc)$;
  $\rd(\vzbxp,\vc)$;
done
$\term$
\end{lstlisting} & 
\begin{lstlisting}[mathescape=true,]
while $\ndet{\{ \False, \True \}}$ do
  $\vc \gets (\vc + 1)\bmod 4$;
  $\wt(\vybp, \tuple{\tbyp,0})$;
  $\wt(\vzby,\vc)$;
  $\rd(\vzbyp,\vc)$;
done
$\term$
\end{lstlisting} \\
\hline
Thread $\tx$ & Thread $\ty$ & Thread $\txp$ & Thread $\typ$\\
\hline
\begin{lstlisting}[mathescape=true,]
$\vaux \gets \ndet{\PCPAlphabet}$;
while $\True$ do
  $prev \gets \vc$;
  $\vc \gets (\vc + 1)\bmod 4$;
  $\wt(\vxa,\tuple{\tx,\vc})$;
  $\wt(\vxb,\tuple{\tx,\vc})$;
  if $\firstit$ then
    $\rd(\vxa,\tuple{\ov{\tax}, \vaux})$;
    $\rd(\vxb,\tuple{\ov{\tbx}, \vaux})$;
  else
    $\rd(\vxa,\tuple{\tax, \vaux})$;
    $\rd(\vxb,\tuple{\tbx, \vaux})$;
    $\rd(\vxap,\tuple{\txp, prev})$;
    $\rd(\vxbp,\tuple{\txp, prev})$;
  endif
  if $\vaux = \bot$ then
	break;
  endif
  $\vaux \gets \ndet{ (\PCPAlphabet \cup \{\bot\})}$;
done
$\term$
\end{lstlisting} & 
\begin{lstlisting}[mathescape=true,]
$\vaux \gets  \ndet{\{1,\dots, n\}}$;
while $\True$ do
  $prev \gets \vc$;
  $\vc \gets (\vc + 1)\bmod 4$;
  $\wt(\vya,\tuple{\ty,\vc})$;
  $\wt(\vyb,\tuple{\ty,\vc})$;
  if $\firstit$ then
    $\rd(\vya,\tuple{ \ov{\tay}, \vaux})$;
    $\rd(\vyb,\tuple{\ov{\tby}, \vaux})$;
  else 
    $\rd(\vya,\tuple{\tay, \vaux})$;
    $\rd(\vyb,\tuple{\tby, \vaux})$;
    $\rd(\vyap,\tuple{\typ, prev})$;
    $\rd(\vybp,\tuple{\typ, prev})$;
  endif
  if $\vaux = \bot$ then
    break;
  endif
  $\vaux \gets \ndet{ \{1,\dots, n, \bot\}}$;
done
$\term$
\end{lstlisting} &
\begin{lstlisting}[mathescape=true,]
while $\ndet{\{\False, \True\}}$ do
  $\vc \gets (\vc + 1)\bmod 4$;
  $\wt(\vxap,\tuple{\txp,\vc})$;
  $\wt(\vxbp,\tuple{\txp,\vc})$;
  $\rd(\vxap,\tuple{\taxp,0})$;
  $\rd(\vxbp,\tuple{\tbxp,0})$;
  $\rd(\vxa,\tuple{\tx,\vc})$;
  $\rd(\vxb,\tuple{\tx,\vc})$;
done
$\term$
\end{lstlisting} &
\begin{lstlisting}[mathescape=true,]
while $\ndet{\{\False, \True\}}$ do
  $\vc \gets (\vc + 1)\bmod 4$;
  $\wt(\vyap,\tuple{\typ,\vc})$;
  $\wt(\vybp,\tuple{\typ,\vc})$;
  $\rd(\vyap,\tuple{\tayp,0})$;
  $\rd(\vybp,\tuple{\tbyp,0})$;
  $\rd(\vya,\tuple{\ty,\vc})$;
  $\rd(\vyb,\tuple{\ty,\vc})$;
done
$\term$
\end{lstlisting} \\
\hline
\end{tabular}
\caption{
The program $\ConcProg$ corresponding to a PCP instance $\PCPInstance$.
}
\label{fig:pcp_code}
\end{figure*}

\Paragraph{Reduction.}
Given a PCP instance $\PCPInstance$, we construct the program $\ConcProg$ shown in \cref{fig:pcp_code}.
$\ndet{X}$ refers to a non-deterministic choice of an element of a set $X$,
while a call to $\firstit$ in a line returns $\True$ (true) the first time this particular line is executed, and $\False$ (false) in every subsequent call.
Each counter $\vc$ is initialized to $0$, and so is $\vcy$ when present.
Every location $u$ is either written by only one thread, or all writes to $u$ in thread $t$ are of the form $\wt(u, \tuple{t, v})$ and all reads from $u$ in thread $t'$ are of the form $\rd(u, \tuple{t,v})$, where the $t$ indicates the thread this event must read from.
\Revision{The program follows the construction described in \cref{SUBSEC:UNDECIDABILITY_INTUITION,SUBSEC:UNDECIDABILITY_OVERVIEW} and \cref{fig:pcp_top,FIG:RAEX_AXY,FIG:RAEX}. }

First, consider the threads $\tax$ and $\tay$.
The main loop of each thread is over the value of its respective thread-local location $\vaux$, which contains a non-deterministic guess of the index $i\in \{1,\dots, n\}$ of the string $\alpha_i$ to be appended to the PCP sequence, or $\bot$, which marks the end of guessing.
In order to avoid producing an empty sequence, the first assignment to $\vaux$ cannot be $\bot$.
Thread $\tay$ communicates the guessed value to $\tax$ through $\vzaxyp$, and thread $\tax$ checks, at the end of its while loop, that it matches its own value.
The location $\vzaxyp$, alongside $\vzaxy$, is also used to ensure synchronization between $\tax$ and $\tay$, by creating the pattern of interaction shown in \cref{FIG:RAEX_AXY}.

Thread $\tay$ communicates its guess in $\vaux$ to thread $\ty$ by writing its value on $\vya$.
In particular, $\tay$ performs a write $\wt(\vya,\tuple{\ov{\tay}, \vaux})$ in the first iteration, and a write $\wt(\vya,\tuple{\tay, \vaux})$ in all subsequent iterations.
The purpose of using $\ov{\tay}$ instead of $\tay$ is to distinguish the first write, and not allow the first read on $\vya$ in $\ty$ to skip reading from it.
These writes on $\vya$ are also interleaved with other writes and reads on $\vzayp$ and $\vzay$, which create a pattern of interaction with $\tayp$ following \cref{FIG:RAEX_AXY}.

Thread $\tax$ also contains a for loop, inside its while loop, in which it writes the characters of $\alpha_{\vaux}$, one-by-one, to $\vxa$, to be read by $\tx$.
Similarly to the write to $\vya$ by $\tay$, the first write to $\vxa$ by $\tax$ is specially signaled by using $\tuple{\ov{\tax}, \alpha_{\vaux}[1]}$ instead of $\tuple{\tax,\alpha_{\vaux}[1]}$ as its value.
These writes are interleaved with other reads and writes on $\vzax$ and $\vzaxp$, which create a pattern of interaction with $\taxp$ following \cref{FIG:RAEX_AXY}.
This pattern and its interactions with writes to $\vxa$ is also shown in \cref{FIG:RAEX}.

Threads $\tbx$ and $\tby$ follow a symmetric process to $\tax$ and $\tay$, which can be explicitly defined by replacing all occurrences of $\alpha$ with $\beta$ above.

Next, consider the thread $\tx$.
Its main loop is over the value of $\vaux$, which contains a non-deterministic guess of the $i$-th letter of the PCP string.
It later performs  the reads $\rd(\vxa,\tuple{\ov{\tax}, \vaux})$ and $\rd(\vxb,\tuple{\ov{\tbx}, \vaux})$ in its first iteration, and the reads $\rd(\vxa,\tuple{\tax, \vaux})$ and $\rd(\vxb,\tuple{\tbx, \vaux})$ in subsequent iterations, verifying that the $i$-th letter produced by $\tax$ and the $i$-th letter produced by $\tbx$ match.
This thread also writes to locations $\vxa$ and $\vxb$, and reads from location $\vxap$ and $\vxbp$ in order to create a pattern of interaction with $\txp$ following \cref{FIG:RAEX}.
Thread $\ty$ is symmetric to $\tx$, but guesses indices of the PCP solution instead of letters and uses the locations $\vya$, $\vyb$, $\vyap$ and $\vybp$, instead of $\vxa$, $\vxb$, $\vxap$ and $\vxbp$.

Finally, the auxiliary threads $\taxp$, $\tayp$, $\tbxp$, $\tbyp$, $\txp$ and $\typ$ perform a non-deterministic sequence of writes and reads of values not related to the PCP instance.
This sequence establishes a synchronous pattern of interaction with the main threads and forces the \Revision{checker} threads $\tx$ and $\ty$ to not skip any of the writes of the guesser threads $\tax$, $\tay$, $\tbx$ and $\tby$, as shown in \cref{FIG:RAEX}. 

\subsection{Soundness}\label{SUBSEC:UNDECIDABILITY_SOUNDNESS}

We first establish the soundness of our reduction, as stated in the following lemma.

\begin{restatable}{lemma}{lemundecidabilitysoundness}\label{lem:undecidability_soundness}
If $\ConcProg$ has an execution where all threads reach $\term$ then $\PCPInstance$ has a solution.
\end{restatable}

The proof is based on the following no-skipping lemma, which establishes that for every thread $t$ and location $u$, the $i$-th read of $t$ on $u$ can only read from the $i$-th write of the (unique) thread writing the corresponding value to $u$.
\makeatletter
\newcommand{\tinyrf}{\mathrel{\ext@arrow 0022\rightarrowfill@{}{\textstyle\rf}}}
\makeatother
\begin{restatable}{lemma}{undecidabilitynoskippingpairs}\label{lem:undecidability_no_skipping_pairs}
Consider any execution $G=\tuple{\Events, \po, \rf, \mo}$ of $\ConcProg$ where all threads reach $\term$.
For all $i$ such that the corresponding read event $\rd_i$ exists in $G$, we have the following $\rf$-edges:
{\footnotesize
\allowdisplaybreaks
\begin{align*}
&\wt_i(\taxp, \vzax) \tinyrf \rd_i(\tax, \vzax), &
&\hspace{-1.15em}\wt_i(\tax, \vzaxp) \tinyrf \rd_i(\taxp, \vzaxp), &
&\hspace{-1.15em}\wt_i(\tayp, \vzay) \tinyrf \rd_i(\tay, \vzay), &
&\hspace{-1.25em}\wt_i(\tay, \vzayp) \tinyrf \rd_i(\tayp, \vzayp), \\
&\wt_i(\tbxp, \vzbx) \tinyrf \rd_i(\tbx, \vzbx), &
&\hspace{-1.15em}\wt_i(\tbx, \vzbxp) \tinyrf \rd_i(\tbxp, \vzbxp), &
&\hspace{-1.15em}\wt_i(\tbyp, \vzby) \tinyrf \rd_i(\tby, \vzby), &
&\hspace{-1.25em}\wt_i(\tby, \vzbyp) \tinyrf \rd_i(\tbyp, \vzbyp), \\
&\wt_i(\tax, \vzaxy) \tinyrf \rd_i(\tay, \vzaxy), &
&\hspace{-2.0em}\wt_i(\tay, \vzaxyp) \tinyrf \rd_i(\tax, \vzaxyp), &
&\hspace{-2.4em}\wt_i(\tbx, \vzbxy) \tinyrf \rd_i(\tby, \vzbxy), &
&\hspace{-3.5em}\wt_i(\tby, \vzbxyp) \tinyrf \rd_i(\tbx, \vzbxyp). 
\numberthis\label{item:no_skip_0}
\end{align*}
\begin{align*}
&\wt_i(\tx, \vxa) \LTo{\rf} \rd_i(\txp, \vxa), &
&\hspace{-0.25em}\wt_i(\txp, \vxap) \LTo{\rf} \rd_i(\tx, \vxap), &
&\hspace{-0.5em}\wt_i(\tx, \vxb) \LTo{\rf} \rd_i(\txp, \vxb), &
&\hspace{-0.75em}\wt_i(\txp, \vxbp) \LTo{\rf} \rd_i(\tx, \vxbp), \\
&\wt_i(\ty, \vya) \LTo{\rf} \rd_i(\typ, \vya), &
&\hspace{-0.25em}\wt_i(\typ, \vyap) \LTo{\rf} \rd_i(\ty, \vyap), &
&\hspace{-0.5em}\wt_i(\ty, \vyb) \LTo{\rf} \rd_i(\typ, \vyb), &
&\hspace{-0.75em}\wt_i(\typ, \vybp) \LTo{\rf} \rd_i(\ty, \vybp).
\numberthis\label{item:no_skip_1}
\end{align*}
\begin{align*}
&\wt_i(\tax, \vxa) \LTo{\rf} \rd_i(\tx, \vxa), &
&\hspace{-0.27em}\wt_i(\taxp, \vxap) \LTo{\rf} \rd_i(\txp, \vxap), &
&\hspace{-0.54em}\wt_i(\tbx, \vxb) \LTo{\rf} \rd_i(\tx, \vxb), &
&\hspace{-0.81em}\wt_i(\tbxp, \vxbp) \LTo{\rf} \rd_i(\txp, \vxbp), \\
&\wt_i(\tay, \vya) \LTo{\rf} \rd_i(\ty, \vya), &
&\hspace{-0.27em}\wt_i(\tayp, \vyap) \LTo{\rf} \rd_i(\typ, \vyap), &
&\hspace{-0.54em}\wt_i(\tby, \vyb) \LTo{\rf} \rd_i(\ty, \vyb), &
&\hspace{-0.81em}\wt_i(\tbyp, \vybp) \LTo{\rf} \rd_i(\typ, \vybp).
\numberthis\label{item:no_skip_2}
\end{align*}
}
\end{restatable}
\cref{item:no_skip_0} establishes the synchronization of guesser threads; its proof follows the steps outlined in the paragraph \Revision{\emph{Synchronization among guesser threads}} of \cref{SUBSEC:UNDECIDABILITY_OVERVIEW}.
\cref{item:no_skip_1} establishes the synchronization of \Revision{checker} threads; its proof follows the same pattern as that of \cref{item:no_skip_0}, as discussed in the paragraph \Revision{\emph{Synchronization among checker threads}} of \cref{SUBSEC:UNDECIDABILITY_OVERVIEW}.
Finally, \cref{item:no_skip_2} establishes no-skipping between guesser and \Revision{checker} threads; its proof is more involved, and \Revision{follows the argument of the paragraph \emph{Simulating a perfect channel} of \cref{SUBSEC:UNDECIDABILITY_INTUITION}, as instantiated} in the paragraph \emph{No-skipping between guesser and \Revision{checker} threads} of \cref{SUBSEC:UNDECIDABILITY_OVERVIEW}.

\Paragraph{Proof of \cref{lem:undecidability_soundness}.}
We finally argue that if all threads in $\ConcProg$ reach $\term$, then $\PCPInstance$ has a solution.

Let $j_1, \dots, j_k, \bot$ be the sequence of values guessed on $\vaux$ in $\tay$.
Since, by \cref{item:no_skip_0} of \cref{lem:undecidability_no_skipping_pairs}, we have $\wt_i(\tay, \vzaxyp) \LTo{\rf} \rd_i(\tax, \vzaxyp)$, this must also be the sequence of values guessed on $\vaux$ in $\tax$.
Notice that $j_i\neq \bot$ for all $i\in[k]$, which is enforced by the looping condition on \texttt{while} on $\tax$ and $\tay$. 
Moreover, the initial assignment of $\vaux$ enforces $k\geq 1$. 
The second field of the values written to $\vya$ by $\tay$ thus forms the sequence $j_1, \dots, j_k, \bot$, and the second field of the values written to $\vxa$ by $\tax$ forms the word $\alpha_{j_1}\cdots \alpha_{j_k} \cdot \bot$. 
Similarly, let $j'_1, \dots, j'_{k'}, \bot$ be the sequence of values guessed on $\vaux$ in $\tby$ and $\tbx$.
The interaction of $\tby$ and $\tbx$ on $\vzbxyp$ ensures this is indeed the same sequence.
The second fields of the values written to $\vyb$ by $\tby$ form the sequence $j'_1, \dots, j'_{k'}, \bot$, and the second fields of the values written to $\vxb$ by $\tbx$ form the word $\beta_{j'_1}\cdots \beta_{j'_{k'}} \cdot \bot$.

Let $\id_1, \dots, \id_{m}, \bot$ be the sequence of values guessed on $\vaux$ in $\ty$. 
By construction, this is the sequence of values that $\ty$ reads from $\vya$ and $\vyb$. 
By \cref{item:no_skip_2} of \cref{lem:undecidability_no_skipping_pairs}, $\wt_i(\tay, \vya) \LTo{\rf} \rd_i(\ty, \vya)$ and $\wt_i(\tby, \vyb) \LTo{\rf} \rd_i(\ty, \vyb)$ for all $i\in[m]$, and since the values read correspond to $\id_1, \dots, \id_{m}, \bot$, we have that $\id_1, \dots, \id_{m} = j_1, \dots, j_k = j'_1, \dots, j'_{k'}$. 
The final value $\bot$ forces $k = k' = m$, i.e., that the sequence of values written by $\tay$ and $\tby$ is read fully. 
A similar argument holds for $\tx$, yielding that the sequence $\ch_1, \dots, \ch_{\ell}, \bot$ assumed by $\vaux$ in $\tx$ is such that $\ch_1 \cdots \ch_{\ell} = \alpha_{j_1}\cdots \alpha_{j_k}=\beta_{j'_1}\cdots \beta_{j'_{k'}}$. 
We thus have that $j_1, \dots, j_k$ is a solution to the PCP instance $\PCPInstance$.
\subsection{Completeness}\label{SUBSEC:UNDECIDABILITY_COMPLETENESS}

Finally, we establish the completeness of the reduction.

\begin{restatable}{lemma}{lemundecidabilitycompleteness}\label{lem:undecidability_completeness}
If $\PCPInstance$ has a solution then $\ConcProg$ has an execution where all threads reach $\term$.
\end{restatable}

Given a solution to $\PCPInstance$, we construct an execution graph $G$ where every thread behaves according to $\ConcProg$, and then prove that $G$ is consistent under Release/Acquire.
This direction is considerably more intricate compared to soundness:~whereas the no-skipping lemma for soundness (\cref{lem:undecidability_no_skipping_pairs}) follows by analyzing the interaction of small sets of threads (each of size either $2$ or $4$), proving the consistency of $G$ has to account for all $12$ threads simultaneously.

\Paragraph{The execution graph $G=\tuple{\Events, \po, \rf, \mo}$.}
Let $j_1, \dots, j_k$ be a solution to $\PCPInstance$; we construct an execution graph $G=\tuple{\Events, \po, \rf, \mo}$ of $\ConcProg$ in which all threads reach $\term$.

\SubParagraph{The events of $G$.}
We specify the events of $G$ by listing how each thread 
behaves according to the solution $j_1, \dots, j_k$ of $\PCPInstance$. 
The event sequence of the non-primed threads  $\tax, \tay, \tbx, \tby, \tx$, and $\ty$ is defined by the values assumed by their respective locations $\vaux$, since this resolves all non-determinism locally in each thread.
For the threads $\tax$, $\tay$, $\tbx$, $\tby$, and $\ty$, we assign to the respective locations $\vaux$ the sequence of values $j_1, \dots, j_k, \bot$, while for $\tx$, we assign the sequence of letters in $\alpha_{j_1}\cdots \alpha_{j_k} \cdot \bot$.
Thus, $\tay$ (resp., $\tby$) writes $j_1, \dots, j_k, \bot$ (one-by-one, in this order) to $\vya$ (resp., $\vyb$), and $\tax$ (resp., $\tbx$) writes $\alpha_{j_1} \cdots \alpha_{j_k} \cdot \bot$ (letter-by-letter, in this order) to $\vxa$ (resp., $\vxb$) in its innermost loop. 
Threads $\tx$ and $\ty$ read the respective values, one-by-one, in each iteration of their loops.

The event sequence of the primed threads $\taxp, \tayp, \tbxp, \tbyp, \txp,$ and $\typ$, is defined by how many times their while loops execute.
The threads $\taxp$, $\tbxp$, and $\txp$ execute their loops $|\alpha_{j_1}\cdots \alpha_{j_k} \cdot \bot|$ times, and threads $\tayp$, $\tbyp$, and $\typ$ execute their loops $k+1$ times.

\SubParagraph{Specifying $G.\rf$.}
Recall that we use notation of the form $\wt_i(\tax, \vxa)$ to refer to the $i$-th write of $\tax$ on $\vxa$, and $\rd_i(\txp, \vxbp)$ to refer to the $i$-th read of $\txp$ on $\vxbp$. 
Also recall that we use event values to force each read $\rd_i$ to read from a particular thread.
We construct a reads-from relation that specifies $\wt_i\LTo{\rf}\rd_i$ for all $i$, where $\wt_i$ is the $i$-th write of the unique thread writing a value that matches that of $\rd_i$.
In particular, $\rf$ consists exactly of the relations shown in \cref{item:no_skip_0,item:no_skip_1,item:no_skip_2} 
of \cref{lem:undecidability_no_skipping_pairs}.
It is straightforward to verify that all such $\rf$ relationships are valid in the sense that the values read and written match.

\SubParagraph{Specifying $G.\mo$.}
Next, we specify the $\mo$ relation of $G$.
First, for any two writes of the form $\wt_i(t, u)$ and $\wt_j(t',u)$ (i.e., the $i$-th and the $j$-th write on location $u$ of threads $t$ and $t'$, respectively), if $i<j$ then $\wt_i\LTo{\mo} \wt_j$.
It remains to order concurrent writes with the same index. 
We resolve these orderings through the following $\mo$-edges.
\makeatletter
\newcommand{\tinymo}{\mathrel{\ext@arrow 0022\rightarrowfill@{}{\textstyle\mo}}}
\makeatother
{
\footnotesize
\allowdisplaybreaks
\begin{align*}
&\wt_i(\tx, \vxa) \tinymo \wt_{i}(\tax, \vxa), &
&\hspace{-1.1em}\wt_i(\tx, \vxb) \tinymo \wt_{i}(\tbx, \vxb), &
&\hspace{-1.2em}\wt_i(\ty, \vya) \tinymo \wt_{i}(\tay, \vya), &
&\hspace{-1.2em}\wt_i(\ty, \vyb) \tinymo \wt_{i}(\tby, \vyb), \\
&\wt_i(\txp, \vxap) \tinymo \wt_i(\taxp, \vxap), &
&\hspace{-1.1em}\wt_i(\txp, \vxbp) \tinymo \wt_i(\tbxp, \vxbp), &
&\hspace{-1.2em}\wt_i(\typ, \vyap) \tinymo \wt_i(\tayp, \vyap), &
&\hspace{-1.2em}\wt_i(\typ, \vybp) \tinymo \wt_i(\tbyp, \vybp).
\numberthis\label{eq:undecidability_mo_edges}
\end{align*}
}%
\cref{FIG:RAEX} illustrates some of these $\mo$-edges.
Note that each of the above locations is written by two threads, while every other location is written by a single thread, hence $\mo$ is total on each location.

In the remaining of this section, we prove that $G$ is consistent under Release/Acquire.
A key notion used extensively in our proofs is monotonicity with respect to event indices, as defined below.

\Paragraph{Monotonicity.}
Given a binary relation $R$ over the events of $G$, we call an $R$-edge $\event_i\LTo{R}\event_j$ \emph{non-decreasing} (resp., \emph{increasing}) if $i\leq j$ (resp., $i<j$), where $i$ and $j$ are the indices of $\event_i$ and $\event_j$, respectively.
We say that $R$ is non-decreasing (resp., increasing) if all its edges are non-decreasing (resp., increasing).
Observe that $\rf$ and $\mo$ are non-decreasing, by construction.
The crux of our proofs for the consistency of $G$ lies in  some key monotonicity properties for $\hb$ that imply the irreflexivity of the relations in \cref{fig:axioms_ra}.
The main challenge is that $\hb$ is \emph{not} non-decreasing, in general.

\Paragraph{Decreasing $\po$-edges.}
The $\hb$ relation is not non-decreasing because certain $\po$-edges of $G$ are decreasing.
We identify and characterize two classes of such edges that will require special treatment in our monotonicity lemmas later.
The first class involves edges that can decrease by more than one.
We call the events that receive such edges \emph{bridge events}.
The second class consists of \emph{decreasing read-to-read $\po$-edges} between non-bridge events.

\SubParagraph{Bridge events.}
Thread $\tax$ (resp., $\tbx$) writes to $\vzaxy$ and reads from $\vzaxyp$  (resp., $\vzbxy$ and $\vzbxyp$) less often than it writes/reads other locations.
This is because these locations are only accessed once per iteration of the \emph{outer} loop of $\tax$ (resp., $\tbx$), while other locations are accessed once per iteration of the \emph{inner} loop of $\tax$ (resp., $\tbx$).
This results in $\po$-edges entering the events of the outer loop potentially decreasing arbitrarily.
We call each location among $\vzaxy$, $\vzaxyp$, $\vzbxy$, $\vzbxyp$, a \emph{bridge location},
and an event $\event$ accessing a bridge location a \emph{bridge event}.
Bridge events are highlighted in \cref{FIG:RAEX_AXY}.
The next lemma captures some monotonicity properties of $\po$-edges that do not involve bridge events. 
Its proof follows by examining the sequence of events that each thread executes in isolation.

\begin{restatable}{lemma}{lempostrictmonotonic}\label{lem:po_strict_monotonic}
Consider any $\po$-edge $\event_i\LTo{\po}\wt_j$ where $\wt_j$ is a non-bridge write.
Then
\begin{enumerate*}[label=(\roman*)]
\item\label{item:po_non_decreasing}
$i\leq j$, and
\item\label{item:po_increasing}
if $\op(\event_i)=\rd$, then $i<j$.
\end{enumerate*}
\end{restatable}

\SubParagraph{Read-to-read edges.}
In the first iteration of the while loop in $\tx$ and $\ty$, read events in certain locations are skipped, while read events in other locations are executed. 
This causes the indices of the read events in the skipped locations to lag behind the indices of the other read events by one, thereby causing decreasing read-to-read $\po$-edges.
For example, the edge $\rd_3(\tx,\vxa)\LTo{\po}\rd_2(\tx,\vxap)$, in \cref{FIG:RAEX}, is decreasing because in the first iteration of its while loop, $\tx$ reads from $\vxa$, but not $\vxap$.
The next lemma is an exhaustive account of all decreasing $\po$-edges between non-bridge read events.
Similarly to \cref{lem:po_strict_monotonic}, its proof follows directly by examining the sequence of events that each thread generates in isolation.

\begin{restatable}{lemma}{lemdecreasingrrpo}\label{lem:decreasing_r_r_po}
Consider any decreasing edge $\rd_i\LTo{\po}\rd_j$ where neither $\rd_i$ nor $\rd_j$ is a bridge event.
Then $j=i-1$, and $\tuple{\rd_i, \rd_j}$ is one of the following pairs
{\footnotesize
\begin{align*}
\tuple{\rd_i(\tx, \vxa), \rd_j(\tx, \vxap)}, \ \tuple{\rd_i(\tx, \vxa), \rd_j(\tx, \vxbp)},  \ &\tuple{\rd_i(\tx, \vxb), \rd_j(\tx, \vxap)}, \ \tuple{\rd_i(\tx, \vxb), \rd_j(\tx, \vxbp)}\\
\tuple{\rd_i(\ty, \vya), \rd_j(\ty, \vyap)}, \ \tuple{\rd_i(\ty, \vya), \rd_j(\ty, \vybp)}, \ 
&\tuple{\rd_i(\ty, \vyb), \rd_j(\ty, \vyap)}, \ \tuple{\rd_i(\ty, \vyb), \rd_j(\ty, \vybp)}\numberthis
\label{eq:decreasing_r_r_po}
\end{align*}
}
In particular, all decreasing edges $\rd_i\LTo{\po}\rd_j$ between non-bridge read events are such that $\rd_i$ reads from a guesser thread. 
\end{restatable}

\Paragraph{Alternating and minimal $\hb$-paths.}
We call a $\hb$-path \emph{alternating} if it has the form
\[
a\LTo{\po^?}\wt_{i_1}\LTo{\rf}\rd_{i_1}\LTo{\po}\wt_{i_2}\LTo{\rf}\rd_{i_2}\LTo{\po}\cdots\LTo{\po}\wt_{i_{\ell}}\LTo{\rf}\rd_{i_{\ell}}\LTo{\po^?}b
\numberthis\label{eq:undecidability_hb_minimal}
\]
i.e., it is an alternating sequence of $\po$ and $\rf$ edges.
We write $\Path\colon a\LPath{\hb}b$ to denote such an alternating $\hb$-path.
Note that if $a \LTo{\hb}b$, then there exists an alternating $\hb$-path of the form of \cref{eq:undecidability_hb_minimal} from $a$ to $b$.
We call $\Path$ a cycle if $a=b$ and $\Path$ contains at least one edge.
We say that $\Path$ is \emph{contained} in a thread set $\{t_1,\dots, t_m\}$ if all of its events $\event$ are such that $\tid(\event)\in \{t_1,\dots, t_m\}$.
We call $\Path$ \emph{minimal} if every $\po$-edge of $\Path$ is on a distinct thread.
Finally, we call $\Path\colon \event_i\LPath{\hb} \event_j$ non-decreasing (resp., increasing) if $i\leq j$ (resp., $i<j$).
The next lemma establishes \ref{eq:hbirr} for $G$.

\begin{restatable}{lemma}{lemhbirr}\label{lem:hbirr}
The $\hb$ relation is irreflexive.  
\end{restatable}
\begin{proof}
Assume towards contradiction that there is an alternating $\hb$-cycle $\Path \colon a\LPath{\hb}a$.
Since no thread reads from itself, $\Path$ cannot be contained on any single thread, \Revision{and thus contains at least one $\rf$-edge.
Since the presence of a cycle implies the presence of all its cyclic shifts, we may assume wlog that $\Path$ starts, and ends, at the source of an $\rf$-edge, and thus has the form}
\[
\Path\colon \wt_{i_1}\LTo{\rf}\rd_{i_1}\LTo{\po}\wt_{i_2}\LTo{\rf}\dots\LTo{\rf}\rd_{i_\ell}\LTo{\po}\wt_{i_1} \ .
\]
\Revision{In particular, every $\po$-edge of $\Path$ connects a read event to a write event.}
Moreover, wlog, we may assume that $\Path$ is minimal; if that is not the case, we can reduce $\Path$ through the following procedure:
\begin{enumerate*}[label=(\roman*)]
    \item let $\rd_1\LTo{\po}\wt_1$ and $\rd_2\LTo{\po}\wt_2$ be two $\po$-edges in the same thread in $\Path$, 
    \item assume wlog that $\rd_1\LTo{\po}\rd_2$, 
    \item take some cyclic shift of $\Path$ in which $\rd_1$ appears before $\wt_2$, and
    \item replace the path from $\rd_1$ to $\wt_2$ with $\rd_1\LTo{\po}\wt_2$.
\end{enumerate*}
Note that this procedure yields a $\hb$-cycle with strictly fewer $\po$-edges, and thus applying it repeatedly eventually yields a minimal cycle.

Now, since $\Path$ is minimal, each of its $\po$-edges lies on a distinct thread, and these threads, linked by the $\rf$-edges of $\Path$, form a simple cycle in the communication topology of \cref{fig:pcp_top}.
Since this topology contains no simple cycle through three or more threads, $\Path$ must be fully contained in one of the following pairs of mutually $\rf$-connected threads
\[
\{\tax,\taxp\}, \quad \{\tax,\tay\}, \quad \{\tay,\tayp\}, \quad
\{\tbx,\tbxp\}, \quad \{\tbx,\tby\}, \quad \{\tby,\tbyp\}, \quad
\{ \tx, \txp \}, \quad \{ \ty, \typ \}\  .
\]
\Revision{Hence $\Path$ has exactly two $\po$-edges, one on each thread}, and thus has the following form
\[
\Path\colon \wt_{i}\LTo{\rf}\rd_{i}\LTo{\po}\wt_{j}\LTo{\rf}\rd_{j}\LTo{\po}\wt_{i} \ .
\]
Note that every event of $\Path$ accesses a location that is accessed by both threads that $\Path$ traverses.
We consider two cases, depending on the pairs of threads that $\Path$ is contained in.

\SubParagraph{Pairs $\{\tax,\taxp\}$, $\{\tay,\tayp\}$, $\{\tbx,\tbxp\}$, $\{\tby,\tbyp\}$, $\{\tx,\txp\}$ and $\{\ty,\typ\}$.}
None of these thread pairs share bridge locations, so all events in $\Path$ are non-bridge events.
By \cref{lem:po_strict_monotonic}, read-to-write $\po$-edges are increasing, and by construction, $\rf$-edges are non-decreasing.
Therefore $i < j < i$, a contradiction.
    
\SubParagraph{Pairs $\{\tax,\tay\}$ and $\{\tbx, \tby\}$.}
We first consider the thread pair $\{\tax,\tay\}$. 
The shared locations of this pair are $\vzaxy$ and $\vzaxyp$, hence every event in $\Path$ must access one of these two locations.
The only read-to-write $\po$-edges connecting events in these locations are of the form $\rd(\tax,\vzaxyp)\LTo{\po}\wt(\tax,\vzaxy)$ and  $\rd(\tay,\vzaxy)\LTo{\po}\wt(\tay,\vzaxyp)$.
Note that $\tax$ alternates writes to $\vzaxy$ and reads from $\vzaxyp$, starting with a write to $\vzaxy$. 
Therefore, $\po$-edges of the form $\rd(\tax,\vzaxyp)\LTo{\po}\wt(\tax,\vzaxy)$ are increasing (see, for instance,  $\rd_1(\tax,\vzaxyp)\LTo{\po}\wt_2(\tax,\vzaxy)$ in \cref{FIG:RAEX_AXY}).
Similarly, $\tay$ alternates writes to $\vzaxyp$ and reads from $\vzaxy$, and additionally skips the first read to $\vzaxy$. 
Therefore, $\po$-edges of the form $\rd(\tay,\vzaxy)\LTo{\po}\wt(\tay,\vzaxyp)$ are also increasing (see, for instance,  $\rd_1(\tay,\vzaxy)\LTo{\po}\wt_3(\tay,\vzaxyp)$ in \cref{FIG:RAEX_AXY}).
This yields $i < j < i$, a contradiction.
\looseness=-1

The same argument applies to the pair $\{\tbx,\tby\}$, which follows an analogous pattern.
In this case, the respective locations are $\vzbxy$ and $\vzbxyp$, and all edges of the form $\rd(\tbx,\vzbxyp)\LTo{\po}\wt(\tbx,\vzbxy)$ and  $\rd(\tby,\vzbxy)\LTo{\po}\wt(\tby,\vzbxyp)$ are increasing.
\end{proof}
Besides establishing \ref{eq:hbirr}, \cref{lem:hbirr} allows us to consider only minimal $\hb$-paths between events of $G$, as stated in the following lemma.

\begin{restatable}{lemma}{lemhbminimal}\label{lem:hb_minimal}
For any two events $a$, $b$ of $G$, if $a\LTo{\hb}b$ then there exists a minimal $\hb$-path $\Path\colon a\LPath{\hb}b$.
\end{restatable}

This result aids us in establishing further properties of $\hb$. 
To further simplify following lemmas by exploiting the symmetry of our construction, we group threads in two sets.

\Paragraph{Thread sets $\ThreadDom_x$ and $\ThreadDom_y$.}
In our construction, threads split naturally into two groups according to the locations they access: $\ThreadDom_x = \{\tax,\taxp,\tbx,\tbxp,\tx,\txp\}$ for threads reading and writing $x$-locations, and $\ThreadDom_y = \{\tay,\tayp,\tby,\tbyp,\ty,\typ\}$ for $y$-locations.
This separation is key to \cref{lem:write_to_write_no_bridge}, 
which states that, wlog, $\hb$-paths between non-bridge events in the same thread set are bridge-free.

\begin{restatable}{lemma}{lemwritetowritenobridge}\label{lem:write_to_write_no_bridge}
Let $a, b$ be non-bridge events such that $\tid(a),\tid(b)\in \ThreadDom_x$ or $\tid(a),\tid(b)\in \ThreadDom_y$.
If $a \LTo{\hb} b$, any minimal $\hb$-path $\Path\colon a\LPath{\hb}b$ does not contain bridge events.
\end{restatable}
The following three lemmas state some additional monotonicity properties of $\hb$ that will be useful later.
\cref{lem:write_to_write_no_bridge_increasing} establishes monotonicity properties of minimal $\hb$-paths that avoid bridge events (for instance, the ones just established in \cref{lem:write_to_write_no_bridge}), 
\cref{lem:rr_bridge_monotonic} states that $\hb$-edges between same-location reads on checker threads are increasing, and \cref{lem:ww_dec_2} states that certain $\hb$-edges from writes on a primed thread to its non-primed variant increase the index by at least two. 

\begin{restatable}{lemma}{lemwritetowritenobridgeincreasing}\label{lem:write_to_write_no_bridge_increasing}
Consider any minimal $\hb$-path $\Path\colon\wt_i\LPath{\hb}\wt_j$ that does not contain a bridge event. 
Then
\begin{enumerate*}[label=(\roman*)]
\item\label{item:ww_nb_non_dec}
$i\leq j$, and
\item\label{item:ww_nb_inc}
if $\tid(\wt_i)\neq \tid(\wt_j)$, then $i<j$.
\end{enumerate*}
\end{restatable}

\begin{restatable}{lemma}{lemrrbridgemonotonic}\label{lem:rr_bridge_monotonic}
Every $\hb$-edge $\rd_i\LTo{\hb}\rd_j$ such that $\lloc(\rd_i)=\lloc(\rd_j)$ and $\tid(\rd_i),\tid(\rd_j) \in \{\tx, \txp, \ty, \typ\}$, is such that $i<j$. 
\end{restatable}

\begin{restatable}{lemma}{lemwwdec}\label{lem:ww_dec_2}
Every $\hb$-edge $\wt_i\LTo{\hb}\wt_j$ where $\tuple{\wt_i,\wt_j}$ is one of the following pairs
{\small
\[
\tuple{\wt_i(\taxp, \vxap), \wt_j(\tax, \vxa)}, \ \tuple{\wt_i(\tayp, \vyap), \wt_j(\tay, \vya)}, \ \tuple{\wt_i(\tbxp, \vxbp), \wt_j(\tbx, \vxb)}, \ \tuple{\wt_i(\tbyp, \vybp), \wt_j(\tby, \vyb)}
\]}
is such that $i<j-1$. 
\end{restatable}

We now arrive at the two general monotonicity properties of $\hb$ between same-location events:~$\hb$ is increasing between writes (\cref{lem:hb_conflicting_writes_strictly_monotonic}), and non-decreasing from a write to a read (\cref{lem:hb_c_monotonic}).

\begin{restatable}{lemma}{lemhbconflictingwritesstrictlymonotonic}\label{lem:hb_conflicting_writes_strictly_monotonic}
For any location $u$, we have that $\identity{\WriteDom};\hb_u;\identity{\WriteDom}$ is increasing.
\end{restatable}

The proof shows that for any two write events $\wt_i$, $\wt_j$ with $\wt_i\LTo{\hb_u}\wt_j$, we have $i<j$.
We distinguish two cases.
If $\tid(\wt_i)=\tid(\wt_j)$, \cref{lem:hbirr} implies that $\wt_i\LTo{\po}\wt_j$ and thus $i< j$.
Otherwise, a case analysis on $u$ shows that either $\tid(\wt_i),\tid(\wt_j)\in \ThreadDom_x$ or $\tid(\wt_i),\tid(\wt_j)\in \ThreadDom_y$.
Then, \cref{lem:hb_minimal} states that there exists a minimal $\hb$-path $\Path\colon\wt_i\LPath{\hb}\wt_j$, \cref{lem:write_to_write_no_bridge} implies that $\Path$ does not contain bridge events, and, finally, \cref{lem:write_to_write_no_bridge_increasing} implies that $\Path$ is increasing.

\begin{restatable}{lemma}{lemhbcmonotonic}\label{lem:hb_c_monotonic}
For any location $u$, we have that $\identity{\WriteDom};\hb_u;\identity{\ReadDom}$ is non-decreasing.
\end{restatable}

The proof proceeds by a case analysis on the shape of a minimal $\hb$-path $\Path\colon\wt_i\LPath{\hb}\rd_j$; most cases follow directly from \cref{lem:hb_conflicting_writes_strictly_monotonic} and the structure of each thread.
The most intricate case is when $\tid(\wt_i)$ is a guesser thread, $\tid(\rd_j)$ is a checker thread, and $\Path$ ends in a $\po$-edge. 
One example of this case, present in \cref{FIG:RAEX}, is $\wt_1(\taxp, \vxap)\LTo{\hb}\rd_2(\tx, \vxap)$ and the minimal $\hb$-path
\[
\wt_1(\taxp, \vxap)\LTo{\po}\wt_1(\taxp, \vzax)\LTo{\rf}\rd_1(\tax, \vzax)\LTo{\po}\wt_3(\tax, \vxa)\LTo{\rf}\rd_3(\tx, \vxa)\LTo{\po}\rd_2(\tx, \vxap).
\]
The difficulty stems from the fact that, as shown in the example, the last $\po$-edge may be decreasing. 
Nevertheless, notice that
\begin{enumerate*}[label=(\roman*)]
\item\label{item:hb_c_monotonic1} this last $\po$-edge of $\Path$ decreases by $1$,
\item\label{item:hb_c_monotonic2} it is preceded by the subpath $\Path'\colon \wt_1(\taxp, \vxap)\LTo{\po}\wt_1(\taxp, \vzax)\LTo{\rf}\rd_1(\tax, \vzax)\LTo{\po}\wt_3(\tax, \vxa)$ which is increasing, and
\item\label{item:hb_c_monotonic3} $\Path$ does not contain any other decreasing edges,
\end{enumerate*}
thereby $\Path$ is non-decreasing overall.
A careful analysis shows that this pattern is not accidental.
For \cref{item:hb_c_monotonic1}, the last decreasing $\po$-edge is always of the form of \cref{eq:decreasing_r_r_po}, and thus can only decrease by $1$, as guaranteed by \cref{lem:decreasing_r_r_po}.
For \cref{item:hb_c_monotonic2}, we prove that $\Path$ must have a subpath where \cref{item:ww_nb_inc} of \cref{lem:write_to_write_no_bridge_increasing} applies, and thus is increasing (like $\Path'$ in our example above).
Finally, \cref{item:hb_c_monotonic3} follows from the facts that $\Path$ is bridge-free due to \cref{lem:write_to_write_no_bridge}, read-to-write $\po$-edges between non-bridge events are increasing due to \cref{lem:po_strict_monotonic}, and $\rf$-edges are non-decreasing.

\Paragraph{Proof of \cref{lem:undecidability_completeness}.}
We are now ready to argue about the consistency of $G$. 
\cref{lem:hbirr} already establishes \ref{eq:hbirr}, so it remains to argue about \ref{eq:wc} and  \ref{eq:rc}. 

\Paragraph{\ref{eq:wc}.}
Consider any $\mo$-edge $\wt_i\LTo{\mo}\wt_j$, and by construction, we have $i\leq j$.
By \cref{lem:hb_conflicting_writes_strictly_monotonic}, we have that $\wt_j \nLTo{\hb} \wt_i$.
Thus, $\mo;\hb$ is irreflexive, as desired.

\Paragraph{\ref{eq:rc}.}
Consider three events $\wt_i$, $\rd_i$, and $\wt_j$ where $\lloc(\wt_i)=\lloc(\rd_i)=\lloc(\wt_j)$. 
We argue that if $\wt_i\LTo{\rf} \rd_i$ and $\wt_i\LTo{\mo} \wt_j$, then $\wt_j\nLTo{\hb} \rd_i$;
therefore, no triple of events can constitute a \ref{eq:rc} violation. 
By construction, $\wt_i\LTo{\mo} \wt_j$ implies $i\leq j$.
If $i<j$, then $\wt_j\nLTo{\hb} \rd_i$ follows from \cref{lem:hb_c_monotonic}.
We are thus left with the case where $i=j$.

In particular, since $i=j$ and $\wt_i\LTo{\mo}\wt_j$, by construction, $\tuple{\wt_i, \rd_i, \wt_j}$ must be one of the following:
{\small
\begin{align*}
&\tuple{\wt_i(\tx, \vxa), \rd_i(\txp, \vxa), \wt_j(\tax, \vxa)}, \quad \tuple{\wt_i(\txp, \vxap), \rd_i(\tx, \vxap), \wt_j(\taxp, \vxap)}\\
&\tuple{\wt_i(\ty, \vya), \rd_i(\typ, \vya), \wt_j(\tay, \vya)}, \quad \tuple{\wt_i(\typ, \vyap), \rd_i(\ty, \vyap), \wt_j(\tayp, \vyap)}\\
&\tuple{\wt_i(\tx, \vxb), \rd_i(\txp, \vxb), \wt_j(\tbx, \vxb)}, \quad \tuple{\wt_i(\txp, \vxbp), \rd_i(\tx, \vxbp), \wt_j(\tbxp, \vxbp)}\\
&\tuple{\wt_i(\ty, \vyb), \rd_i(\typ, \vyb), \wt_j(\tby, \vyb)}, \quad \tuple{\wt_i(\typ, \vybp), \rd_i(\ty, \vybp), \wt_j(\tbyp, \vybp)}
\numberthis \label{list:wrw_conflict}
\end{align*}}%

We note two properties of each triplet $\tuple{\wt_i, \rd_i, \wt_j}$ of \cref{list:wrw_conflict}. 
First, either $\tid(\rd_i), \tid(\wt_j) \in \ThreadDom_x$ or $\tid(\rd_i), \tid(\wt_j) \in \ThreadDom_y$. 
Second, $\tid(\wt_i)$ and $\tid(\rd_i)$ are distinct checker threads, and $\tid(\wt_j)$ is a guesser thread. 

Assume towards contradiction that $\wt_j \LTo{\hb}\rd_i$.
By \cref{lem:hb_minimal}, there must exist a minimal $\hb$-path $\Path\colon \wt_j \LPath{\hb}\rd_i$.
Since $\tid(\rd_i), \tid(\wt_j) \in \ThreadDom_x$ or $\tid(\rd_i), \tid(\wt_j) \in \ThreadDom_y$, \cref{lem:write_to_write_no_bridge} guarantees that $\Path$ does not contain bridge events. 
Next, observe that any $\hb$-path from $\wt_j$ to $\rd_i$ is of the form $\wt_j \LTo{\hb^?} \wt_{\ell} \LTo{\rf} \rd_{\ell} \LTo{\hb^?} \rd_i$, where $\tid(\wt_{\ell})$ is a guesser thread  and $\tid(\rd_{\ell})$ is a checker thread. 
Let $u=\lloc(\wt_i)=\lloc(\rd_i)=\lloc(\wt_j)$.
We consider the two cases, depending on whether $\lloc(\wt_{\ell})=u$.

\SubParagraph{Case $\lloc(\wt_{\ell}) = u$.} 
If $\rd_i = \rd_\ell$, then $\wt_i = \wt_{\ell}$.
Since $\wt_i \neq \wt_j$, \cref{lem:hb_conflicting_writes_strictly_monotonic} and $\wt_j\LTo{\hb}\wt_i$ implies that $j<i$, which contradicts $j=i$.
Otherwise, since $\lloc(\rd_\ell)=\lloc(\rd_i)$, \cref{lem:rr_bridge_monotonic} implies that $\ell<i$.
We also have that $j\leq \ell$, since either $\wt_j=\wt_\ell$ (so $j=\ell$), or 
 $\wt_j\neq \wt_\ell$, and \cref{lem:hb_conflicting_writes_strictly_monotonic} implies that $j<\ell$.
Combining the two inequalities, we arrive at $j<i$, contradicting that $j=i$.

\SubParagraph{Case $\lloc(\wt_{\ell}) \neq u$.}
This case is more involved.
First, note that there are eight potential locations for $\wt_{\ell}$, namely $\vxa$, $\vxap$, $\vya$, $\vyap$, $\vxb$, $\vxbp$, $\vyb$, and $\vybp$, since these are the locations that connect a guesser thread to checker thread via an $\rf$-edge. 
Since $\wt_j\LTo{\hb}\wt_\ell$, the location of $\wt_\ell$ is constrained by the thread topology in \cref{fig:pcp_top}.
The guesser threads $\tax,\taxp,\tay,$ and $\tayp$, which write to locations $\vxa,\vxap,\vya,$ and $\vyap$, have no $\hb$-paths to the guesser threads $\tbx,\tbxp,\tby,$ and $\tbyp$, which write to locations $\vxb,\vxbp,\vyb,$ and $\vybp$, and vice-versa.
Therefore, $u$ and $\lloc(\wt_\ell)$ belong to the same group, i.e., either $u, \lloc(\wt_\ell)\in \{\vxa,\vxap,\vya,\vyap\}$ or $u, \lloc(\wt_\ell)\in\{\vxb,\vxbp,\vyb,\vybp\}$.

Similarly, since $\rd_\ell\LTo{\hb}\rd_i$, the location of $\rd_\ell$ is constrained by the thread topology in \cref{fig:pcp_top}.
The checker threads $\tx$ and $\txp$, which write to locations $\vxa,\vxap,\vxb$, and $\vxbp$, have no $\hb$-paths to the checker threads in $\ty$ and $\typ$, which write to locations $\vya,\vyap,\vyb$, and $\vybp$, and vice-versa. 
Therefore, $u$ and $\lloc(\rd_\ell)$ belong to the same group, i.e., either $u, \lloc(\rd_\ell)\in \{\vxa,\vxap,\vxb,\vxbp\}$ or $u, \lloc(\rd_\ell)\in\{\vya,\vyap,\vyb,\vybp\}$.
Since $\lloc(\wt_\ell)=\lloc(\rd_\ell)$, this constraint also applies to $\wt_\ell$.
These two constraints together imply that both $u$ and $\lloc(\wt_\ell)$ are contained in exactly one of $\{\vxa,\vxap\}$, $\{\vya,\vyap\}$, $\{\vxb,\vxbp\}$, or $\{\vyb,\vybp\}$.
Since $\lloc(\wt_{\ell}) \neq u$, the pair $\tuple{\wt_j, \wt_{\ell}}$ must be one of the following:
{\small
\begin{align*}
\tuple{\wt_j(\tax, \vxa), \wt_{\ell}(\taxp, \vxap)}, \hspace{0.25em} \tuple{\wt_j(\tay, \vya), \wt_{\ell}(\tayp, \vyap)}, \hspace{0.25em} &\tuple{\wt_j(\taxp, \vxap), \wt_{\ell}(\tax, \vxa)}, \hspace{0.25em} \tuple{\wt_j(\tayp, \vyap), \wt_{\ell}(\tay, \vya)}\\
\tuple{\wt_j(\tbx, \vxb), \wt_{\ell}(\tbxp, \vxbp)}, \hspace{0.25em} \tuple{\wt_j(\tby, \vyb), \wt_{\ell}(\tbyp, \vybp)}, \hspace{0.25em} &\tuple{\wt_j(\tbxp, \vxbp), \wt_{\ell}(\tbx, \vxb)}, \hspace{0.25em} \tuple{\wt_j(\tbyp, \vybp), \wt_{\ell}(\tby, \vyb)}\numberthis
\end{align*}
}
We split into two cases according to whether $\lloc(\wt_j)$ is the non-primed or primed location in this pair.
These correspond to the left and right columns of \cref{list:wrw_conflict}, respectively.

\begin{compactitem}
    \item \emph{Non-primed locations ($\vxa, \vya, \vxb$, and $\vyb$).} 
    We have $j < \ell$ due to \cref{lem:write_to_write_no_bridge_increasing}. 
    The path $\rd_{\ell} \LPath{\hb^?} \rd_i$ is composed of $\rf$-edges, read-to-write $\po$-edges and a single read-to-read $\po$-edge, ending at $\rd_i$. 
    The $\rf$-edges are non-decreasing by construction; 
    the read-to-write $\po$-edges are increasing due to \cref{lem:po_strict_monotonic}. 
    Finally, the final $\po$-edge cannot be any of the cases in \cref{lem:decreasing_r_r_po}, as these all end in primed locations, and $\lloc(\rd_i)$ is non-primed.
    Therefore, $\ell \leq i$. 
    Combining the two inequalities, we arrive at $j<i$, contradicting that $j=i$.

    \item \emph{Primed locations  ($\vxap, \vyap, \vxbp$, and $\vybp$).} 
    In these cases, $\tuple{\wt_j,\wt_\ell}$ falls into one of the cases listed in \cref{lem:ww_dec_2}.
    Therefore, $j<\ell-1$. 
    The path $\rd_{\ell} \LPath{\hb^?} \rd_i$ is composed of $\rf$-edges, read-to-write $\po$-edges and a single read-to-read $\po$-edge, ending at $\rd_i$. 
    The $\rf$-edges are non-decreasing by construction; 
    the read-to-write $\po$-edges are increasing due to \cref{lem:po_strict_monotonic}. 
    Finally, by \cref{lem:decreasing_r_r_po}, the final $\po$-edge is decreasing by at most one;
    overall, this yields $\ell - 1 \leq i$.
    Combining the two inequalities, we arrive at $j<i$, contradicting that $j=i$.
    
\end{compactitem}

\Revision{We thus obtain that} $G$ does not violate \ref{eq:rc}.
Together with the absence of \ref{eq:wc} and \ref{eq:hbirr} violations (\cref{lem:hbirr}), this shows that $G$ is consistent under Release/Acquire, completing the proof of \cref{lem:undecidability_completeness}.
Combined with \cref{lem:undecidability_soundness}, this proves \cref{thm:undecidability}.

\section{Decidability with Bounded Context Switches}\label{SEC:DECBCS}

In this section, we turn our attention to reachability under bounded context switches, and prove \cref{thm:bounded_context_switches}.
First, we consider the RMW-free fragment of Release/Acquire (\cref{SUBSEC:REDUCIBLE_TRACES,SUBSEC:IRREDUCIDBLE_REACHABILITY}), since the main proof challenges arise already in this setting.
Then, we extend our proof to handle a bounded number of RMWs, \Revision{and comment on the complexity of the problem} (\cref{SUBSEC:DECIDABILITY_RMWS}).

\Paragraph{Notation on event sequences.}
Consider a sequence of events $\pi = \event_1, \event_2, \dots, \event_n$.
Given some event $\event$, we write $\event\in \pi$ to denote that $\event=\event_i$ for some $i\in \IntSet{n}$.
Given two events $\event_i, \event_j\in \pi$, we write $\event_i<_{\pi}\event_j$ to denote that $i<j$, and write $\leq_{\pi}$ for the reflexive closure of $<_{\pi}$.
Finally, given two events $\event_i, \event_j$ with $\event_i<_{\pi}\event_j$,  we define  $\In_{\pi}(\event_i,\event_j)=\{\event\in \pi \colon \event_i<_{\pi} \event\leq_{\pi}\event_j \}$ (note that containment is inclusive on the right endpoint only).

\Paragraph{Contexts.}
Consider a sequence of events $\pi = \event_1, \event_2, \dots, \event_n$.
Given a thread $t$, a \emph{$t$-run} of $\pi$ is a contiguous sub-sequence $\event_{i+1}, \event_{i+2}, \dots, \event_{i+j}$ of $\pi$ such that $\tid(\event_{i+\ell})=t$, for all $\ell\in \IntSet{j}$.
We simply say that $\event_{i+1}, \event_{i+2}, \dots, \event_{i+j}$ is a \emph{run} when $t$ is not important, or clear from the context.
We say that the run is \emph{maximal} if it is not a subsequence of another run of $\pi$.
Thus, $\pi$ can be uniquely decomposed into a sequence $\pi=\pi_1\cdot \pi_2\cdots \pi_{\ell}$ such that each $\pi_i$ is a maximal run.
In such a case, we let $\ctx{\pi}={\ell}$ be the \emph{number of contexts}
\footnote{
\Revision{
We note that contexts were defined slightly differently in}~\cite{Abdulla19} \Revision{, based on an operational model of Release/Acquire.
Our definition is at least as permissive, and thus the undecidaiblity of reachability with RMWs under 4 contexts also holds here.
}
}
of $\pi$.
A \emph{trace} is a pair $\LinearTrace=\tuple{G, (\pi_1, \pi_2,\dots, \pi_{\ell})}$ where 
\begin{enumerate*}[label*=(\roman*)]
\item $G$ is an execution graph,
\item $\pi = \pi_1\cdot \pi_2 \cdots \pi_{\ell}$ is a topological ordering of $G.\hb$, and
\item each $\pi_i$ is a maximal run of $\pi$.
\end{enumerate*}
Given an event $\event\in G.\Events$, we write $\ContextId_{\LinearTrace}(\event)$ for the unique $i\in \IntSet{\ell}$ such that $\event\in \pi_i$.


\Paragraph{Context-bounded reachability.}
Given a concurrent program $\ConcProg$ and a trace $\LinearTrace=\tuple{G, (\pi_1, \pi_2,\dots, \pi_{\ell})}$, we
say that $\LinearTrace$ is a trace of $\ConcProg$ if $G\in \ExecutionGraphsOf{\ConcProg}{\ramm}$.
The problem of \emph{context-bounded reachability} asks whether, given a concurrent program $\ConcProg$ and some $k\in \Nats$, there exists a trace $\LinearTrace=\tuple{G, (\pi_1, \pi_2,\dots, \pi_{\ell})}$ of $\ConcProg$, where $\ell\leq k$, that reaches the final state $\tuple{\Prog_{t_1}.\FinalState, \Prog_{t_2}.\FinalState, \dots, \Prog_{t_{\numThreads}}.\FinalState}$ of $\ConcProg$.

\subsection{Reducible Traces}\label{SUBSEC:REDUCIBLE_TRACES}

We first consider the RMW-free fragment of Release/Acquire.
The key insight behind the decidability of this fragment is that, in the absence of RMWs, context-bounded reachability exhibits a small model property:~if the final state of the program is indeed reachable, then it can be reached by an execution with size bounded by an exponential in the size of the program and the number of context switches.
\looseness=-1

We prove that reachability has small models by showing that, given a trace $\LinearTrace=\tuple{G, (\pi_1, \pi_2,\dots, \pi_{\ell})}$ witnessing reachability, if some run $\pi_i$ is long enough, then it contains a segment that can be collapsed.
The challenge in this style of proofs is to identify, for each event $\event\in \pi_i$, a bounded summary of the run up to $\event$, which is sufficient to conclude that repeating summaries mark a collapsible segment. 
In the current section, we define our event summaries, introduce some additional conditions under which repeating summaries mark a collapsible segment, indicating that the trace is \emph{reducible}, and establish a bound on the size of irreducible traces.
In the next section we prove that reachability can be solved by only considering irreducible witnesses.

\Paragraph{Event summaries.}
Consider a program $\ConcProg$ and a trace $\LinearTrace=\tuple{G, (\pi_1, \dots, \pi_{\ell})}$ of $\ConcProg$.
Consider some event $\event\in G.\Events$, and let $t=\tid(\event)$ and $c=\ContextId_{\LinearTrace}(\event)$, for brevity.
Given some location $x\in \LocationDom$, we write $\LatestWrite_{\LinearTrace}(\event, x)$ for the maximal, wrt $\leq_{\pi_{c}}$, write $\wt(t,x)\in \pi_c$ such that $\wt(t,x)\leq_{\pi_{c}} \event$. 
If no such write exists, we let $\LatestWrite_{\LinearTrace}(\event,x)=\bot$.
We associate with $\event$ a \emph{summary} $\Summary_{\LinearTrace}(\event)=\tuple{\SummaryStates, \SummaryLastValue, \SummaryExternalRf}$, where each component is defined as follows.
\begin{compactitem}
\item $\SummaryStates\subseteq \Prog_{t}.\States$ is the set of states that $\Prog_{t}$ reaches by executing all events of thread $t$ up to (including) $\event$.
Formally, let $\event_1, \event_2,\dots, \event_n$ be the (unique) prefix of $G.\po^{t}$ up to $\event_n=\event$.
Then $q\in \SummaryStates$ iff $\llab(\event_1), \llab(\event_2),\dots, \llab(\event_n)$ reaches $q$ in $\Prog_{t}$.
\item $\SummaryLastValue\colon \LocationDom\to \ValueDom\cup \{ \bot \}$ is a map from each location $x$ to the value of the most recent write of thread $t$ to $x$ in $\pi_{c}$.
Formally, $\SummaryLastValue(x)=\val_w(\LatestWrite_{\LinearTrace}(\event,x))$, if $\LatestWrite_{\LinearTrace}(\event,x)\neq \bot$, and $\SummaryLastValue(x)=\bot$ otherwise.
\item $\SummaryExternalRf\subseteq \LocationDom$ is the set of locations that $t$ reads from a remote thread between its most recent local write on said location in $\pi_{c}$ and $\event$. 
Formally, $x\in \SummaryExternalRf$ iff $\LatestWrite_{\LinearTrace}(\event,x)\neq \bot$ and there exists a read $\rd(t,x)\in\In_{\pi_{c}}(\LatestWrite_{\LinearTrace}(\event,x), \event)$ such that $\LatestWrite_{\LinearTrace}(\event,x) \nLTo{\rf} \rd(t,x)$.
\end{compactitem}

\Paragraph{Reducible traces.}
Consider a trace  $\LinearTrace=\tuple{G, (\pi_1, \dots, \pi_{\ell})}$.
A pair of events $\tuple{\CSei, \CSej}$, with $\CSei, \CSej\in G.\Events$, is called \emph{collapsible} if $\CSei <_\pi \CSej$ and, letting $c=\ContextId_{\LinearTrace}(\CSei)$, the following conditions hold:
\begin{enumerate*}[label=(\roman*)]
\item\label{item:cond_red_1} $\ContextId_{\LinearTrace}(\CSej)=c$,
\item\label{item:cond_red_2} $\Summary_{\LinearTrace}(\CSei)=\Summary_{\LinearTrace}(\CSej)$,
\item\label{item:cond_red_3} for every write $\wt(t,x)\in \In_{\pi_{c}}(\CSei, \CSej)$, if $\wt(t,x)\LTo{\rf}\event$, for some event $\event\in G.\Events$, then $\ContextId_{\LinearTrace}(\event)=c$, and 
\item\label{item:cond_red_4} for every $x\in \LocationDom$ and $\event \in (G.\Events \setminus G.\Events^{\tid(\CSei)})$, we have $\LatestWrite_{\LinearTrace}(\CSei,x)\LTo{\hb} \event$ iff $\LatestWrite_{\LinearTrace}(\CSej,x)\LTo{\hb} \event$.
\end{enumerate*}
Intuitively, $\CSei$ and $\CSej$ have the same context id and summary, 
there is no write between them that is observed by a later run, and
the latest write on each location reaches the same set of events in other threads via $\hb$.
We call $\LinearTrace$ \emph{reducible} if it contains a collapsible pair of events, and \emph{irreducible} otherwise.
The following lemma gives a bound on the size of irreducible traces.

\begin{restatable}{lemma}{lemmairreduciblesmall}\label{lem:irreducible_small}
For any concurrent program $\ConcProg$ and trace $\LinearTrace=\tuple{G, (\pi_1, \dots, \pi_{\ell})}$ of $\ConcProg$,
if $\LinearTrace$ is irreducible then \Revision{$|G.\Events|=2^{O(\ell\cdot (|\ConcProg.\States| + |\LocationDom|\cdot \log(|\ValueDom|+1)))}$}.
\end{restatable}

\subsection{A Small Model Property}\label{SUBSEC:IRREDUCIDBLE_REACHABILITY}

The small model property of reachability in the RMW-free fragment follows from \cref{lem:irreducible_small} and the following lemma, which we prove in this section. 

\begin{restatable}{lemma}{lemirreduciblereachability}\label{lem:irreducible_reachability}
If there is a trace $\LinearTrace = \tuple{G, (\pi_1, \pi_2, \dots, \pi_{\ell})}$ of $\ConcProg$ reaching a state $\tuple{p_1, p_2, \dots, p_{\numThreads}}$, 
then there exists an irreducible trace with the same property.
\end{restatable}

\newcommand\tikzreX{0.6}
\newcommand\tikzreY{0.65}

\begin{figure}
\begin{subfigure}[b]{.3\textwidth}
\centering
    \begin{tikzpicture}[ line width=1pt, main/.style = {rectangle, inner sep=2pt}]

\draw [->] (2*\tikzreX,2*\tikzreY) -- (2*\tikzreX,-2*\tikzreY) node [near start, above=1cm]{}; 

\draw [->] (5*\tikzreX,2*\tikzreY) -- (5*\tikzreX,-2*\tikzreY) node [near start, above=1cm]{}; 

\draw [densely dashed, gray,  line width = 0.4pt] (3.5*\tikzreX,2.5*\tikzreY) to node [at start, above=0.2cm, black] {$G \to G'$}(3.5*\tikzreX,-2*\tikzreY);

\node[main, fill=white] at (2*\tikzreX,2.5*\tikzreY) (pi1){$\pi_c$}; 
\node[main, fill=white] at (5*\tikzreX,2.5*\tikzreY) (pi2){$\pi_c'$};

\node[main, fill=white, inner sep=2pt] at (2*\tikzreX,1.5*\tikzreY) (w4i1){\small$\dots$}; 
\node[main, fill=white]  at (2*\tikzreX,1*\tikzreY)(w4c1){\small$\CSei$}; 
\node[main, fill=white, inner sep=2pt]  at (2*\tikzreX,0.5*\tikzreY) (w4i2){\small$\dots$};
\node[main, fill=white]  at (2*\tikzreX,0*\tikzreY) (w4c2){{\small$\CSej$}};  
\node[main, fill=white, inner sep=2pt] at (2*\tikzreX,-0.5*\tikzreY) (w4c3){\small$\dots$}; 

\node[main, fill=white, inner sep=2pt] at (5*\tikzreX,1.5*\tikzreY) (w4i1){\small$\dots$}; 
\node[main, fill=white] at (5*\tikzreX,1*\tikzreY) (w4i1){\small$\CSei$}; 
\node[main, fill=white, inner sep=2pt] at (5*\tikzreX,-0.5*\tikzreY) (w4c3){\small$\dots$}; 

\fill[ red, opacity = 0.2] (1.5*\tikzreX,0.75*\tikzreY) rectangle (2.5*\tikzreX,-0.25*\tikzreY);

\end{tikzpicture}
\caption{$G$ and $G'$}
\label{subfig:bcs_gp_1}
\end{subfigure}%
\begin{subfigure}[b]{.3\textwidth}
\centering
\begin{tikzpicture}[ line width=1pt, main/.style = {rectangle, inner sep=2pt}] 

\draw [->] (2*\tikzreX,2*\tikzreY) -- (2*\tikzreX,-2*\tikzreY) node [near start, above=1cm]{}; 

\draw [->] (5*\tikzreX,2*\tikzreY) -- (5*\tikzreX,-2*\tikzreY) node [near start, above=1cm]{}; 

\draw [densely dashed, gray,  line width = 0.4pt] (3.5*\tikzreX,2.5*\tikzreY) to node [at start, above=0.2cm, black] {$G \to G'$}(3.5*\tikzreX,-2*\tikzreY);

\node[main, fill=white] at (2*\tikzreX,2.5*\tikzreY) (pi1){$\pi_c$}; 
\node[main, fill=white] at (5*\tikzreX,2.5*\tikzreY) (pi2){$\pi_c'$}; 

\node[main, fill=white] at (5*\tikzreX,1.5*\tikzreY) (w4i1){\small$\CSwi$}; 
\node[main, fill=white, inner sep=2pt] at (5*\tikzreX,1*\tikzreY) (w4i3){\small$\dots$}; 
\node[main, fill=white, inner sep=2pt] at (5*\tikzreX,-1*\tikzreY) (w4c4){\small$\dots$}; 
\node[main, fill=white] at (5*\tikzreX,-1.5*\tikzreY) (w4c3){\small$\rd$}; 
\node[main, fill=white] at (2*\tikzreX,1.5*\tikzreY) (aw4i1){\small$\CSwi$}; 
\node[main, fill=white, inner sep=2pt] at (2*\tikzreX,1*\tikzreY) (aw4i3){\small$\dots$}; 
\node[main, fill=white, inner sep=2pt]  at (2*\tikzreX,0.5*\tikzreY)(aw4c1){\small$\dots$}; 
\node[main, fill=white]  at (2*\tikzreX,0*\tikzreY) (aw4i2){\small$\wt$};

\node[main, fill=white, inner sep=2pt]  at (2*\tikzreX,-0.5*\tikzreY) (aw4c2){\small$\dots$};  

\node[main, fill=white, inner sep=2pt] at (2*\tikzreX,-1*\tikzreY) (aw4c4){\small$\dots$}; 

\node[main, fill=white] at (2*\tikzreX,-1.5*\tikzreY) (aw4c3){\small$\rd$}; 

\fill[ red, opacity = 0.2] (1.5*\tikzreX,0.75*\tikzreY) rectangle (2.5*\tikzreX,-0.75*\tikzreY);

\draw [rf] (aw4i2.0) to[bend left=30] (aw4c3.0);
\draw [rf] (w4i1.0) to[bend left=20] (w4c3.0);

\end{tikzpicture}
\caption{$G.\rf$ and $G'.\rf$}
\label{subfig:bcs_gp_2}
\end{subfigure}%
\begin{subfigure}[b]{.3\textwidth}
\centering
\begin{tikzpicture}[ line width=1pt, main/.style = {rectangle, inner sep=2pt}] 

\draw [->] (2*\tikzreX,2*\tikzreY) -- (2*\tikzreX,-2*\tikzreY) node [near start, above=1cm]{}; 

\draw [->] (7*\tikzreX,2*\tikzreY) -- (7*\tikzreX,-2*\tikzreY) node [near start, above=1cm]{}; 

\draw [densely dashed, gray,  line width = 0.4pt] (3.5*\tikzreX,2.5*\tikzreY) to node [at start, above=0.2cm, black] {$G \to G'$}(3.5*\tikzreX,-2*\tikzreY);

\node[main, fill=white] at (2*\tikzreX,2.5*\tikzreY) (pi1){$\pi_c$}; 
\node[main, fill=white] at (7*\tikzreX,2.5*\tikzreY) (pi2){$\pi_c'$}; 

\node[main, fill=white] at (2*\tikzreX,1.5*\tikzreY) (w12){\small$\CSwi$}; 
\node[main, fill=white, inner sep=2pt] at (2*\tikzreX,1*\tikzreY) (aw4i3){\small$\dots$}; 
\node[main, fill=white, inner sep=2pt]  at (2*\tikzreX,0.5*\tikzreY)(aw4c1){\small$\dots$}; 
\node[main, fill=white]  at (2*\tikzreX,0*\tikzreY) (w11){\small$\CSwj$};
\node[main, fill=white, inner sep=2pt]  at (2*\tikzreX,-0.5*\tikzreY) (aw4c2){\small$\dots$};  
\node[main, fill=white, inner sep=2pt] at (2*\tikzreX,-1*\tikzreY) (aw4c3){\small$\dots$};

\node[main, fill=white] at (0*\tikzreX,2*\tikzreY) (wa11){\small$\wt$}; 
\node[main, fill=white] at (0*\tikzreX,0.75*\tikzreY) (wa12){\small$\wt$}; 
\node[main, fill=white] at (0*\tikzreX,-0.5*\tikzreY) (wa13){\small$\wt$}; 

\draw[mo] (wa11) to (w12.north west);
\draw[mo] (wa11) to (w11.north west);
\draw[mo] (wa12) to (w11.west);
\draw[mo] (w12.west) to (wa12);
\draw[mo] (w12.south west) to (wa13);
\draw[mo] (w11.south west) to (wa13);

\fill[ red, opacity = 0.2] (1.5*\tikzreX,0.75*\tikzreY) rectangle (2.5*\tikzreX,-0.75*\tikzreY);

\node[main, fill=white] at (7*\tikzreX,1.5*\tikzreY) (w22){\small$\CSwi$}; 
\node[main, fill=white, inner sep=2pt] at (7*\tikzreX,1*\tikzreY) (aw4i3){\small$\dots$}; 
\node[main, fill=white, inner sep=2pt] at (7*\tikzreX,-1*\tikzreY) (aw4c3){\small$\dots$};

\node[main, fill=white] at (5*\tikzreX,2*\tikzreY) (wa21){\small$\wt$}; 
\node[main, fill=white] at (5*\tikzreX,0.75*\tikzreY) (wa22){\small$\wt$}; 
\node[main, fill=white] at (5*\tikzreX,-0.5*\tikzreY) (wa23){\small$\wt$}; 

\draw[mo] (wa21) to (w22.north west);
\draw[mo] (wa22) to (w22.west);
\draw[mo] (w22.south west) to (wa23);

\end{tikzpicture}
\caption{$G.\mo$ and $G'.\mo$}
\label{subfig:bcs_gp_3}
\end{subfigure}%
\caption{A graph $G$ with $\Summary(\CSei) = \Summary(\CSej)$ and the corresponding $G'$ (\subref{subfig:bcs_gp_1}).
(\subref{subfig:bcs_gp_2}) illustrates the recovery of missing $\rf$ edges, and
(\subref{subfig:bcs_gp_3}) illustrates how $\mo$ is adapted to restore \ref{eq:rc} for the recovered $\rf$.}
\label{fig:bcs_gp}
\end{figure}

\Paragraph{Reducing $\LinearTrace = \tuple{G, (\pi_1, \pi_2, \dots, \pi_{\ell})}$.}
Assume that $\LinearTrace$ is reducible via a collapsible pair of events $\tuple{\CSei, \CSej}$, and let $t=\tid(\CSei)=\tid(\CSej)$ and $c=\ContextId_{\LinearTrace}(\CSei)=\ContextId_{\LinearTrace}(\CSej)$.
We construct a new trace $\LinearTrace' = \tuple{G', (\pi'_1, \pi'_2, \dots, \pi'_{\ell})}$, such that $G'.\Events\subset G.\Events$ and $G'$ also reaches $\tuple{p_1, p_2, \dots, p_{\numThreads}}$.
Since reachability witnesses are finite, \cref{lem:irreducible_reachability} then follows.

We now describe the construction of $G'$ in three steps, each illustrated in \cref{fig:bcs_gp}.
The sequence $(\pi'_1, \pi'_2, \dots, \pi'_{\ell})$ is obtained naturally, by projecting each $\pi_{l}$ to the events in $G'.\Events$.

\SubParagraph{1.~Collapsing the segment between $\CSei$ and $\CSej$.}
We obtain the event set $G'.\Events$ by removing from $G$ all events between $\CSei$ (exclusive) and $\CSej$ (inclusive), i.e., $G'.\Events=G.\Events \setminus\In_{\pi_{c}}(\CSei, \CSej)$.
The program order $G'.\po$ is $G.\po$ restricted to $G'.\Events$, i.e., $G'.\po=G.\po \cap (G'.\Events\times G'.\Events)$.
See \cref{subfig:bcs_gp_1} for an illustration.

\SubParagraph{2.~Recovering missing $\rf$-edges.}
The reads-from relation between all pairs of events that are present in both $G$ and $G'$ is preserved, i.e.,
if $\wt\LTo{G.\rf}\rd$ and $\wt,\rd\in G'.\Events$ then $\wt\LTo{G'.\rf}\rd$.
Now, consider an $\rf$-edge $\wt\LTo{G.\rf}\rd$ for which $\rd\in G'.\Events$ but $\wt\not \in G'.\Events$.
Then, we have to pair $\rd$ with a new writer in $G'$.
Since $\tuple{\CSei, \CSej}$ form a collapsible pair, we have $\wt\in \In_{\pi_{c}}(\CSei, \CSej)$, and thus $\ContextId_{\LinearTrace}(\rd)=c$.
Let $x=\lloc(\rd)$.
Since $\Summary_{\LinearTrace}(\CSei)=\Summary_{\LinearTrace}(\CSej)$, there is some value $v\in\ValueDom$ such that  $\Summary_{\LinearTrace}(\CSei).\SummaryLastValue(x)=\Summary_{\LinearTrace}(\CSej).\SummaryLastValue(x)=v$.
\Revision{Moreover, since $\wt\LTo{G.\hb}\rd$ and $\rd\not\in\In_{\pi_{c}}(\CSei, \CSej)$, the read $\rd$ follows $\CSej$ in $\pi_c$, and thus $\wt=\LatestWrite_{\LinearTrace}(\CSej, x)$; otherwise, some write $\wt'$ on $x$ with $\wt<_{\pi_c}\wt'\leq_{\pi_c}\CSej$ would satisfy $\wt\LTo{G.\mo}\wt'$ (by \ref{eq:wc}) and $\wt'\LTo{G.\po}\rd$, violating \ref{eq:rc} in $G$.}
\Revision{Hence, $\val_w(\wt)=v$.}
We set $\CSwi\LTo{G'.\rf}\rd$, where $\CSwi=\LatestWrite_{\LinearTrace}(\CSei, x)$.
See \cref{subfig:bcs_gp_2} for an illustration.

\SubParagraph{3.~Redirecting $\mo$-edges.}
With new $\rf$-edges in place, we may have to adapt $\mo$ in order to restore \ref{eq:rc}.
In particular, consider some location $x\in \LocationDom$.
If $\Summary_{\LinearTrace}(\CSei).\SummaryLastValue(x)=\bot$ or $x \in \Summary_{\LinearTrace}(\CSei).\SummaryExternalRf$, we set $G'.\mo_x=G.\mo_x\cap (G'.\Events \times G'.\Events)$.
Otherwise, let $\CSwi=\LatestWrite_{\LinearTrace}(\CSei, x)$ and $\CSwj=\LatestWrite_{\LinearTrace}(\CSej, x)$.
We set $G'.\mo_x=\ov{\mo_x}\cap (G'.\Events \times G'.\Events)$, where $\ov{\mo_x}$ is identical to $G.\mo_x$, with $\CSwi$ and $\CSwj$ swapped.
See \cref{subfig:bcs_gp_3} for an illustration. 

\newcommand\tikzrrX{0.6}
\newcommand\tikzrrY{0.8}

\begin{figure}
\begin{subfigure}[b]{.35\textwidth}
\centering
    \begin{tikzpicture}[ line width=1pt, main/.style = {circle, inner sep=1pt}] 

            \node[] (t10) at (0*\tikzrrX,\tikzrrY*0.7){$\Prog_{t_1}$};

            \node[main, fill=white, draw=black, text=black] (t11) at (0*\tikzrrX,\tikzrrY*0){\footnotesize$\InitState$};
            \node[main, fill=white, draw=black, text=black] (t12) at (0*\tikzrrX,\tikzrrY*-1){\footnotesize$q_{1}$};
            \node[main, fill=white, draw=black, text=black] (t13) at (0*\tikzrrX,\tikzrrY*-2){\footnotesize$q_{2}$};
            \node[main, fill=white, draw=black, text=black, double] (t14) at (0*\tikzrrX,\tikzrrY*-3){\footnotesize$\FinalState$};

            \draw [->] (t11) tonode [left=0pt]{\footnotesize{$\wt(\varix,1)$}} (t12);
            \draw [->] (t12) tonode [left=0pt]{\footnotesize{$\wt(\variy,1)$}} (t13);
            \draw [->] (t13) tonode [left=0pt]{\footnotesize{$\rd(\varix,1)$}} (t14);

            \node[] (t20) at (3*\tikzrrX,\tikzrrY*0.7){$\Prog_{t_2}$};

            \node[main, fill=white, draw=black, text=black] (t21) at (3*\tikzrrX,\tikzrrY*0){\footnotesize$\InitState$};
            \node[main, fill=white, draw=black, text=black] (t22) at (3*\tikzrrX,\tikzrrY*-1.5){\footnotesize$q_{1}$};
            \node[main, fill=white, draw=black, text=black] (t23) at (5*\tikzrrX,\tikzrrY*-1.5){\footnotesize$q_{2}$};
            \node[main, fill=white, draw=black, text=black, double] (t24) at (3*\tikzrrX,\tikzrrY*-3){\footnotesize$\FinalState$};
            
            \draw [->] (t21) tonode [left=0pt]{\footnotesize{$\wt(\varix,1)$}} (t22);
            \draw [->] (t22) to[bend left=20] node [above=-1pt]{\footnotesize{$\rd(\varix,1)$}} (t23);
            \draw [->] (t23) to[bend left=20] node [below=-1pt]{\footnotesize{$\wt(\varix,1)$}} (t22);
            \draw [->] (t22) tonode [left=0pt]{\footnotesize{$\rd(\variy,1)$}} (t24);

\end{tikzpicture}
\caption{$\ConcProg = \tuple{\Prog_{t_1},\Prog_{t_2}}$.}
\label{subfig:bcs_rem_1}
\end{subfigure}%
\begin{subfigure}[b]{.3\textwidth}
\centering
\begin{tikzpicture}[ line width=1pt, main/.style = {rectangle, inner sep=1pt}] 

\draw [->] (2*\tikzrrX,2*\tikzrrY) -- (2*\tikzrrX,-1.6*\tikzrrY) node [near start, above=0.75cm]{$t_1$}; 

\node[main, fill=white] at (2*\tikzrrX,1*\tikzrrY) (e11){\footnotesize{$\wt(\varix,1)$}}; 
\node[main, fill=white] at (2*\tikzrrX,0.5*\tikzrrY) (e12){\footnotesize{$\wt(\variy,1)$}}; 
\node[main, fill=white]  at (2*\tikzrrX,-0.5*\tikzrrY)(e13){\footnotesize{$\rd(\varix,1)$}}; 

\draw [->] (4.5*\tikzrrX,2*\tikzrrY) -- (4.5*\tikzrrX,-1.6*\tikzrrY) node [near start, above=0.75cm]{$t_2$};

\node[main, fill=white] at (4.5*\tikzrrX,1.5*\tikzrrY) (e21){\footnotesize{$\wt(\varix,1)$}}; 
\node[main, fill=white] at (4.5*\tikzrrX,1*\tikzrrY) (e22){\footnotesize{$\rd(\varix,1)$}}; 
\node[main, fill=white]  at (4.5*\tikzrrX,0.5*\tikzrrY)(e23){\footnotesize{$\wt(\varix,1)$}}; 
\node[main, fill=white] at (4.5*\tikzrrX,0*\tikzrrY) (e24){\footnotesize{$\rd(\varix,1)$}}; 
\node[main, fill=white] at (4.5*\tikzrrX,-0.5*\tikzrrY) (e25){\footnotesize{$\wt(\varix,1)$}}; 
\node[main, fill=white]  at (4.5*\tikzrrX,-1*\tikzrrY)(e26){\footnotesize{$\rd(\variy,1)$}}; 

\fill[ black, opacity = 0.1] (1.2*\tikzrrX,1.25*\tikzrrY) rectangle (2.8*\tikzrrX,0.25*\tikzrrY) {};
\node[] (t20) at (1.4*\tikzrrX,\tikzrrY*1.4){$\pi_{1}$};

\fill[ black, opacity = 0.1] (1.2*\tikzrrX,-0.25*\tikzrrY) rectangle (2.8*\tikzrrX,-0.75*\tikzrrY) {};
\node[] (t20) at (1.4*\tikzrrX,\tikzrrY*-0.1){$\pi_{3}$};

\fill[ black, opacity = 0.1] (3.8*\tikzrrX,1.75*\tikzrrY) rectangle (5.3*\tikzrrX,-1.25*\tikzrrY) {};
\node[] (t20) at (5.1*\tikzrrX,\tikzrrY*1.9){$\pi_{2}$};

\draw[mo] (e21.180) to (e11);
\draw[mo] (e11.-10) to (e23.180);
\draw [mo] (e23.10) to[bend left=40] (e25.0);
\draw[rf] (e11) to (e22);
\draw [rf] (e23.-10) to[bend left=20] (e24.0);
\draw[rf] (e25.180) to (e13.0);
\draw[rf] (e12.-30) to (e26.160);

\end{tikzpicture}
\caption{$\tuple{G, (\pi_1,\pi_2,\pi_3)}$.}
\label{subfig:bcs_rem_2}
\end{subfigure}%
\begin{subfigure}[b]{.3\textwidth}
\centering
\begin{tikzpicture}[ line width=1pt, main/.style = {rectangle, inner sep=1pt}] 

\draw [->] (2*\tikzrrX,2*\tikzrrY) -- (2*\tikzrrX,-1.6*\tikzrrY) node [near start, above=0.75cm]{$t_1$}; 

\node[main, fill=white] at (2*\tikzrrX,1*\tikzrrY) (e11){\footnotesize{$\wt(\varix,1)$}}; 
\node[main, fill=white] at (2*\tikzrrX,0.5*\tikzrrY) (e12){\footnotesize{$\wt(\variy,1)$}}; 
\node[main, fill=white]  at (2*\tikzrrX,-0.5*\tikzrrY)(e13){\footnotesize{$\rd(\varix,1)$}}; 

\draw [->] (4.5*\tikzrrX,2*\tikzrrY) -- (4.5*\tikzrrX,-1.6*\tikzrrY) node [near start, above=0.75cm]{$t_2$};

\node[main, fill=white] at (4.5*\tikzrrX,1.5*\tikzrrY) (e21){\footnotesize{$\wt(\varix,1)$}}; 
\node[main, fill=white] at (4.5*\tikzrrX,0*\tikzrrY) (e24){\footnotesize{$\rd(\varix,1)$}}; 
\node[main, fill=white] at (4.5*\tikzrrX,-0.5*\tikzrrY) (e25){\footnotesize{$\wt(\varix,1)$}}; 
\node[main, fill=white]  at (4.5*\tikzrrX,-1*\tikzrrY)(e26){\footnotesize{$\rd(\variy,1)$}}; 

\fill[ black, opacity = 0.1] (1.2*\tikzrrX,1.25*\tikzrrY) rectangle (2.8*\tikzrrX,0.25*\tikzrrY) {};
\node[] (t20) at (1.4*\tikzrrX,\tikzrrY*1.4){$\pi'_{1}$};

\fill[ black, opacity = 0.1] (1.2*\tikzrrX,-0.25*\tikzrrY) rectangle (2.8*\tikzrrX,-0.75*\tikzrrY) {};
\node[] (t20) at (1.4*\tikzrrX,\tikzrrY*-0.1){$\pi'_{3}$};

\fill[ black, opacity = 0.1] (3.8*\tikzrrX,1.75*\tikzrrY) rectangle (5.3*\tikzrrX,-1.25*\tikzrrY) {};
\node[] (t20) at (5.1*\tikzrrX,\tikzrrY*1.9){$\pi'_{2}$};

\draw[mo] (e11) to (e21.180);
\draw [mo] (e21.-10) to[bend left=30] (e25.0);
\draw [rf] (e21.-20) to[bend left=20] (e24.0);
\draw[rf] (e25.180) to (e13.0);
\draw[rf] (e12.-30) to (e26.160);

\end{tikzpicture}
\caption{$\tuple{G', (\pi'_1,\pi'_2,\pi'_3)}$}
\label{subfig:bcs_rem_3}
\end{subfigure}%
\caption{
The graph $G\in \ExecutionGraphsOf{\ConcProg}{\ramm}$ can be reduced to $G'$ while still reaching the final state of $\ConcProg$. This is done without increasing the number of contexts of $G$, by simply removing events in $G$ and changing $G.\mo$ and $G.\rf$.}
\label{fig:bcs_red}
\end{figure}

\cref{fig:bcs_red} shows a concrete example of witness reduction.
Since $\Summary_{\LinearTrace}(\CSei).\SummaryStates=\Summary_{\LinearTrace}(\CSej).\SummaryStates$, it follows easily that $G'$ represents a valid execution of $\ConcProg$ that reaches the same set of states as $G$.
The intricate part of the proof is in showing that $G'$ 
is Release/Acquire-consistent, and thus $G'\in \ExecutionGraphsOf{\ConcProg}{}$.

\Paragraph{The consistency of $G'$.}
We argue that $G'$ satisfies each of the Release/Acquire axioms of \cref{fig:axioms_ra}, besides \ref{eq:at}, since $G'$ is RMW-free.

\Paragraph{\ref{eq:hbirr}.}
Note that $G'.\po\subseteq G.\po$, while $G'.\rf\setminus G.\rf\subseteq G'.\po$.
This implies that $G'.\hb\subseteq G.\hb$, and since $G.\hb$ is irreflexive, the same holds for $G'.\hb$.

\Paragraph{\ref{eq:wc}.}
Since $G'.\hb\subseteq G.\hb$, a potential violation of \ref{eq:wc} in $G'$ must include an edge from $(G'.\mo \setminus G.\mo)$.
Any such edge has the form $\wt\LTo{G'.\mo}\CSwi$, where $\CSwi=\LatestWrite_{\LinearTrace}(\CSei, x)$ for some location $x$.
By construction, this implies
\begin{enumerate*}[label=(\roman*)]
\item $\CSwi\LTo{G.\mo}\wt$, and
\item $\wt\LTo{G.\mo}\CSwj$,
\end{enumerate*}
where $\CSwj = \LatestWrite_{\LinearTrace}(\CSej, x)$ and $x=\lloc(\wt)=\lloc(\CSwi)$. 
Note that this ordering is only possible if $\tid(\wt)\neq \tid(\CSwi)$, \Revision{as, by \ref{eq:wc}, any write of the same thread that is $\mo$-between $\CSwi$ and $\CSwj$ is also $\po$-between them, and} all writes on $x$ in the same thread between $\CSwi$ (exclusive) and $\CSwj$ (inclusive) have been removed in $G'$.
Since $G$ satisfies \ref{eq:wc}, we have $\CSwj \nLTo{G.\hb} \wt$. 
Moreover, since $\tuple{\CSei, \CSej}$ is collapsible, for any event $\event$ such that $\tid(\event)\neq \tid(\CSwi)$, we have that $\CSwi\LTo{\hb} \event$ iff $\CSwj\LTo{\hb} \event$. 
Therefore, $\CSwj \nLTo{G.\hb} \wt$ implies $\CSwi \nLTo{G.\hb} \wt$. 
Since $G'.\hb\subseteq G.\hb$, we also have $\CSwi \nLTo{G'.\hb} \wt$, and hence the edge $\wt\LTo{G'.\mo}\CSwi$ does not participate in a \ref{eq:wc} violation in $G'$.

\Paragraph{\ref{eq:rc}.} 
A violation of \ref{eq:rc} in $G'$ is witnessed by a triplet of distinct events $\tuple{\wt,\rd, \wt'}$ accessing the same location $x$ and such that
\begin{enumerate*}[label=(\roman*)]
\item\label{item:rc_violation_rf} $\wt\LTo{G'.\rf}\rd$, 
\item\label{item:rc_violation_hb} $\wt' \LTo{G'.\hb} \rd$,  and
\item\label{item:rc_violation_mo} $\wt \LTo{G'.\mo} \wt'$.
\end{enumerate*}
We argue that any triplet $\tuple{\wt,\rd, \wt'}$ fails one of \cref{item:rc_violation_rf,item:rc_violation_hb,item:rc_violation_mo}.
We consider two cases, depending on whether $\wt\LTo{G.\rf}\rd$, i.e., whether the writer of $\rd$ is the same in $G$ and $G'$.

\newcommand\tikzrcX{0.6}
\newcommand\tikzrcY{0.8}

\begin{figure}
\begin{subfigure}[t]{.4\textwidth}
\centering
    \begin{tikzpicture}[ line width=1pt, main/.style = {rectangle, inner sep=1.5pt}] 

\draw [->] (1.5*\tikzrcX,2*\tikzrcY) -- (1.5*\tikzrcX,-2*\tikzrcY) node [near start, above=1cm]{}; 

\draw [->] (5.5*\tikzrcX,2*\tikzrcY) -- (5.5*\tikzrcX,-2*\tikzrcY) node [near start, above=1cm]{}; 

\draw [densely dashed, gray,  line width = 0.4pt] (3*\tikzrcX,2.5*\tikzrcY) to node [at start, above=0.2cm, black] {$G \to G'$}(3*\tikzrcX,-2*\tikzrcY);

\node[main, fill=white] at (1.5*\tikzrcX,2.5*\tikzrcY) (pi1){$\pi$}; 
\node[main, fill=white] at (5.5*\tikzrcX,2.5*\tikzrcY) (pi2){$\pi'$}; 

\node[main, fill=white] at (1.5*\tikzrcX,1.5*\tikzrcY) (r11){\small$\rd^{\ref{item:vio_1}}$}; 
\node[main, fill=white] at (1.5*\tikzrcX,1*\tikzrcY) (w12){\small$\CSwi$}; 
\node[main, fill=white] at (1.5*\tikzrcX,0.5*\tikzrcY) (r12){\small$\rd^{\ref{item:vio_2}}$}; 
\node[main, fill=white, inner sep=2pt]  at (1.5*\tikzrcX,0*\tikzrcY)(aw4c1){\small$\dots$}; 
\node[main, fill=white]  at (1.5*\tikzrcX,-0.5*\tikzrcY) (w11){\small$\CSwj$};
\node[main, fill=white, inner sep=2pt]  at (1.5*\tikzrcX,-1*\tikzrcY) (aw4c2){\small$\dots$};  
\node[main, fill=white] at (1.5*\tikzrcX,-1.5*\tikzrcY) (r13){\small$\rd^{\ref{item:vio_3}}$};

\node[main, fill=white, inner sep=2.5pt] at (-0.5*\tikzrcX,0.75*\tikzrcY) (wa12){\small$\wt$}; 

\draw[mo] (wa12) to (w11.west);
\draw[mo] (w12.west) to (wa12);
\draw[rf] (wa12) to (r11.west);
\draw[rf] (wa12) to (r12.west);
\draw[rf] (wa12) to (r13.west);

\fill[ red, opacity = 0.2] (1*\tikzrcX,0.25*\tikzrcY) rectangle (2*\tikzrcX,-1.25*\tikzrcY);

\node[main, fill=white] at (5.5*\tikzrcX,1.5*\tikzrcY) (r21){\small$\rd^{\ref{item:vio_1}}$}; 
\node[main, fill=white] at (5.5*\tikzrcX,1*\tikzrcY) (w22){\small$\CSwi$}; 
\node[main, fill=white] at (5.5*\tikzrcX,0.5*\tikzrcY) (r22){\small$\rd^{\ref{item:vio_2}}$}; 
\node[main, fill=white] at (5.5*\tikzrcX,-1.5*\tikzrcY) (r23){\small$\rd^{\ref{item:vio_3}}$};

\node[main, fill=white, inner sep=2.5pt] at (3.5*\tikzrcX,0.75*\tikzrcY) (wa22){\small$\wt$}; 

\draw[mo] (wa22) to (w22.west);
\draw[rf] (wa22) to (r21.west);
\draw[rf] (wa22) to (r22.west);
\draw[rf] (wa22) to (r23.west);

\end{tikzpicture}
\caption{The case of $\wt\LTo{G.\rf}\rd$.}
\label{subfig:bcs_rc_1}
\end{subfigure}
\hspace{5pt}
\begin{subfigure}[t]{.4\textwidth}
\centering
    \begin{tikzpicture}[ line width=1pt, main/.style = {rectangle, inner sep=1.5pt}] 

\draw [->] (1.5*\tikzrcX,2*\tikzrcY) -- (1.5*\tikzrcX,-2*\tikzrcY) node [near start, above=1cm]{}; 
\draw [->] (5.5*\tikzrcX,2*\tikzrcY) -- (5.5*\tikzrcX,-2*\tikzrcY) node [near start, above=1cm]{}; 

\draw [densely dashed, gray,  line width = 0.4pt] (3*\tikzrcX,2.5*\tikzrcY) to node [at start, above=0.2cm, black] {$G \to G'$}(3*\tikzrcX,-2*\tikzrcY);

\node[main, fill=white] at (1.5*\tikzrcX,2.5*\tikzrcY) (pi1){$\pi$}; 
\node[main, fill=white] at (5.5*\tikzrcX,2.5*\tikzrcY) (pi2){$\pi'$}; 

\node[main, fill=white] at (1.5*\tikzrcX,1.5*\tikzrcY) (w12){\small$\CSwi$}; 
\node[main, fill=white, inner sep=2pt] at (1.5*\tikzrcX,1*\tikzrcY) (r12){\small$\dots$}; 
\node[main, fill=white, inner sep=2pt]  at (1.5*\tikzrcX,0.5*\tikzrcY)(aw4c1){\small$\dots$}; 
\node[main, fill=white]  at (1.5*\tikzrcX,0*\tikzrcY) (w11){\small$\CSwj$};
\node[main, fill=white, inner sep=2pt]  at (1.5*\tikzrcX,-0.5*\tikzrcY) (aw4c2){\small$\dots$};  
\node[main, fill=white, inner sep=2pt] at (1.5*\tikzrcX,-1*\tikzrcY) (r11){\small$\dots$}; 
\node[main, fill=white] at (1.5*\tikzrcX,-1.5*\tikzrcY) (r13){\small$\CSrj$};

\node[main, fill=white] at (0*\tikzrcX,0.75*\tikzrcY) (wa12){\small$\CSwp$}; 

\draw[mo]  (wa12) to (w11.west);
\draw[mo] (w12.west) to (wa12);
\draw[hb] (wa12) to (r13.west);
\draw [rf] (w11.0) to[bend left=30] (r13.0);

\fill[ red, opacity = 0.2] (1*\tikzrcX,0.75*\tikzrcY) rectangle (2*\tikzrcX,-0.75*\tikzrcY);

\node[main, fill=white] at (5.5*\tikzrcX,1.5*\tikzrcY) (w22){\small$\CSwi$}; 
\node[main, fill=white, inner sep=2pt] at (5.5*\tikzrcX,1*\tikzrcY) (r22){\small$\dots$}; 
\node[main, fill=white, inner sep=2pt] at (5.5*\tikzrcX,-1*\tikzrcY) (r21){\small$\dots$}; 
\node[main, fill=white] at (5.5*\tikzrcX,-1.5*\tikzrcY) (r23){\small$\CSrj$};

\node[main, fill=white] at (4*\tikzrcX,0.75*\tikzrcY) (wa22){\small$\CSwp$}; 

\draw[mo]  (wa22) to (w22.west);
\draw[hb] (wa22) to (r23.west);
\draw [rf] (w22.0) to[bend left=20] (r23.0);

\end{tikzpicture}
\caption{The case of $\wt\LTo{G'.\rf\setminus G.\rf}\rd$.}
\label{subfig:bcs_rc_2}
\end{subfigure}%
\caption{
The two cases in the proof of \ref{eq:rc} for $G'$ on a triplet $\tuple{\wt,\rd,\wt'}$.
}
\label{fig:bcs_rc}
\end{figure}
    
\SubParagraph{The case of $\wt\LTo{G.\rf}\rd$.} 
Recall that $G'.\hb \subseteq G.\hb$, hence if $\wt'\LTo{G'.\hb} \rd$, we also have $\wt' \LTo{G.\hb} \rd$.
Since also $\wt\LTo{G.\rf}\rd$ and $G$ satisfies \ref{eq:rc},
we have $\wt' \LTo{G.\mo } \wt$.
Assume towards contradiction that $\wt \LTo{G'.\mo } \wt'$, and by construction, we have $\wt'=\CSwi$, where $\CSwi=\LatestWrite_{\LinearTrace}(\CSei, x)$.
Let $\pi'=\pi'_1\cdot\pi'_2\cdots \pi'_{\ell}$ (recall that $\tuple{G', (\pi'_1, \pi'_2, \dots, \pi'_{\ell})}$ is the trace we have constructed).
We consider where $\rd$ appears in $\pi'$, relatively to $\CSwi$ and $\CSei$ (see also \cref{subfig:bcs_rc_1}).
\begin{compactenum}
\item $\rd <_{\pi'} \CSwi$.
This is not possible, since $\pi'$ is a total extension of $G'.\hb$ and we have $\CSwi\LTo{G'.\hb} \rd$.\label{item:vio_1}

\item $\CSwi <_{\pi'} \rd \leq_{\pi'} \CSei$.
Note that, in this case, $\ContextId_{\LinearTrace}(\rd)=c$, i.e., $\rd$ belongs to the same thread as $\CSwi$ and $\CSei$, and thus $\CSwi\LTo{G.\po} \rd$.
Since $\wt \LTo{G'.\mo\setminus G.\mo} \CSwi$, the condition for redirecting $\mo$-edges for location $x$ must have been triggered in the reduction of $\tuple{G, (\pi_1, \pi_2, \dots, \pi_{\ell})}$. 
This is only possible if $x \not\in \Summary_{\LinearTrace}(\CSei).\SummaryExternalRf$.
This, in turn, implies that $\rd$ reads from a local write in $G$, i.e., $\wt=\CSwi$, contradicting the requirement that $\wt$ and $\CSwi$ are distinct.\label{item:vio_2}

\item $\CSej <_{\pi'} \rd$. 
Since $\wt \LTo{G'.\mo\setminus G.\mo } \wt'$, we have $\wt\LTo{G.\mo}\CSwj$, where $\CSwj = \LatestWrite_{\LinearTrace}(\CSej, x)$.
Since $G$ satisfies \ref{eq:rc}, we have $\CSwj \nLTo{G.\hb} \rd$. 
Since $\CSwj <_{\pi} \CSej <_{\pi} \rd$, this implies $\tid(\rd) \neq \tid(\CSwj)$, or equivalently, $\tid(\rd) \neq \tid(\CSwi)$, as otherwise we would have $\CSwj \LTo{G.\po} \rd$ and thus $\CSwj \LTo{G.\hb} \rd$.
Since $\tuple{\CSei, \CSej}$ is collapsible, for any event $\event$ with $\tid(\event)\neq \tid(\CSwi)$, we have $\CSwi\LTo{G.\hb} \event$ iff $\CSwj\LTo{G.\hb} \event$. 
Therefore, $\CSwj \nLTo{G.\hb} \rd$ implies $\CSwi \nLTo{G.\hb} \rd$. 
Since $G'.\hb \subseteq G.\hb$, we obtain $\CSwi \nLTo{G'.\hb} \rd$.\label{item:vio_3}
\end{compactenum}
Hence, every triplet $\tuple{\wt,\rd, \wt'}$ of  $G'$ with  $\wt\LTo{G.\rf}\rd$ satisfies \ref{eq:rc}. 

\SubParagraph{The case of $\wt\LTo{G'.\rf\setminus G.\rf}\rd$.}
Any edge in $(G'.\rf \setminus G.\rf)$ is of the form $\CSwi \LTo{G'.\rf} \CSrj$, where $\CSwi = \LatestWrite_{\LinearTrace}(\CSei, x)$ and $\LatestWrite_{\LinearTrace}(\CSej,x) \LTo{G.\rf} \CSrj$.
We thus have $\wt=\CSwi$ and $\rd= \CSrj$.
Moreover, let $\CSwj=\LatestWrite_{\LinearTrace}(\CSej,x)$ be the original writer of $\CSrj$ in $G$.
Since $G'.\hb\subseteq G.\hb$, if $\wt'\LTo{G'.\hb}\CSrj$, we also have $\wt'\LTo{G.\hb}\CSrj$, and as $G$ satisfies \ref{eq:rc}, we have $\wt'\LTo{G.\mo}\CSwj$. 
We argue that $\Summary_{\LinearTrace}(\CSei).\SummaryLastValue(x)\neq\bot$ and $x \not\in \Summary_{\LinearTrace}(\CSei).\SummaryExternalRf$, which are the conditions under which $\CSwi$ takes the place of $\CSwj$ in $G'.\mo$, establishing that $\wt'\LTo{G'.\mo}\CSwi$, as illustrated in \cref{subfig:bcs_rc_2}.

The presence of $\CSwj$ establishes that $\Summary_{\LinearTrace}(\CSei).\SummaryLastValue(x)\neq\bot$.
Moreover, if $x \in \Summary_{\LinearTrace}(\CSei).\SummaryExternalRf$, there would be a read event $\rd''$, also on location $x$, such that
\begin{enumerate*}[label=(\roman*)]
\item $\CSwj\LTo{G.\po} \rd'' \LTo{G.\po} \CSej \LTo{G.\po} \CSrj$, and
\item $\wt'' \LTo{G.\rf} \rd''$, where $\wt''\neq \wt_j$.
\end{enumerate*}
This, however, would imply a violation of \ref{eq:rc} for $G$, witnessed either  by the triplet $\tuple{\CSwj, \CSrj, \wt''}$ (if $\CSwj \LTo{G.\mo} \wt''$), or $\tuple{\wt'',  \rd'', \CSwj}$ (if $\wt'' \LTo{G.\mo} \CSwj$).

Hence, every triplet $\tuple{\wt,\rd, \wt'}$ of  $G'$ with  $\wt\LTo{G'.\rf\setminus G.\rf}\rd$ also satisfies \ref{eq:rc}, concluding the proof of the consistency of $G'$.
\subsection{Small Models with Bounded RMWs}\label{SUBSEC:DECIDABILITY_RMWS}

Finally, we address the case of reachability as witnessed by traces that contain both bounded context switches and bounded RMWs, thereby concluding \cref{thm:bounded_context_switches}.

\newcommand\tikzrmwX{0.6}
\newcommand\tikzrmwY{0.65}

\begin{figure}
\begin{subfigure}[b]{.4\textwidth}
\centering
    \begin{tikzpicture}[ line width=1pt, main/.style = {rectangle, inner sep=2pt}] 

\draw [->] (1.5*\tikzrmwX,2*\tikzrmwY) -- (1.5*\tikzrmwX,-2*\tikzrmwY) node [near start, above=1cm]{}; 
\draw [->] (5.5*\tikzrmwX,2*\tikzrmwY) -- (5.5*\tikzrmwX,-2*\tikzrmwY) node [near start, above=1cm]{}; 

\draw [densely dashed, gray,  line width = 0.4pt] (3*\tikzrmwX,2.5*\tikzrmwY) to node [at start, above=0.2cm, black] {$G \to G'$}(3*\tikzrmwX,-2*\tikzrmwY);

\node[main, fill=white] at (1.5*\tikzrmwX,2.5*\tikzrmwY) (pi1){$\pi$}; 
\node[main, fill=white] at (5.5*\tikzrmwX,2.5*\tikzrmwY) (pi2){$\pi'$}; 

\node[main, fill=white] at (1.5*\tikzrmwX,1.5*\tikzrmwY) (wk1){\small$\CSwi$}; 
\node[main, fill=white, inner sep=2pt] at (1.5*\tikzrmwX,1*\tikzrmwY) (r12){\small$\dots$}; 
\node[main, fill=white, inner sep=2pt]  at (1.5*\tikzrmwX,0.5*\tikzrmwY)(aw4c1){\small$\dots$}; 
\node[main, fill=white]  at (1.5*\tikzrmwX,0*\tikzrmwY) (wr1){\small$\CSwj$};
\node[main, fill=white, inner sep=2pt]  at (1.5*\tikzrmwX,-0.5*\tikzrmwY) (aw4c2){\small$\dots$};  
\node[main, fill=white, inner sep=2pt] at (1.5*\tikzrmwX,-1*\tikzrmwY) (r11){\small$\dots$}; 
\node[main, fill=white] at (1.5*\tikzrmwX,-1.5*\tikzrmwY) (r13){\small$\CSrj$};

\node[main, fill=white] at (0*\tikzrmwX,0.75*\tikzrmwY) (wa1){\small$\CSwp$}; 

\draw[mo]  (wa1) to (wr1.north west);
\draw[mo] (wk1.west) to (wa1);
\draw[hb] (wa1) to (r13.west);
\draw [rf] (wr1.0) to[bend left=30] (r13.0);

\fill[ red, opacity = 0.2] (1*\tikzrmwX,0.75*\tikzrmwY) rectangle (2*\tikzrmwX,-0.75*\tikzrmwY);

\node[main, fill=white] at (5.5*\tikzrmwX,1.5*\tikzrmwY) (wk2){\small$\CSwi$}; 
\node[main, fill=white, inner sep=2pt] at (5.5*\tikzrmwX,1*\tikzrmwY) (r22){\small$\dots$}; 
\node[main, fill=white, inner sep=2pt] at (5.5*\tikzrmwX,-1*\tikzrmwY) (r21){\small$\dots$}; 
\node[main, fill=white] at (5.5*\tikzrmwX,-1.5*\tikzrmwY) (r23){\small$\CSrj$};

\node[main, fill=white] at (4*\tikzrmwX,0.75*\tikzrmwY) (wa2){\small$\CSwp$}; 

\draw[mo]  (wa2) to (wk2.west);
\draw[hb] (wa2) to (r23.west);
\draw [rf] (wk2.0) to[bend left=20] (r23.0);

\end{tikzpicture}
\caption{Changing $\mo$ in $G'$.}
\label{subfig:bcs_rmw_1}
\end{subfigure}%
\hspace{11pt}
\begin{subfigure}[b]{.4\textwidth}
\centering
    \begin{tikzpicture}[ line width=1pt, main/.style = {rectangle, inner sep=2pt}] 

\draw [->] (1.5*\tikzrmwX,2*\tikzrmwY) -- (1.5*\tikzrmwX,-2*\tikzrmwY) node [near start, above=1cm]{}; 
\draw [->] (6*\tikzrmwX,2*\tikzrmwY) -- (6*\tikzrmwX,-2*\tikzrmwY) node [near start, above=1cm]{}; 

\draw [densely dashed, gray,  line width = 0.4pt] (3*\tikzrmwX,2.5*\tikzrmwY) to node [at start, above=0.2cm, black] {$G \to G'$}(3*\tikzrmwX,-2*\tikzrmwY);

\node[main, fill=white] at (1.5*\tikzrmwX,2.5*\tikzrmwY) (pi1){$\pi$}; 
\node[main, fill=white] at (6*\tikzrmwX,2.5*\tikzrmwY) (pi2){$\pi'$}; 

\node[main, fill=white] at (1.5*\tikzrmwX,1.5*\tikzrmwY) (wk1){\small$\CSrmwi$}; 
\node[main, fill=white, inner sep=2pt] at (1.5*\tikzrmwX,1*\tikzrmwY) (r12){\small$\dots$}; 
\node[main, fill=white, inner sep=2pt]  at (1.5*\tikzrmwX,0.5*\tikzrmwY)(aw4c1){\small$\dots$}; 
\node[main, fill=white]  at (1.5*\tikzrmwX,0*\tikzrmwY) (wr1){\small$\CSwj$};
\node[main, fill=white, inner sep=2pt]  at (1.5*\tikzrmwX,-0.5*\tikzrmwY) (aw4c2){\small$\dots$};  
\node[main, fill=white, inner sep=2pt] at (1.5*\tikzrmwX,-1*\tikzrmwY) (r11){\small$\dots$}; 
\node[main, fill=white] at (1.5*\tikzrmwX,-1.5*\tikzrmwY) (r13){\small$\CSrj$};

\node[main, fill=white] at (-0.5*\tikzrmwX,1.5*\tikzrmwY) (wb1){\small$\CSwpp$}; 
\node[main, fill=white] at (-0.5*\tikzrmwX,0*\tikzrmwY) (wa1){\small$\CSwp$}; 

\draw[rf] (wb1) to (wk1.180);
\draw[mo]  (wb1) to (wa1);
\draw[mo] (wk1.south west) to (wa1.north east);
\draw[hb] (wa1) to (r13.west);
\draw [rf] (wr1.0) to[bend left=30] (r13.0);

\fill[ red, opacity = 0.2] (1*\tikzrmwX,0.75*\tikzrmwY) rectangle (2*\tikzrmwX,-0.75*\tikzrmwY);

\node[main, fill=gray!30] at (6*\tikzrmwX,1.5*\tikzrmwY) (wk2){\small$\CSrmwi$}; 
\node[main, fill=white, inner sep=2pt] at (6*\tikzrmwX,1*\tikzrmwY) (r22){\small$\dots$}; 
\node[main, fill=white, inner sep=2pt] at (6*\tikzrmwX,-1*\tikzrmwY) (r21){\small$\dots$}; 
\node[main, fill=gray!30] at (6*\tikzrmwX,-1.5*\tikzrmwY) (r23){\small$\CSrj$};

\node[main, fill=white] at (4*\tikzrmwX,1.5*\tikzrmwY) (wb2){\small$\CSwpp$}; 
\node[main, fill=gray!30] at (4*\tikzrmwX,0*\tikzrmwY) (wa2){\small$\CSwp$}; 

\draw[rf] (wb2) to (wk2.180);
\draw[mo]  (wb2) to (wa2);
\draw[mo]  (wa2) to (wk2.south west);
\draw[hb] (wa2) to (r23.west);
\draw [rf] (wk2.south east) to[bend left=15] (r23.north east);

\end{tikzpicture}
\caption{An \ref{eq:at} violation on $\CSrmwi$.}
\label{subfig:bcs_rmw_2}
\end{subfigure}%
\caption{
RMWs prevent the collapsibility of segments between repeating summaries.
}
\label{fig:bcs_rmw}
\end{figure}

\Paragraph{Where does the proof above break with RMWs?}
To gain proof intuition, it is natural to ask why our construction in the previous section fails in the presence of RMWs. 
\cref{fig:bcs_rmw} illustrates this case.
Recall that our construction works as follows:~if there is a segment between  two collapsible events $\CSei$ and $\CSej$ containing a write $\CSwj$ that is read by a read $\CSrj$ after $\CSej$, in $G'$ we identify another write $\CSwi$ which
\begin{enumerate*}[label=(\roman*)]
\item becomes the writer of $\CSrj$, and
\item receives all the $\mo$ edges of $\CSwj$ in $G$, in order to restore \ref{eq:rc} for $\CSrj$ relative to other conflicting writes $\wt'$ for which $\wt'\LTo{G'.\hb} \CSrj$ (see \cref{subfig:bcs_rmw_1}).
\end{enumerate*}
Now suppose $\CSwi$ is replaced by an RMW $\CSrmwi$,
reading from another write $\CSwpp$ for which also $\CSwpp\LTo{G'.\mo} \wt'$ (see \cref{subfig:bcs_rmw_2}).
Performing the same sequence of changes, namely $\CSrmwi \LTo{G'.\rf} \CSrj$ and $\CSwp \LTo{G'.\mo} \CSrmwi$, would violate \ref{eq:at}, as $\wt'$ is $\mo$-between $\CSrmwi$ and its writer.
Moreover, the $\mo$ order from $\CSwpp$ to $\CSwp$ might not be reversible, e.g., because $\CSwpp\LTo{G'.\hb} \wt'$.
In the following, we lift the proof of the previous section to handle executions with a bounded number of RMWs.

\Paragraph{Context-bounded reachability with RMWs.}
The problem of \emph{context-bounded reachability with bounded RMWs} asks whether, given a concurrent program $\ConcProg$ and parameters $k_c, k_{\rmw} \in \Nats$, there exists a trace $\LinearTrace = \tuple{G, (\pi_1, \pi_2,\dots, \pi_{\ell})}$ of $\ConcProg$ with $\ell \leq k_c$ and $|(G.E\cap \RMWDom)| \leq k_{\rmw}$ that reaches the final state $\tuple{\Prog_{t_1}.\FinalState, \Prog_{t_2}.\FinalState, \dots, \Prog_{t_{\numThreads}}.\FinalState}$ of $\ConcProg$.

\Paragraph{Reducible traces with RMWs.}
To complete our proof of \cref{thm:bounded_context_switches}, we extend the definition of reducible traces as follows.
First, we redefine $\LatestWrite_{\LinearTrace}(\CSei, x)$ to properly handle RMW events:~$\LatestWrite_{\LinearTrace}(\event,x)$ is now the maximal \emph{writing} (i.e., either a plain write or an RMW) event $\event'(t, x)$ with respect to $\leq_{\pi_{c}}$ such that $\event'(t,x)\leq_{\pi_{c}} \event$. 
Second, we deem two events $\CSei$ and $\CSej$ collapsible if, in addition to \cref{item:cond_red_1,item:cond_red_2,item:cond_red_3,item:cond_red_4}, they satisfy the following condition:
\begin{enumerate*}[label=(\roman*)]
\addtocounter{enumi}{4}
\item for all $x \in \LocationDom$, if $\LatestWrite_{\LinearTrace}(\CSei, x)\neq \LatestWrite_{\LinearTrace}(\CSej, x)$, then $\op(\LatestWrite_{\LinearTrace}(\CSei, x))=\wt$.\label{item:cond_red_5}
\end{enumerate*}
When reducing a trace, this condition prevents inserting $\rf$-edges from RMWs, and thus avoids the need to further insert $\mo$ edges to RMWs, thereby avoiding the offending case of \cref{subfig:bcs_rmw_2}.
The following lemma bounds the size of irreducible traces, generalizing \cref{lem:irreducible_small}.

\begin{restatable}{lemma}{lemmairreduciblesmallrmw}\label{lem:irreducible_small_rmw}
For any concurrent program $\ConcProg$ and trace $\LinearTrace=\tuple{G, (\pi_1, \dots, \pi_{\ell})}$ of $\ConcProg$,
if $\LinearTrace$ is irreducible then \Revision{$|G.\Events|=(|G.\Events \cap \RMWDom|+1)\cdot 2^{O(\ell\cdot (|\ConcProg.\States| + |\LocationDom|\cdot \log(|\ValueDom|+1)))}$}. 
\end{restatable}
Finally, given a reducible trace $\tuple{G, (\pi_1, \pi_2, \dots, \pi_{\ell})}$, we construct a trace $\tuple{G', (\pi'_1, \pi'_2, \dots, \pi'_{\ell})}$ by following the process described in \cref{SUBSEC:IRREDUCIDBLE_REACHABILITY}.
It remains to argue that $G'$ is consistent.

\Paragraph{The consistency of $G'$.}
The proof for \ref{eq:hbirr}, \ref{eq:wc} and \ref{eq:rc} remains the same as in \cref{SUBSEC:IRREDUCIDBLE_REACHABILITY}.
We further need to argue that $G'$ does not violate \ref{eq:at}. 
A violation is witnessed by a triplet of distinct events $\tuple{\wt, \rmw,\CSwp}$ such that 
\begin{enumerate*}[label=(\roman*)]
\item\label{item:at_violation_rf} $\wt \LTo{G'.\rf} \rmw$,
\item\label{item:at_violation_mo1} $\CSwp \LTo{G'.\mo} \rmw $, and
\item\label{item:at_violation_mo2} $\wt \LTo{G'.\mo} \CSwp$.
\end{enumerate*}
We argue that any triplet $\tuple{\wt, \rmw,\CSwp}$ that satisfies
\cref{item:at_violation_rf} and \cref{item:at_violation_mo1} is such that $\CSwp\LTo{G'.\mo} \wt$, i.e., it fails \cref{item:at_violation_mo2}. 
Due to \cref{item:cond_red_5} of collapsibility, no $\mo$-edges pointing to RMWs are introduced in $G'$, thus $\CSwp \LTo{G.\mo} \rmw$.
We consider two cases, depending on whether $\wt\LTo{G.\rf}\rmw$, i.e., whether the writer of $\rmw$ is the same in $G$ and $G'$.

\SubParagraph{The case of $\wt\LTo{G.\rf}\rmw$.} 
Since also $\CSwp\LTo{G.\mo}\rmw$ and $G$ satisfies \ref{eq:at}, we have $\CSwp \LTo{G.\mo} \wt$.
For the last relation to be reversed in $G'$, i.e., $\wt \LTo{G'.\mo} \CSwp$, it follows by construction that $\CSwp = \LatestWrite_{\LinearTrace}(\CSei, x)$ and there is some removed event $\CSwj = \LatestWrite_{\LinearTrace}(\CSej, x)$ such that $\wt\LTo{G.\mo} \CSwj$.
Since $G$ satisfies \ref{eq:at}, we have $\rmw\LTo{G.\mo} \CSwj$.
In this case, however, the construction also forces $\rmw \LTo{G'.\mo} \CSwp$, a contradiction.

\SubParagraph{The case of $\wt\LTo{G'.\rf\setminus G.\rf}\rmw$.} 
Observe that, in this case, we have $\wt = \LatestWrite_{\LinearTrace}(\CSei, x)$ and $\LatestWrite_{\LinearTrace}(\CSej,x) \LTo{G.\rf_x} \rmw$.  
Let $\CSwj= \LatestWrite_{\LinearTrace}(\CSej,x)$ be the original writer of $\rmw$. 
As already argued in the proof of \ref{eq:rc} of \cref{SUBSEC:IRREDUCIDBLE_REACHABILITY}, 
$\wt$ also obtains in $G'$ the $\mo$ predecessors of $\CSwj$ in $G$.
Since $\tuple{\CSwj, \rmw, \CSwp}$ does not violate \ref{eq:at} in $G$, we have $\CSwp \LTo{G.\mo} \CSwj$, and thus \Revision{$\CSwp \LTo{G'.\mo} \wt$}.


\Paragraph{\Revision{Complexity.}}
\Revision{
The exponential size bound of \cref{lem:irreducible_small_rmw} yields an $\NEXPTIME$ decision procedure, by guessing an exponentially-sized witness and verifying its consistency in time polynomial in its size.
}
\Revision{
On the other hand, since bounded PCP is $\NP$-hard, our reduction behind \cref{thm:undecidability} yields that
reachability under bounded context switches alone is $\NP$-hard.
}
\Revision{Note that the $\NEXPTIME$ upper bound requires that}
\begin{enumerate*}[label=(\roman*)]
\item \Revision{$k_c$ and $k_{\rmw}$ are given in unary, and}
\item \Revision{$\ConcProg$ is given in LTS form, rather than as code, whose translation to LTSs incurs an exponential blowup.}
\end{enumerate*}
\Revision{Under succinct representations, i.e., with $k_c$ and $k_{\rmw}$ in binary and $\ConcProg$ given as code, the complexity bound becomes $\TWONEXPTIME$.}

\section{Related Work}\label{SEC:RELATED_WORK}

The decidability of verification, typically instantiated in its most basic form of state reachability, has been a subject of systematic study under weak memory.
In~\cite{Atig2010,Atig2012}, it was shown that $\tsomm$ and $\psomm$ admit decidable verification, but the problem becomes undecidable when allowing read to read/write reorderings.
More recently, verification was shown to be undecidable under Power~\cite{Abdulla2021b} and $\rdmamm$~\cite{Abdulla26}.
Variants of Release/Acquire, namely, Strong/Weak/Localized Release/Acquire were recently shown to admit decidable verification~\cite{Lahav2022,Singh2024}.
This indicates that the undecidability of standard Release/Acquire depends intricately on the specifics of this model.

A significant body of work has focused on analyzing concurrent programs under some sort of context bounding.
This setting has been practically utilized for programs under Sequential Consistency~\cite{Qadeer2005,Musuvathi2007,Torre2009}, and its complexity has been well-characterized~\cite{Esparza2014,Chini2017,Baumann2020}.
Bounded context switches have also been employed  in TSO~\cite{Atig2011,Atig2021} and Power~\cite{Abdulla2017}.
Since undecidability under Release/Acquire with RMWs already holds for executions with 4 contexts, \cite{Abdulla19} developed the stronger restriction of ``view switches'', and showed that \Revision{bounding them recovers decidability.
Recently, it was shown that the wider class of bounded-treewidth executions also admits decidable reachability}~\cite{Conrado26}.
\section{Conclusion}\label{SEC:CONCLUSION}
We have proven that the reachability problem of concurrent programs under Release/Acquire semantics is undecidable even for RMW-free programs, resolving an open question of~\cite{Abdulla19}.
Moreover, we have shown that RMW-free programs enjoy decidable verification under bounded context switches, in contrast to programs with RMWs, for which undecidability already holds with 4-context executions.
\Revision{Interesting future work includes establishing tight complexity bounds for context-bounded reachability, and realizing them in practical model-checking tooling}.


\begin{acks}
This work was partially supported by research grants (VIL42117) from VILLUM FONDEN and  5251-00076B from the Independent Research Foundation of Denmark.
We also thank anonymous reviewers of POPL~'27 for helpful suggestions on an earlier draft.
\end{acks}

\section*{\Revision{Disclosure of Use of AI in Research}}

\Revision{Claude AI was solely used in the Lean formalization of \cref{thm:undecidability}. 
This formalization was done post-factum, with basis on the original hand-written proofs in the paper. 
The authors take full responsibility for this paper, including the Lean formalization.}

\section*{\Revision{Data Availability Statement}}

\Revision{The Lean formalization of \cref{thm:undecidability} is publicly available at} \url{https://github.com/kobusgiovanna/decidability-ra}. 


\bibliographystyle{ACM-Reference-Format}
\bibliography{bibliography}

@article{Musuvathi2007,
author = {Musuvathi, Madanlal and Qadeer, Shaz},
title = {Iterative context bounding for systematic testing of multithreaded programs},
year = {2007},
issue_date = {June 2007},
publisher = {Association for Computing Machinery},
address = {New York, NY, USA},
volume = {42},
number = {6},
issn = {0362-1340},
url = {https://doi.org/10.1145/1273442.1250785},
doi = {10.1145/1273442.1250785},
journal = {SIGPLAN Not.},
month = jun,
pages = {446–455},
numpages = {10}
}

@inproceedings{Torre2009,
author = {Torre, Salvatore and Madhusudan, P. and Parlato, Gennaro},
title = {Reducing Context-Bounded Concurrent Reachability to Sequential Reachability},
year = {2009},
isbn = {9783642026577},
publisher = {Springer-Verlag},
address = {Berlin, Heidelberg},
url = {https://doi.org/10.1007/978-3-642-02658-4_36},
doi = {10.1007/978-3-642-02658-4_36},
booktitle = {Proceedings of the 21st International Conference on Computer Aided Verification},
pages = {477–492},
numpages = {16},
location = {Grenoble, France},
series = {CAV '09}
}

@InProceedings{Atig2012,
author="Atig, Mohamed Faouzi
and Bouajjani, Ahmed
and Burckhardt, Sebastian
and Musuvathi, Madanlal",
editor="Seidl, Helmut",
title="What's Decidable about Weak Memory Models?",
booktitle="Programming Languages and Systems",
year="2012",
publisher="Springer Berlin Heidelberg",
address="Berlin, Heidelberg",
pages="26--46",
isbn="978-3-642-28869-2"
}

@misc{Dijkstra1962,
  author       = {Edsger W. Dijkstra},
  title        = {About the Sequentiality of Process Descriptions},
  howpublished = {EWD 35 (English translation)},
  year         = {1962},
  note         = {Translated by Martien van der Burgt and Heather Lawrence, available at \url{https://www.cs.utexas.edu/~EWD/translations/EWD35-English.html}},
  url          = {https://www.cs.utexas.edu/~EWD/translations/EWD35-English.html}
}

@article{Peterson1981,
title = {Myths about the mutual exclusion problem},
journal = {Information Processing Letters},
volume = {12},
number = {3},
pages = {115-116},
year = {1981},
issn = {0020-0190},
doi = {https://doi.org/10.1016/0020-0190(81)90106-X},
url = {https://www.sciencedirect.com/science/article/pii/002001908190106X},
author = {G.L. Peterson}
}

@InProceedings{Atig2011,
author="Atig, Mohamed Faouzi
and Bouajjani, Ahmed
and Parlato, Gennaro",
editor="Gopalakrishnan, Ganesh
and Qadeer, Shaz",
title="Getting Rid of Store-Buffers in TSO Analysis",
booktitle="Computer Aided Verification",
year="2011",
publisher="Springer Berlin Heidelberg",
address="Berlin, Heidelberg",
pages="99--115",
isbn="978-3-642-22110-1"
}

@inproceedings{Singh2024,
author = {Singh, Abhishek Kr and Lahav, Ori},
title = {Decidable Verification under Localized Release-Acquire Concurrency},
year = {2024},
isbn = {978-3-031-57255-5},
publisher = {Springer-Verlag},
address = {Berlin, Heidelberg},
url = {https://doi.org/10.1007/978-3-031-57256-2_12},
doi = {10.1007/978-3-031-57256-2_12},
booktitle = {Tools and Algorithms for the Construction and Analysis of Systems: 30th International Conference, TACAS 2024, Held as Part of the European Joint Conferences on Theory and Practice of Software, ETAPS 2024, Luxembourg City, Luxembourg, April 6–11, 2024, Proceedings, Part III},
pages = {235–254},
numpages = {20},
location = {Luxembourg City, Luxembourg}
}

@inproceedings{Atig2010,
  title = {On the Verification Problem for Weak Memory Models},
  booktitle = {Proceedings of the 37th Annual {{ACM SIGPLAN-SIGACT}} Symposium on {{Principles}} of Programming Languages},
  author = {Atig, Mohamed Faouzi and Bouajjani, Ahmed and Burckhardt, Sebastian and Musuvathi, Madanlal},
  year = {2010},
  month = jan,
  series = {{{POPL}} '10},
  pages = {7--18},
  publisher = {Association for Computing Machinery},
  address = {New York, NY, USA},
  doi = {10.1145/1706299.1706303},
  isbn = {978-1-60558-479-9}
}

@inproceedings{Vafeiadis2015,
  title = {Common {{Compiler Optimisations}} Are {{Invalid}} in the {{C11 Memory Model}} and What We Can Do about It},
  booktitle = {Proceedings of the 42nd {{Annual ACM SIGPLAN-SIGACT Symposium}} on {{Principles}} of {{Programming Languages}}},
  author = {Vafeiadis, Viktor and Balabonski, Thibaut and Chakraborty, Soham and Morisset, Robin and Zappa Nardelli, Francesco},
  year = 2015,
  month = jan,
  pages = {209--220},
  publisher = {ACM},
  address = {Mumbai India},
  doi = {10.1145/2676726.2676995},
  isbn = {978-1-4503-3300-9}
}

@article{Esparza2014,
  title = {Pattern-{{Based Verification}} for {{Multithreaded Programs}}},
  author = {Esparza, Javier and Ganty, Pierre and Poch, Tom{\'a}{\v s}},
  year = {2014},
  month = sep,
  journal = {ACM Trans. Program. Lang. Syst.},
  volume = {36},
  number = {3},
  pages = {9:1--9:29},
  issn = {0164-0925},
  doi = {10.1145/2629644}
}

@article{Chini2017,
  title = {On the {{Complexity}} of {{Bounded Context Switching}}},
  author = {Chini, Peter and Kolberg, Jonathan and Krebs, Andreas and Meyer, Roland and Saivasan, Prakash},
  year = {2017},
  journal = {LIPIcs, Volume 87, ESA 2017},
  volume = {87},
  pages = {27:1-27:15},
  publisher = {Schloss Dagstuhl -- Leibniz-Zentrum f{\"u}r Informatik},
  issn = {1868-8969},
  doi = {10.4230/LIPICS.ESA.2017.27},
  collaborator = {Pruhs, Kirk and Sohler, Christian},
  copyright = {Creative Commons Attribution 3.0 Unported license, info:eu-repo/semantics/openAccess},
  isbn = {9783959770491}
}

@article{Atig2021,
author = {Atig, Mohamed Faouzi},
title = {What is decidable under the TSO memory model?},
year = {2021},
issue_date = {October 2020},
publisher = {Association for Computing Machinery},
address = {New York, NY, USA},
volume = {7},
number = {4},
url = {https://doi.org/10.1145/3458593.3458595},
doi = {10.1145/3458593.3458595},
journal = {ACM SIGLOG News},
month = mar,
pages = {4–19},
numpages = {16}
}

@INPROCEEDINGS {Abdulla1993,
author = { Abdulla, P. and Jonsson, B. },
booktitle = { Proceedings of 8th Annual IEEE Symposium on Logic in Computer Science },
title = {{ Verifying programs with unreliable channels }},
year = {1993},
volume = {},
ISSN = {},
pages = {160,161,162,163,164,165,166,167,168,169,170},
doi = {10.1109/LICS.1993.287591},
url = {https://doi.ieeecomputersociety.org/10.1109/LICS.1993.287591},
publisher = {IEEE Computer Society},
address = {Los Alamitos, CA, USA},
month =Jun}

@article{Schnoebelen2002,
  title = {Verifying Lossy Channel Systems Has Nonprimitive Recursive Complexity},
  author = {Schnoebelen, {\relax Ph}.},
  year = {2002},
  month = sep,
  journal = {Information Processing Letters},
  volume = {83},
  number = {5},
  pages = {251--261},
  issn = {00200190},
  doi = {10.1016/S0020-0190(01)00337-4},
  copyright = {https://www.elsevier.com/tdm/userlicense/1.0/}
}

@InProceedings{Baumann2020,
  author =	{Baumann, Pascal and Majumdar, Rupak and Thinniyam, Ramanathan S. and Zetzsche, Georg},
  title =	{{The Complexity of Bounded Context Switching with Dynamic Thread Creation}},
  booktitle =	{47th International Colloquium on Automata, Languages, and Programming (ICALP 2020)},
  pages =	{111:1--111:16},
  series =	{Leibniz International Proceedings in Informatics (LIPIcs)},
  ISBN =	{978-3-95977-138-2},
  ISSN =	{1868-8969},
  year =	{2020},
  volume =	{168},
  editor =	{Czumaj, Artur and Dawar, Anuj and Merelli, Emanuela},
  publisher =	{Schloss Dagstuhl -- Leibniz-Zentrum f{\"u}r Informatik},
  address =	{Dagstuhl, Germany},
  URL =		{https://drops.dagstuhl.de/entities/document/10.4230/LIPIcs.ICALP.2020.111},
  URN =		{urn:nbn:de:0030-drops-125187},
  doi =		{10.4230/LIPIcs.ICALP.2020.111}
}

@article{Raad2019,
author = {Raad, Azalea and Doko, Marko and Ro\v{z}i\'{c}, Lovro and Lahav, Ori and Vafeiadis, Viktor},
title = {On library correctness under weak memory consistency: specifying and verifying concurrent libraries under declarative consistency models},
year = {2019},
issue_date = {January 2019},
publisher = {Association for Computing Machinery},
address = {New York, NY, USA},
volume = {3},
number = {POPL},
url = {https://doi.org/10.1145/3290381},
doi = {10.1145/3290381},
journal = {Proc. ACM Program. Lang.},
month = jan,
articleno = {68},
numpages = {31}
}

@article{Lahav2022,
  title = {What's {{Decidable About Causally Consistent Shared Memory}}?},
  author = {Lahav, Ori and Boker, Udi},
  year = {2022},
  month = jun,
  journal = {ACM Transactions on Programming Languages and Systems},
  volume = {44},
  number = {2},
  pages = {1--55},
  issn = {0164-0925, 1558-4593},
  doi = {10.1145/3505273}
}

@inproceedings{Sarkar2009,
  title = {The Semantics of X86-{{CC}} Multiprocessor Machine Code},
  booktitle = {Proceedings of the 36th Annual {{ACM SIGPLAN-SIGACT}} Symposium on {{Principles}} of Programming Languages},
  author = {Sarkar, Susmit and Sewell, Peter and Nardelli, Francesco Zappa and Owens, Scott and Ridge, Tom and Braibant, Thomas and Myreen, Magnus O. and Alglave, Jade},
  year = 2009,
  month = jan,
  series = {{{POPL}} '09},
  pages = {379--391},
  publisher = {Association for Computing Machinery},
  address = {New York, NY, USA},
  doi = {10.1145/1480881.1480929},
  isbn = {978-1-60558-379-2}
}

@article{Alglave2012,
  title = {A Formal Hierarchy of Weak Memory Models},
  author = {Alglave, Jade},
  year = 2012,
  month = oct,
  journal = {Formal Methods in System Design},
  volume = {41},
  number = {2},
  pages = {178--210},
  issn = {1572-8102},
  doi = {10.1007/s10703-012-0161-5}
}

@inproceedings{Batty2011,
  title = {Mathematizing {{C}}++ Concurrency},
  booktitle = {Proceedings of the 38th Annual {{ACM SIGPLAN-SIGACT}} Symposium on {{Principles}} of Programming Languages},
  author = {Batty, Mark and Owens, Scott and Sarkar, Susmit and Sewell, Peter and Weber, Tjark},
  year = {2011},
  month = jan,
  series = {{{POPL}} '11},
  pages = {55--66},
  publisher = {Association for Computing Machinery},
  address = {New York, NY, USA},
  doi = {10.1145/1926385.1926394},
  url = {https://dl.acm.org/doi/10.1145/1926385.1926394},
  isbn = {978-1-4503-0490-0}
}

@InProceedings{Castaneda24,
  author =	{Casta\~{n}eda, Armando and Chockler, Gregory and Dongol, Brijesh and Lahav, Ori},
  title =	{{What Cannot Be Implemented on Weak Memory?}},
  booktitle =	{38th International Symposium on Distributed Computing (DISC 2024)},
  pages =	{11:1--11:22},
  series =	{Leibniz International Proceedings in Informatics (LIPIcs)},
  ISBN =	{978-3-95977-352-2},
  ISSN =	{1868-8969},
  year =	{2024},
  volume =	{319},
  editor =	{Alistarh, Dan},
  publisher =	{Schloss Dagstuhl -- Leibniz-Zentrum f{\"u}r Informatik},
  address =	{Dagstuhl, Germany},
  URL =		{https://drops.dagstuhl.de/entities/document/10.4230/LIPIcs.DISC.2024.11},
  URN =		{urn:nbn:de:0030-drops-212371},
  doi =		{10.4230/LIPIcs.DISC.2024.11}
}

@inproceedings{Abdulla19,
author = {Abdulla, Parosh Aziz and Arora, Jatin and Atig, Mohamed Faouzi and Krishna, Shankaranarayanan},
title = {Verification of programs under the release-acquire semantics},
year = {2019},
isbn = {9781450367127},
publisher = {Association for Computing Machinery},
address = {New York, NY, USA},
url = {https://doi.org/10.1145/3314221.3314649},
doi = {10.1145/3314221.3314649},
booktitle = {Proceedings of the 40th ACM SIGPLAN Conference on Programming Language Design and Implementation},
pages = {1117–1132},
numpages = {16},
location = {Phoenix, AZ, USA},
series = {PLDI 2019}
}

@article{Post46,
  title={A variant of a recursively unsolvable problem},
  author={Emil L. Post},
  journal={Bulletin of the American Mathematical Society},
  year={1946},
  volume={52},
  pages={264-268},
  url={https://api.semanticscholar.org/CorpusID:122948861}
}

@article{Pulte2017,
author = {Pulte, Christopher and Flur, Shaked and Deacon, Will and French, Jon and Sarkar, Susmit and Sewell, Peter},
title = {Simplifying ARM concurrency: multicopy-atomic axiomatic and operational models for ARMv8},
year = {2018},
issue_date = {January 2018},
publisher = {Association for Computing Machinery},
address = {New York, NY, USA},
volume = {2},
number = {POPL},
url = {https://doi.org/10.1145/3158107},
doi = {10.1145/3158107},
journal = {Proc. ACM Program. Lang.},
month = dec,
articleno = {19},
numpages = {29}
}

@article{Lamport1979,
  title = {How to {{Make}} a {{Multiprocessor Computer That Correctly Executes Multiprocess Programs}}},
  author = {{Lamport}},
  year = 1979,
  month = sep,
  journal = {IEEE Transactions on Computers},
  volume = {C-28},
  number = {9},
  pages = {690--691},
  issn = {1557-9956},
  doi = {10.1109/TC.1979.1675439}
}

@InProceedings{Abdulla2021b,
author="Abdulla, Parosh Aziz
and Atig, Mohamed Faouzi
and Bouajjani, Ahmed
and Derevenetc, Egor
and Leonardsson, Carl
and Meyer, Roland",
editor="Georgiou, Chryssis
and Majumdar, Rupak",
title="On the State Reachability Problem for Concurrent Programs Under Power",
booktitle="Networked Systems",
year="2021",
publisher="Springer International Publishing",
address="Cham",
pages="47--59",
isbn="978-3-030-67087-0"
}

@inproceedings{Abdulla2017,
author = {Abdulla, Parosh Aziz and Atig, Mohamed Faouzi and Bouajjani, Ahmed and Ngo, Tuan Phong},
title = {Context-Bounded Analysis for POWER},
year = {2017},
isbn = {9783662545799},
publisher = {Springer-Verlag},
address = {Berlin, Heidelberg},
url = {https://doi.org/10.1007/978-3-662-54580-5_4},
doi = {10.1007/978-3-662-54580-5_4},
booktitle = {Proceedings, Part II, of the 23rd International Conference on Tools and Algorithms for the Construction and Analysis of Systems - Volume 10206},
pages = {56–74},
numpages = {19}
}

@INPROCEEDINGS{Kozen1977,
  author={Kozen, Dexter},
  booktitle={18th Annual Symposium on Foundations of Computer Science (sfcs 1977)}, 
  title={Lower bounds for natural proof systems}, 
  year={1977},
  volume={},
  number={},
  pages={254-266},
  doi={10.1109/SFCS.1977.16}}

@article{Sistla1985,
author = {Sistla, A. P. and Clarke, E. M.},
title = {The complexity of propositional linear temporal logics},
year = {1985},
issue_date = {July 1985},
publisher = {Association for Computing Machinery},
address = {New York, NY, USA},
volume = {32},
number = {3},
issn = {0004-5411},
url = {https://doi.org/10.1145/3828.3837},
doi = {10.1145/3828.3837},
journal = {J. ACM},
month = jul,
pages = {733–749},
numpages = {17}
}

@InProceedings{Qadeer2005,
author="Qadeer, Shaz
and Rehof, Jakob",
editor="Halbwachs, Nicolas
and Zuck, Lenore D.",
title="Context-Bounded Model Checking of Concurrent Software",
booktitle="Tools and Algorithms for the Construction and Analysis of Systems",
year="2005",
publisher="Springer Berlin Heidelberg",
address="Berlin, Heidelberg",
pages="93--107",
isbn="978-3-540-31980-1"
}

@misc{Abdulla26,
      title={On the Verification Problem of Remote Direct Memory Access programs (Extended Version with Appendix)}, 
      author={Parosh Aziz Abdulla and Mohamed Faouzi Atig and Govind Rajanbabu and Stephan Spengler},
      year={2026},
      eprint={2605.10631},
      archivePrefix={arXiv},
      primaryClass={cs.LO},
      url={https://arxiv.org/abs/2605.10631}, 
}

@InProceedings{Conrado26,
  author =	{Conrado, Giovanna Kobus and Pavlogiannis, Andreas},
  title =	{{An MSO Framework for Weak-Memory Verification and Robustness}},
  booktitle =	{37th International Conference on Concurrency Theory (CONCUR 2026)},
  pages =	{25:1--25:23},
  series =	{Leibniz International Proceedings in Informatics (LIPIcs)},
  ISBN =	{978-3-95977-447-5},
  ISSN =	{1868-8969},
  year =	{2026},
  volume =	{391},
  editor =	{Sokolova, Ana and Totzke, Patrick},
  publisher =	{Schloss Dagstuhl -- Leibniz-Zentrum f{\"u}r Informatik},
  address =	{Dagstuhl, Germany},
  URL =		{https://drops.dagstuhl.de/entities/document/10.4230/LIPIcs.CONCUR.2026.25},
  URN =		{urn:nbn:de:0030-drops-273568},
  doi =		{10.4230/LIPIcs.CONCUR.2026.25}
}

\pagebreak

\appendix
\newcommand{\LVars}{\mathsf{LVar}}
\newcommand{\LVals}{\mathsf{LVal}}
\newcommand{\Store}{\varphi}
\newcommand{\StoreInit}{\Store_{\mathsf{in}}}
\newcommand{\Cmd}{\mathsf{C}}
\newcommand{\cfg}[2]{\langle #1, #2 \rangle}
\newcommand{\eval}[2]{\llbracket #1 \rrbracket_{#2}}
\newcommand{\skipc}{\mathsf{skip}}
\newcommand{\breakc}{\mathsf{break}}
\newcommand{\ifc}{\mathsf{if}}
\newcommand{\thenc}{\mathsf{then}}
\newcommand{\elsec}{\mathsf{else}}
\newcommand{\whilec}{\mathsf{while}}
\newcommand{\doc}{\mathsf{do}}
\newcommand{\forc}{\mathsf{for}}
\newcommand{\loopseq}{\mathbin{\rhd}}
\newcommand{\SilentTo}{\rightsquigarrow}
\newcommand{\StepTo}[1]{\xrightarrow{#1}}
\newcommand{\Reaches}[1]{\Rightarrow^{#1}}
\newcommand{\defeq}{\triangleq}

\section{\Revision{Semantics for \cref{SUBSEC:UNDECIDABILITY_FORMAL}}}\label{SEC:APP_SEMANTICS}

\Revision{\cref{fig:pcp_code}  presents the program $\ConcProg$ in pseudocode.
Here we make its translation to the LTSs of \cref{SUBSEC:PRELIM_PROG} precise, which formalizes the sequence of events generated by each thread and enables the mechanization of our results.
We define a core imperative language and its operational semantics, derive the remaining constructs of \cref{fig:pcp_code}, and define the induced LTS.}

\Paragraph{Syntax.}
\Revision{Fix a thread $t\in\ThreadDom$, a finite set $\LVars$ of thread-local \emph{variables}, and a finite set $\LVals$ of \emph{local values}.
A \emph{store} is a function $\Store\colon\LVars\to\LVals$.
We assume a set of side-effect-free \emph{expressions} $e$ over $\LVars$, each evaluating to a value $\eval{e}{\Store}\in\LVals$ under a store $\Store$, and write $b$ for expressions evaluating to $\True$ or $\False$.
\emph{Commands} are given by the grammar}
\[
\begin{array}{rcl}
C &::=& \skipc
   \;\mid\; v \gets e
   \;\mid\; v \gets \ndet{X}
   \;\mid\; \wt(u, e)
   \;\mid\; \rd(u, e)\\[2pt]
  &\mid& C_1 ; C_2
   \;\mid\; \ifc\; b\; \thenc\; C_1\; \elsec\; C_2
   \;\mid\; \whilec\; b\; \doc\; C
   \;\mid\; \breakc
   \;\mid\; \term
\end{array}
\]
\Revision{where $v\in\LVars$, $u\in\LocationDom$, and $X\subseteq\LVals$ is finite and nonempty.
The program of $t$ is a command $C ; \term$, where $C$ does not contain $\term$ and every $\breakc$ occurs within the body of a $\whilec$.}

\Paragraph{Operational semantics.}
\Revision{We define \emph{silent} steps $\SilentTo$, which do not emit a label, and \emph{labeled} steps $\StepTo{a}$ with $a\in\LabelDom^t\cup\{\term\}$, where $\term$ is not an event label, between \emph{configurations} $\cfg{C}{\Store}$. 
The silent steps are as follows.}
\begin{mathpar}
\inferrule[assign]{ }{
  \cfg{v \gets e}{\Store} \SilentTo \cfg{\skipc}{\Store[v\mapsto \eval{e}{\Store}]}
}
\and
\inferrule[ndet]{c \in X}{
  \cfg{v \gets \ndet{X}}{\Store} \SilentTo \cfg{\skipc}{\Store[v\mapsto c]}
}
\and
\inferrule[seq-skip]{ }{
  \cfg{\skipc ; C}{\Store} \SilentTo \cfg{C}{\Store}
}
\and
\inferrule[seq-break]{ }{
  \cfg{\breakc ; C}{\Store} \SilentTo \cfg{\breakc}{\Store}
}
\and
\inferrule[if-true]{\eval{b}{\Store} = \True}{
  \cfg{\ifc\, b\, \thenc\, C_1\, \elsec\, C_2}{\Store} \SilentTo \cfg{C_1}{\Store}
}
\and
\inferrule[if-false]{\eval{b}{\Store} = \False}{
  \cfg{\ifc\, b\, \thenc\, C_1\, \elsec\, C_2}{\Store} \SilentTo \cfg{C_2}{\Store}
}
\and
\inferrule[while-true]{\eval{b}{\Store} = \True}{
  \cfg{\whilec\, b\, \doc\, C}{\Store} \SilentTo
  \cfg{C \loopseq (\whilec\, b\, \doc\, C)}{\Store}
}
\and
\inferrule[while-false]{\eval{b}{\Store} = \False}{
  \cfg{\whilec\, b\, \doc\, C}{\Store} \SilentTo \cfg{\skipc}{\Store}
}
\and
\inferrule[loop-skip]{ }{
  \cfg{\skipc \loopseq C}{\Store} \SilentTo \cfg{C}{\Store}
}
\and
\inferrule[loop-break]{ }{
  \cfg{\breakc \loopseq C}{\Store} \SilentTo \cfg{\skipc}{\Store}
}
\end{mathpar}
\Revision{Here, $C \loopseq C'$ is an auxiliary form that arises only during execution, marking the body $C$ of a loop: a $\breakc$ discards the commands that follow it up to the enclosing $\loopseq$, which exits the loop.
The labeled steps are}
\begin{mathpar}
\inferrule[write]{ }{
  \cfg{\wt(u,e)}{\Store} \StepTo{\wt(t, u, \eval{e}{\Store})} \cfg{\skipc}{\Store}
}
\and
\inferrule[read]{ }{
  \cfg{\rd(u,e)}{\Store} \StepTo{\rd(t, u, \eval{e}{\Store})} \cfg{\skipc}{\Store}
}
\and
\inferrule[term]{ }{
  \cfg{\term}{\Store} \StepTo{\term} \FinalState
}
\end{mathpar}
and, writing $\StepTo{\lambda}$ for either kind of step, both are closed under
\begin{mathpar}
\inferrule[seq-step]{
  \cfg{C_1}{\Store} \StepTo{\lambda} \cfg{C_1'}{\Store'}
}{
  \cfg{C_1 ; C_2}{\Store} \StepTo{\lambda} \cfg{C_1' ; C_2}{\Store'}
}
\and
\inferrule[loop-step]{
  \cfg{C_1}{\Store} \StepTo{\lambda} \cfg{C_1'}{\Store'}
}{
  \cfg{C_1 \loopseq C_2}{\Store} \StepTo{\lambda} \cfg{C_1' \loopseq C_2}{\Store'}
}
\end{mathpar}
\Revision{Reads and writes leave the store unchanged and only emit a label; their effect is determined by the execution graph, so a read acts as a guard: e.g., $\rd(\vzax, prev)$ only occurs in executions in which it reads the value $\eval{prev}{\Store}$. }

\Paragraph{Derived forms.}
\Revision{The remaining constructs of \cref{fig:pcp_code} are derived as follows, using fresh variables $b$, $N$ and $\mathit{fst}_\iota$:}
\begin{align*}
\whilec\;\ndet{X}\;\doc\;C &\;\defeq\; b \gets \ndet{X};\; \whilec\; b\; \doc\; (C;\; b \gets \ndet{X}),\\
\forc\; j \in 1,\dots,e\; \doc\; C &\;\defeq\; j \gets 1;\; N \gets e;\; \whilec\; j \le N\; \doc\; (C;\; j \gets j + 1),\\
\ifc\;\firstit\;\thenc\; C_1\; \elsec\; C_2 &\;\defeq\; \ifc\; \mathit{fst}_\iota\; \thenc\; (\mathit{fst}_\iota \gets \False;\; C_1)\; \elsec\; C_2,
\end{align*}
\Revision{where $\iota$ identifies the occurrence of $\firstit$, and $\ifc\; b\; \thenc\; C$ abbreviates $\ifc\; b\; \thenc\; C\; \elsec\; \skipc$.
The initial store $\StoreInit$ maps $\vc$ and $\vcy$ to $0$, each $\mathit{fst}_\iota$ to $\True$, and every other variable to a fixed value, which is irrelevant as it is assigned before being read.
Finally, we treat $\alpha$ as a function on $\IntSet{n}\cup\{\bot\}$ with $\alpha_\bot=\bot$, and similarly for $\beta$; the assignments $\alpha_\bot \gets \bot$ and $\beta_\bot \gets \bot$ in \cref{fig:pcp_code} only fix this convention.}

\Paragraph{The induced LTS.}
\Revision{Let $C_t$ be the program of $t$ in \cref{fig:pcp_code} after expanding the derived forms, and write $c \Reaches{a} c'$ if $c \SilentTo^* c'' \StepTo{a} c'$ for some configuration $c''$.

The sequential program $\Prog_t$ is the LTS $\tuple{\States_t, \Trans_t, \InitState, \FinalState}$ over $\LabelDom^t$, where $\InitState = \cfg{C_t}{\StoreInit}$, $\States_t$ consists of $\FinalState$ and all configurations reachable from $\InitState$, and}
\[
\Trans_t = \{(c, a, c') : c \Reaches{a} c'\} \cup \{(c, a, \FinalState) : c \Reaches{a} c' \Reaches{\term} \FinalState \text{ for some } c'\}.
\]

\Revision{Thus a word reaches $\FinalState$ iff the thread can silently reach the command $\term$. 
As $\InitState\neq\FinalState$, reaching $\FinalState$ requires emitting at least one label, which holds for all threads in the execution constructed in \cref{SUBSEC:UNDECIDABILITY_COMPLETENESS}.
Moreover, $\States_t$ is finite, as every variable of \cref{fig:pcp_code} ranges over a finite set. }
\section{Proofs of \cref{SEC:UNDECRA}}\label{SEC:APP_UNDECRA}

\subsection{Proofs of \cref{SUBSEC:UNDECIDABILITY_SOUNDNESS}}

\undecidabilitynoskippingpairs*

\begin{proof}[Proof of  \cref{item:no_skip_0,item:no_skip_1}]

These invariants are established in pairs, starting with
\[
\wt_i(\taxp, \vzax) \LTo{\rf} \rd_i(\tax, \vzax), \quad
\wt_i(\tax, \vzaxp) \LTo{\rf} \rd_i(\taxp,\vzaxp)\ ,
\]
and generally, each pair involves writes from a thread $t$ read by thread $t'$ and vice versa.
The proof of each pair has identical structure, and is based on the fact that each of the involved threads alternates between writing to one of the involved locations and reading from the other.
To avoid redundancy, we only argue about the first pair.
Our proof follows an induction on $i$.

First, observe that, if $\rd_1(\tax, \vzax)\in E$, then $\wt_1(\taxp, \vzax) \LTo{\rf} \rd_1(\tax, \vzax)$.
This holds because $\rd_1(\tax, \vzax)$ must read value $1$.
Any later write $\wt_j(\taxp, \vzax)$ that writes this value is such that $j\geq 5$.
However, any such $\wt_j(\taxp, \vzax)$ is preceded by $\rd_4(\taxp, \vzaxp)$, which reads value $0$.
In turn, this value is written by writes of $\tax$ that appear after $\rd_1(\tax, \vzax)$. 
Hence, if $\rd_1(\tax, \vzax)$ were to read from any write $\wt_j(\taxp, \vzax)$ with $j\geq 2$, this would cause an $\hb$-cycle, violating \ref{eq:hbirr}.
Symmetric reasoning on $\rd_1(\taxp, \vzaxp)$ establishes that $\wt_1(\tax,\vzaxp)\LTo{\rf} \rd_1(\taxp, \vzaxp)$.

The argument proceeds inductively.
In particular, we may assume, for a given $i$, that for all $j<i$ we have $\rd_j(\tax, \vzax), \wt_j(\taxp, \vzax) \in E$, $\wt_{j}(\taxp, \vzax)\LTo{\rf}\rd_{j}(\tax, \vzax)$, $\rd_j(\taxp, \vzaxp),\wt_j(\tax, \vzaxp)\in E$, and  $\wt_j(\tax, \vzaxp)\LTo{\rf }\rd_j(\taxp, \vzaxp)$.
We have that  $\rd_i(\tax, \vzax)$ cannot read from any $\wt_{j}(\taxp, \vzax)$, with $j<i$, since
\begin{enumerate*}[label=(\roman*)]
\item for all $j<i-1$, $\wt_{j}(\taxp, \vzax)\LTo{\mo}\wt_{i-1}(\taxp, \vzax)$,
\item by the induction hypothesis, we have  $\wt_{i-1}(\taxp, \vzax)\LTo{\hb}\rd_{i-1}(\tax, \vzax)$ and thus  $\wt_{i-1}(\taxp, \vzax)\LTo{\hb}\rd_i(\tax, \vzax)$, and
\item $\val_{\wt}(\wt_{i-1}(\taxp, \vzax))\neq \val_{\rd}(\rd_i(\tax, \vzax))$.
\end{enumerate*}
Therefore, reading from any $\wt_{j}(\taxp, \vzax)$ that writes $\val_{\rd}(\rd_i(\tax, \vzax))$, with $j<i$, would violate \ref{eq:rc}.

Moreover, any later write $\wt_j(\taxp, \vzax)$ with $\val_{\wt}(\wt_j(\taxp, \vzax))=\val_{\rd}(\rd_i(\tax, \vzax))$ is such that $j\geq i+4$.
However, any such $\wt_j(\taxp, \vzax)$, if it exists, is preceded by $\rd_{i+2}(\taxp, \vzaxp)$.
In turn, $\rd_{i+2}(\taxp, \vzaxp)$ must read its value from a write that appears after $\rd_i(\tax, \vzax)$ in $\tax$.
This is because any earlier write that writes $\val_{\rd}(\rd_{i+2}(\taxp, \vzaxp))$ would be of the form $\wt_{\ell}(\tax, \vzaxp)$, with $\ell<i-1$, and, by the induction hypothesis, we have that $\wt_{i-1}(\tax, \vzaxp)\LTo{\rf }\rd_{i-1}(\taxp, \vzaxp)$; therefore, $\rd_{i+2}(\taxp, \vzaxp)$ reading from $\wt_{\ell}(\tax, \vzaxp)$ would violate \ref{eq:rc}.
This implies that $\rd_{i+2}(\taxp, \vzaxp)$ must read its value from some write later than $\wt_{i}(\tax, \vzaxp)$ in $\tax$, and thus also later than $\rd_i(\tax, \vzax)$.
Hence, if $\wt_j(\taxp, \vzax)\LTo{\rf}\rd_i(\tax, \vzax)$ with $j>i$, we would have a violation of \ref{eq:hbirr}.
We thus arrive at $\wt_i(\taxp, \vzax)\LTo{\rf}\rd_i(\tax, \vzax)$, as desired.
Symmetric reasoning on $\rd_i(\taxp, \vzaxp)$ establishes that $\wt_i(\tax,\vzaxp)\LTo{\rf} \rd_i(\taxp, \vzaxp)$.

The desired result follows.
\end{proof}

\begin{proof}[Proof of \cref{item:no_skip_2}]

We only argue about the first pair
\[
\wt_i(\tax, \vxa) \LTo{\rf} \rd_i(\tx, \vxa), \quad
\wt_i(\taxp, \vxap) \LTo{\rf} \rd_i(\txp, \vxap)\ ,
\]
as all remaining pairs follow the same argument.
The proof follows an induction on $i$.

We first prove that if $\rd_1(\tx,\vxa)\in E$, then $\wt_1(\tax, \vxa)\LTo{\rf}\rd_1(\tx,\vxa)$.
It must be the case that $\wt_1(\tax, \vxa),\rd_1(\tx,\vxa)\in E$, as, by the construction shown in \cref{fig:pcp_code}, $\tax$ and $\tx$ execute their innermost loops at least once.
Furthermore, $\rd_1(\tx,\vxa)$ must read value $\tuple{\ov{\tax}, \AnyValue}$, which can only be written by $\wt_1(\tax, \vxa)$, so $\wt_1(\tax, \vxa)\LTo{\rf}\rd_1(\tx,\vxa)$.

We now show that if $\rd_1(\txp,\vxap)\in E$, then $\wt_1(\taxp,\vxap)\LTo{\rf}\rd_1(\txp,\vxap)$.
Since $\wt_1(\tx, \vxa)\LTo{\hb} \rd_1(\tx,\vxa)$ and $\wt_1(\tax, \vxa)\LTo{\rf}\rd_1(\tx,\vxa)$, $\wt_1(\tx, \vxa)\LTo{\mo}\wt_1(\tax, \vxa)$ due to \ref{eq:rc}.
The thread $\taxp$ must execute its main loop at least once, due to the facts that $\rd_1(\txp,\vxap)\in E$ and this event, because of its value, must read from $\taxp$.
Therefore, $\rd_1(\taxp,\vzaxp)\in E$, and due to \cref{item:no_skip_0}, we have that $\wt_1(\tax, \vzaxp) \LTo{\rf} \rd_1(\taxp,\vzaxp)$.
This implies that $\wt_1(\tax, \vxa)\LTo{\hb} \wt_i(\taxp, \vxap)$ for all $i \geq 2$.
If $\wt_i(\taxp,\vxap)\LTo{\rf}\rd_1(\txp, \vxap)$ for any $i \geq 2$, then we would have $\wt_1(\tax, \vxa)\LTo{\hb}\rd_1(\txp,\vxa)$.
Since $\rd_1(\txp,\vxap)\in E$, the thread $\txp$ executes its main loop at least once, so $\rd_1(\txp, \vxa)\in E$, and, by \cref{item:no_skip_1}, we have $\wt_1(\tx, \vxa) \LTo{\rf} \rd_1(\txp, \vxa)$.
Together with $\wt_1(\tx, \vxa)\LTo{\mo}\wt_1(\tax, \vxa)$, this would create a violation of \ref{eq:rc}.
We thus have $\wt_1(\taxp,\vxap)\LTo{\rf}\rd_1(\txp,\vxap)$.

The argument then repeats inductively. 
In particular, we may assume, for a given $i$, that for all $j<i$, we have $\rd_j(\tx,\vxa),\rd_j(\txp,\vxap)\in E$, $\wt_j(\tax, \vxa) \LTo{\rf} \rd_j(\tx, \vxa)$, and $\wt_j(\taxp, \vxap) \LTo{\rf} \rd_j(\txp, \vxap)$.
We first prove that if $\rd_i(\tx, \vxa)\in E$, then $\wt_i(\tax, \vxa) \LTo{\rf} \rd_i(\tx, \vxa)$.
The event $\rd_i(\tx,\vxa)$ cannot  read from any $\wt_j(\tax, \vxa)$ with $j<i$, because, by the induction hypothesis, we have $\wt_j(\tax, \vxa)\LTo{\rf}\rd_j(\tx,\vxa)$, while by construction, we have $\rd_j(\tx,\vxa)\LTo{\hb}\wt_i(\tx, \vxa)\LTo{\hb}\rd_i(\tx,\vxa)$.
These imply that $\wt_j(\tax, \vxa)\LTo{\hb}\wt_i(\tx, \vxa)$, and thus $\wt_j(\tax, \vxa)\LTo{\mo}\wt_i(\tx, \vxa)$ due to \ref{eq:wc}.
Together with $\wt_j(\tax, \vxa)\LTo{\rf}\rd_i(\tx,\vxa)$, the two facts $\wt_i(\tx, \vxa)\LTo{\hb}\rd_i(\tx,\vxa)$ and $\wt_j(\tax, \vxa)\LTo{\mo}\wt_i(\tx, \vxa)$ would violate \ref{eq:rc}.

We now show that the event $\rd_i(\tx,\vxa)$ cannot  read from any $\wt_j(\tax, \vxa)$ with $j>i$.
Since $\wt_{i-1}(\taxp,\vxap)\LTo{\rf}\rd_{i-1}(\txp,\vxap)$ by the induction hypothesis,
and also $\wt_{i-1}(\txp, \vxap)\LTo{\hb}\rd_{i-1}(\txp,\vxap)$, we also have $\wt_{i-1}(\txp, \vxap)\LTo{\mo}\wt_{i-1}(\taxp,\vxap)$ due to \ref{eq:rc}.
If $\wt_{j}(\tax, \vxa)\in E$ for $j>i$, then, by the construction of $\tax$ shown in \cref{fig:pcp_code}, $\rd_{i-1}(\tax, \vzax)\in E$ also. 
Due to \cref{item:no_skip_0}, this would require $\wt_{i-1}(\taxp, \vzax)\LTo{\rf} \rd_{i-1}(\tax, \vzax)$, 
which implies that $\wt_{i-1}(\taxp, \vxap)\LTo{\hb}\wt_{j}(\tax, \vxa)$ for all $j>i$.
If $\wt_{j}(\tax, \vxa)\LTo{\rf}\rd_i(\tx,\vxa)$ for any $j>i$, then we would have $\wt_{i-1}(\taxp,\vxap)\LTo{\hb} \rd_{i-1}(\tx, \vxap)$.
Since, by the induction hypothesis, we have $\wt_{i-1}(\txp, \vxap) \LTo{\rf} \rd_{i-1}(\tx, \vxap)$, together with $\wt_{i-1}(\txp, \vxap)\LTo{\mo}\wt_{i-1}(\taxp,\vxap)$, this would violate \ref{eq:rc}.
We thus have $\wt_i(\tax, \vxa)\LTo{\rf}\rd_i(\tx,\vxa)$.

Finally, we argue that if $\rd_i(\txp, \vxap)\in E$, then $\wt_i(\taxp, \vxap) \LTo{\rf} \rd_i(\txp, \vxap)$.
In particular, $\rd_i(\txp, \vxap)$ cannot read from any $\wt_j(\taxp, \vxap)$ with $j<i$, because, by the induction hypothesis, we have $\wt_j(\taxp, \vxap)\LTo{\rf}\rd_j(\txp, \vxap)$, while by construction, we have $\rd_j(\txp, \vxap)\LTo{\hb}\wt_i(\txp, \vxap)\LTo{\hb} \rd_i(\txp, \vxap)$.
These imply that $\wt_j(\taxp, \vxap)\LTo{\hb} \wt_i(\txp, \vxap)$, and thus $\wt_j(\taxp, \vxap)\LTo{\mo} \wt_i(\txp, \vxap)$ due to \ref{eq:wc}.
Together with $\wt_j(\taxp, \vxap)\LTo{\rf}\rd_i(\txp, \vxap)$, the two facts $\wt_j(\taxp, \vxap)\LTo{\mo} \wt_i(\txp, \vxap)$ and $\wt_i(\txp, \vxap)\LTo{\hb}\rd_i(\txp, \vxap)$ would violate \ref{eq:rc}.

We now show that the event $\rd_i(\txp, \vxap)$ cannot  read from any $\wt_j(\taxp, \vxap)$ with $j>i$.
Since $\wt_i(\tax, \vxa)\LTo{\rf}\rd_i(\tx,\vxa)$ and $\wt_i(\tx, \vxa)\LTo{\hb} \rd_i(\tx,\vxa)$, this implies that $\wt_i(\tx, \vxa)\LTo{\mo}\wt_i(\tax, \vxa)$, due to \ref{eq:rc}.
If $\wt_j(\taxp, \vxap)\in E$ for $j>i$, then, by the construction of $\taxp$ shown in \cref{fig:pcp_code}, $\rd_i(\taxp,\vzaxp)\in E$ also. 
Due to \cref{item:no_skip_0}, this would require $\wt_i(\tax, \vzaxp) \LTo{\rf} \rd_i(\taxp,\vzaxp)$, which implies that $\wt_i(\tax, \vxa)\LTo{\hb} \wt_j(\taxp, \vxap)$ for all $j>i$.
If $\wt_j(\taxp,\vxap)\LTo{\rf}\rd_i(\txp, \vxap)$ for any $j>i$, then we would have $\wt_i(\tax, \vxa)\LTo{\hb}\rd_i(\txp,\vxa)$.
Since, as previously established, $\wt_i(\tx, \vxa) \LTo{\rf} \rd_i(\txp, \vxa)$, together with $\wt_i(\tx, \vxa)\LTo{\mo}\wt_i(\tax, \vxa)$, this would violate \ref{eq:rc}.
We thus have $\wt_i(\taxp,\vxap)\LTo{\rf}\rd_i(\txp,\vxap)$.

The desired result follows.
\end{proof}

\subsection{Proofs of \cref{SUBSEC:UNDECIDABILITY_COMPLETENESS}}

We present here the proofs of the auxiliary lemmas of \cref{SUBSEC:UNDECIDABILITY_COMPLETENESS}. 

\lempostrictmonotonic*
\begin{proof}
We analyze the $\po$ of different threads separately.
Let an \emph{iteration} be a sequence of events corresponding to one execution of the main loop of each thread, as presented in \cref{fig:pcp_code}
; for $\tax$ and $\tbx$, we use the innermost loop.

\Paragraph{Threads $\taxp$, $\tayp$, $\tbxp$, $\tbyp$, $\txp$ and $\typ$.} 
Each iteration in these threads executes its respective events exactly once and in the same order. 
Furthermore, no location is written to or read from twice in the same iteration. 
As a consequence, the indices of all events increase by exactly one in each iteration, and, in these threads, $\po$ is naturally non-decreasing (satisfying \cref{item:po_non_decreasing}); 
moreover, since each iteration in these threads executes a sequence of writes followed by a sequence of reads, a write event succeeding a read event cannot be in its same iteration, therefore every read-to-write $\po$-edge is increasing (satisfying \cref{item:po_increasing}). 

\Paragraph{Threads $\tx$, $\ty$, $\tay$ and $\tby$.} 
These threads skip a suffix of reads in the first iteration. 
Otherwise, their respective events are executed exactly once and in the same order and no location is written to or read from twice in the same iteration . 
For instance, in $\tx$, the $(i+1)$-th iteration contains $\rd_{i}(\vxap)$ and $\rd_{i}(\vxbp)$ rather than $\rd_{i+1}(\vxap)$ and $\rd_{i+1}(\vxbp)$.
However, all writes in these threads succeeding the $(i+1)$-th iteration are still of the form $\wt_{j}$ for $i+1<j$.
Since any $\po$-edge from a read in the $i+1$-th iteration to such a write satisfies $i+1<j$, $i<j$, and therefore \cref{item:po_increasing} holds.

\Paragraph{Threads $\tax$ and $\tbx$.}
We only consider $\tax$, as the case for $\tbx$ is symmetric. 
Let us first consider edges $\event_i \LTo{\po} \wt_j$ that do not include events on $\vzaxy$ or $\vzaxyp$. 
In this case, we are in the same scenario as $\tx$, $\ty$, $\tay$ and $\tby$, where a suffix of reads (namely, a read on $\vzax$) is skipped in the first iteration. 
As argued in the paragraph above, \cref{item:po_non_decreasing} and \cref{item:po_increasing} hold. 

Let us now consider edges that do include events on $\vzaxy$ or $\vzaxyp$.
Since these events are not non-bridge writes (and therefore candidates for $\wt_j$), we only need to consider the cases $\wt_i(\vzaxy) \LTo{\po} \wt_j$ and $\rd_i(\vzaxyp) \LTo{\po} \wt_j$. 

\SubParagraph{Case $\wt_i(\vzaxy) \LTo{\po} \wt_j$.} 
The inner loop of $\tax$ always executes at least once after an event $\wt_i(\vzaxy)$; 
therefore, this event precedes the $j'$-th iteration, where $j'\geq i$. 
Since all non-bridge writes of $\tax$ are located in its inner loop, all non-bridge writes $\po$-ordered after $\wt_i(\vzaxy)$ must all have index $j \geq j' \geq i$, satisfying \cref{item:po_non_decreasing}.

\SubParagraph{Case $\rd_i(\vzaxyp) \LTo{\po} \wt_j$.} 
The case above shows that $\po$-edges connecting bridge writes to non-bridge writes are non-decreasing.
Therefore, all non-bridge writes $\wt_j$ such that $\wt_{i+1}(\vzaxy)\LTo{\po} \wt_j$ are such that $j\geq i+1$.
Since $\rd_i(\vzaxyp)$ immediately precedes $\wt_{i+1}(\vzaxy)$ in $\po$-order (by inspection of \cref{fig:pcp_code}),
all non-bridge writes $\wt_j$ such that $\rd_i(\vzaxyp)\LTo{\po} \wt_j$ 
are such that $j\geq i+1$, satisfying \cref{item:po_increasing} . 
\end{proof}

\lemdecreasingrrpo*
\begin{proof}
Since same-location $\po$-edges between read events must be increasing by definition, 
we only need to consider the threads with non-bridge reads on more than one location. 
These are $\tx$, $\txp$, $\ty$ and $\typ$.
Let an \emph{iteration} be a sequence of events corresponding to one execution of the main loop of each thread, as presented in \cref{fig:pcp_code}.

\Paragraph{Threads $\txp$ and $\typ$.} 
These threads execute the same read events exactly once in each iteration, with no location being read from twice in the same iteration.
Therefore the indices of their events are, overall, non-decreasing.

\Paragraph{Threads $\tx$ and $\ty$.}
Thread $\tx$ skips reads on $\vxap$ and $\vxbp$ in the first iteration.
Otherwise, each iteration in these threads executes its respective events exactly once and in the same order. 
Furthermore, no location is written to or read from twice in the same iteration. 
 This means that, from the second iteration of thread $\tx$ onward, the two reads $\rd_i(\vxa)$ and  $\rd_i(\vxb)$ are succeeded by reads $\rd_{i-1}(\vxap)$ and  $\rd_{i-1}(\vxbp)$. 
The four corresponding pairs $\tuple{\rd_i, \rd_{i-1}}$ yield the stated decreasing edges. 

Thread $\ty$ follows the same construction, but with locations $\vya$, $\vyap$, $\vyb$ and $\vybp$, replacing $\vxa$, $\vxap$, $\vxb$, and $\vxbp$, respectively.
\end{proof}

\lemhbminimal*
\begin{proof}
Consider any alternating $\hb$-path $\Path\colon a\LPath{\hb}b$.
If $\Path$ is already minimal, we are done.
Otherwise, $\Path$ contains two $\po$-edges on the same thread, i.e., it has the form
\[
a \LPath{\hb^?} c \LTo{\po} d \LPath{\hb^?} e \LTo{\po} f \LPath{\hb^?} b\ ,
\]
where $\tid(c)=\tid(f)$.
We argue that $c\LTo{\po}f$; otherwise, we have $f\LTo{\po}c$, and thus $f \LTo{\po} c \LTo{\po} d \LPath{\hb^?} e \LTo{\po} f$ is an $\hb$-cycle, contradicting \cref{lem:hbirr}.
We can therefore replace the sub-path from $c$ to $f$ by the single $\po$-edge $c\LTo{\po}f$, yielding
\[
a \LPath{\hb^?} c \LTo{\po} f \LPath{\hb^?} b\ ,
\]
which has smaller length.
Applying this procedure exhaustively yields a minimal $\hb$-path $a\LPath{\hb}b$.
\end{proof}

\lemwritetowritenobridgeincreasing*
\begin{proof}
We analyze each case separately.

\Paragraph{\cref{item:ww_nb_non_dec}.}
A minimal $\hb$-path between write events may only contain $\rf$-edges, and write-to-write and read-to-write $\po$-edges. 
By construction, every $\rf$-edge is non-decreasing. 
Furthermore, by \cref{lem:po_strict_monotonic}, every write-to-write and read-to-write $\po$-edge that does not contain bridge events is non-decreasing.
Therefore, we have $i\leq j$, as desired. 

\Paragraph{\cref{item:ww_nb_inc}.}
Similarly to the case above, $\Path$ may only contain non-decreasing edges.
Furthermore, since $\tid(\wt_i)\neq \tid(\wt_j)$, $\Path$ must contain at least one $\rf$-edge, and, therefore, at least one read-to-write $\po$-edge.
This $\po$-edge is increasing due to \cref{item:po_increasing} of \cref{lem:po_strict_monotonic}.
Since all edges of $\Path$ are non-decreasing, and it has at least one increasing edge, we have $i<j$, as desired.
\end{proof}

\lemwritetowritenobridge*
\begin{proof}
    Let a bridge edge be an $\rf$-edge connecting bridge events.
    Since every inner event in a minimal $\hb$-path is adjacent to an $\rf$-edge, any bridge event in $\Path$ must be an endpoint of a bridge edge in $\Path$.
    However, any bridge edge crosses between $\ThreadDom_x$ and $\ThreadDom_y$, so returning to the originating thread domain requires at least a second bridge edge.
    These two bridge edges would share a thread (from the topology in \cref{fig:pcp_top}), introducing two $\po$-edges on that thread into $\Path$, contradicting the fact that $\Path$ is minimal.
\end{proof}

\lemrrbridgemonotonic*
\begin{proof}
    From \cref{lem:hb_minimal}, there exists a minimal $\hb$-path $\Path\colon\rd_i\LPath{\hb}\rd_j$.
    Since there are no outgoing $\hb$-edges from $\{\tx,\txp\}$ or $\{\ty,\typ\}$ to any other thread, $\Path$ is fully contained in one of these two sets, and in particular contains no bridge events.
    If $\Path$ is a single $\po$-edge, then $i < j$ holds by definition.
    Otherwise, $\Path$ has the form $\rd_i \LTo{\po} \wt_{\ell} \LPath{\hb^?} \wt_{\ell'} \LTo{\rf} \rd_{\ell'} \LTo{\po^?} \rd_j$;
    we bound the indices in three steps.
    \begin{compactitem}
        \item $i < \ell$: by \cref{item:po_increasing} of \cref{lem:po_strict_monotonic}, every $\po$-edge from a read to a non-bridge write is increasing.
        \item $\ell \leq \ell'$: this sub-path only contains $\rf$ and write-to-write $\po$-edges. 
        The former are non-decreasing by construction;
        the latter are non-decreasing by \cref{item:po_non_decreasing} of \cref{lem:po_strict_monotonic}.
        \item $\ell' \leq j$: all decreasing read-to-read $\po$-edges are listed in \cref{lem:decreasing_r_r_po}; 
        for all of them, the earlier event reads from a thread in $\{\tax,\tbx,\tay,\tby\}$. 
        Since $\rd_{\ell'}$ reads from $\wt_{\ell'}$, and $\wt_{\ell'}$ is in one of $\{\tx,\txp\}$ or $\{\ty,\typ\}$, $\rd_{\ell'}\LTo{\po^?}\rd_j$ must be non-decreasing.
        
    \end{compactitem}
    Combining these inequalities yields $i < \ell \leq \ell' \leq j$, hence $i < j$.
\end{proof}

\lemwwdec*
\begin{proof}
    We prove the case $\tuple{\wt_i(\taxp, \vxap), \wt_j(\tax, \vxa)}$; the others are analogous.

    \cref{lem:hb_minimal} implies that there exists a minimal $\hb$-path $\Path\colon\wt_i\LPath{\hb}\wt_j$. 
    Since $\tid(\wt_i),\tid(\wt_j)\in\ThreadDom_x$, \cref{lem:write_to_write_no_bridge} implies that $\Path$ contains no bridge events.
    From \cref{fig:pcp_top}, the only threads reachable from $\wt_i$ that can reach $\wt_j$ without crossing bridge edges (i.e., $\rf$-edges connecting bridge events) are $\{\tax,\taxp\}$, so $\Path$ is fully contained in $\{\tax,\taxp\}$.
    All internal events of $\Path$ are endpoints of $\rf$-edges connecting the threads $\tax$ and $\taxp$; all such edges are in locations $\vzax$ and $\vzaxp$, so every internal event of $\Path$ has location in $\{\vzax,\vzaxp\}$.

    Since $\tid(\wt_i)\neq\tid(\wt_j)$, $\Path$ contains at least one $\rf$-edge, and so its final edge is a read-to-write $\po$-edge of the form $\rd_\ell(\tax,\vzax)\LTo{\po}\wt_j(\tax,\vxa)$.
    Let an iteration be an execution of the inner loop of $\tax$.
    Since $\rd(\tax, \vzax)$ is not executed in the first iteration of $\tax$, and it is executed exactly once in every other iteration, the event $\rd_\ell(\tax,\vzax)$ is executed in the $(\ell+1)$-th iteration of $\tax$.
    Since $\wt(\tax,\vxa)$ is executed exactly once in each iteration of $\tax$, the event $\wt_j(\tax,\vxa)$ is executed in the $j$-th iteration of $\tax$. 
    Furthermore, $\rd_\ell(\tax,\vzax)$ is the last event of the $(\ell+1)$-th iteration of $\tax$, so if $\rd_\ell(\tax,\vzax)\LTo{\po}\wt_j(\tax,\vxa)$, then it must be the case that $\ell +1 < j$, or, equivalently  $\ell < j-1$.
    The remainder of $\Path$ contributes only non-decreasing steps (write-to-write and read-to-write $\po$-edges by \cref{item:po_non_decreasing} of \cref{lem:po_strict_monotonic}, and $\rf$-edges by construction), so $i\leq\ell < j-1$.
\end{proof}

\lemhbconflictingwritesstrictlymonotonic*
\begin{proof}
    Consider two write events $\wt_i$, $\wt_j$ such that $\wt_i\LTo{\hb}\wt_j$ and $\lloc(\wt_i)=\lloc(\wt_j)=u$, and we argue that $i<j$.
    If $\tid(\wt_i)=\tid(\wt_j)$, the result follows from the irreflexivity of $\hb$:~ if $i\geq j$, then \Revision{$\wt_j\LTo{\po^?}\wt_i$} by construction, which, together with $\wt_i\LTo{\hb}\wt_j$, would violate \cref{lem:hbirr}.
    Otherwise, we have that $u$ must be written by more than one thread, which implies that $u \in \{\vxa,\vxap,\vya,\vyap,\vxb,\vxbp,\vyb,\vybp\}$.
    Note that $u$ is a non-bridge location, thus neither $\wt_i$ nor $\wt_j$ is a bridge event.
    By \cref{lem:hb_minimal}, there exists a minimal $\hb$-path $\Path\colon\wt_i\LPath{\hb}\wt_j$.
    Furthermore, the $x$-locations $\vxa$, $\vxap$, $\vxb$, $\vxbp$ are written exclusively by threads in $\ThreadDom_x$, and the $y$-locations $\vya$, $\vyap$, $\vyb$, $\vybp$ exclusively by threads in $\ThreadDom_y$.
    Therefore, \cref{lem:write_to_write_no_bridge} implies that $\Path$ does not contain bridge events.
    Then, \cref{lem:write_to_write_no_bridge_increasing} applies to conclude that $\Path$ is increasing, thus $i<j$.
\end{proof}

\lemhbcmonotonic*
\begin{proof}
Consider two events $\wt_i,\rd_j$ such that $\wt_i\LTo{\hb}\rd_j$ and $\lloc(\wt_i)=\lloc(\rd_j)=u$, and we argue that $i\leq j$.
By \cref{lem:hb_minimal}, there exists a minimal $\hb$-path $\Path:\wt_i \LPath{\hb} \rd_j$.
The proof proceeds by case analysis on the shape of $\Path$, starting with simple cases and then handling the general form.

\Paragraph{$\Path$ is a single $\po$-edge.}
This implies that $\tid(\wt_i)=\tid(\rd_j)$, so $\tid(\wt_i)$ both writes and reads on $u$. 
The pair $\tuple{\tid(\wt_i),u}$ can only be one of the following:
\[
\begin{array}{llll}
\tuple{\tx, \vxa}, &\tuple{\tx, \vxb}, &\tuple{\ty, \vya}, &\tuple{\ty, \vyb} \\
\tuple{\txp, \vxap}, &\tuple{\txp, \vxbp}, &\tuple{\typ, \vyap}, &\tuple{\typ, \vybp} \\
\end{array}
\]
In each case, writes and reads to $u$ alternate in program order, so $\po^{\tid(\wt_i)} \cap (\WriteDom_{u}\times\ReadDom_{u})$ is non-decreasing, and thus $i\leq j$.

\Paragraph{$\Path$ ends in an $\rf$-edge.}
Then $\Path\colon \wt_i \LPath{\hb^?} \wt_j \LTo{\rf} \rd_j$.
Since $\lloc(\wt_j)=u$, \cref{lem:hb_conflicting_writes_strictly_monotonic} yields $i< j$.

\Paragraph{$\Path$ ends in a $\po$-edge.}
Then $\Path\colon \wt_i \LPath{\hb^?} \wt_{\ell} \LTo{\rf} \rd_{\ell} \LTo{\po} \rd_j$.
This case is significantly more involved, and will be split further. 
The sub-cases are determined by whether $\lloc(\wt_\ell)=u$ and whether $\wt_i$ and  $\rd_j$ are bridge events.

\paragraph{$\wt_i$ and $\rd_j$ are bridge events.}
Assume $\lloc(\wt_i)=\lloc(\rd_j)=\vzaxy$, so $\tid(\wt_i)=\tax$ and $\tid(\rd_j)=\tay$.
The other cases for $\tuple{\lloc(\wt_i), \tid(\wt_i), \tid(\rd_j)}$, namely, $\tuple{\vzaxyp, \tay, \tax}$, $\tuple{\vzbxy, \tbx, \tby}$, and $\tuple{\vzbxyp, \tby, \tbx}$, are symmetric. 
Every $\hb$-path from $\tax$ to $\tay$ must pass through an $\rf$-edge at location $\vzaxy$ (see \cref{fig:pcp_top}), so $\Path$ has the form $\wt_i\LPath{\hb^?}\wt_{\ell'}\LTo{\rf}\rd_{\ell'}\LPath{\hb^?}\rd_j$ with $\lloc(\wt_{\ell'})=\lloc(\rd_{\ell'})=\vzaxy$.
This implies that $\tid(\wt_{\ell'})=\tax$ and  $\tid(\rd_{\ell'})=\tay$, since $\tax$ and $\tay$ are the only threads that write/read in $\vzaxy$, respectively.
Since $\Path$ is minimal and $\tid(\wt_i)=\tid(\wt_{\ell'})$, the sub-path $\wt_i\LPath{\hb^?}\wt_{\ell'}$ must be either empty or a $\po$-edge, giving $i\leq\ell'$.
For the same reason, $\rd_{\ell'}\LPath{\hb^?}\rd_j$ must also be either empty or a $\po$-edge, giving $\ell'\leq j$;
combining these inequalities yields $i\leq j$, as desired.

\paragraph{$\wt_i$ and $\rd_j$ are not bridge events, $\lloc(\wt_\ell) = u$.}
If $\wt_i=\wt_{\ell}$, then clearly $i=\ell$.
Otherwise, we have $\wt_i\LTo{\hb}\wt_{\ell}$, and  \cref{lem:hb_conflicting_writes_strictly_monotonic} implies that $i<\ell$.
Thus, in either case, we have $i\leq\ell$.
Moreover, by definition ($\po$-ordered same-location reads have increasing indices), we have $\ell < j$. 
We thus arrive at $i<j$.

\paragraph{$\wt_i$ and $\rd_j$ are not bridge events, $\lloc(\wt_\ell) \neq u$.}
Since only bridge locations are accessed by threads in both $\ThreadDom_x$ and $\ThreadDom_y$, and neither $\wt_i$ nor $\rd_j$ is a bridge event, we have that either both events belong to threads in $\ThreadDom_x$, or both belong to threads in $\ThreadDom_y$.
Then \cref{lem:write_to_write_no_bridge} implies that $\Path$ does not contain bridge events.
Since $\tid(\rd_j)=\tid(\rd_\ell)$ and $\lloc(\rd_j)\neq \lloc(\rd_\ell)$, we have that $\tid(\rd_j)$ reads from at more one non-bridge location.
The only threads for which this is the case are $\{\tx,\ty,\txp,\typ\}$, which allows us to assert that  $\tid(\rd_j)$ is a \Revision{checker} thread.
We now consider two cases for $\tid(\wt_i)$: whether it is a \Revision{checker} or a guesser thread.

\begin{compactitem}
\item 
\textit{$\tid(\wt_i)$ is a \Revision{checker} thread.}
Since $\Path$ contains no bridge events, \cref{lem:write_to_write_no_bridge_increasing} gives $i\leq\ell$.
From \cref{fig:pcp_top}, \Revision{checker} threads only have $\hb$-paths to \Revision{checker} threads, so $\tid(\wt_{\ell})\in\{\tx,\ty,\txp,\typ\}$.
The final edge $\rd_\ell\LTo{\po}\rd_j$ is non-decreasing according to \cref{lem:decreasing_r_r_po},  since $\rd_\ell$ reads from a \Revision{checker} thread.
Therefore $\ell\leq j$, and we arrive at $i\leq j$.

\item 
\textit{$\tid(\wt_i)$ is a guesser thread.}
We analyze the case of $\tid(\wt_i)\in\{\tax,\taxp\}$, which implies that $u\in\{\vxa,\vxap\}$.
The only reads in these locations are in $\{\tx,\txp\}$, therefore $\tid(\rd_j)\in\{\tx,\txp\}$.
The other cases for the tuple $\tuple{\tid(\wt_i), \tid(\rd_j), u}$, namely, when it is one of$\tuple{\{\tay,\tayp\}, \{\ty,\typ\}, \{\vya,\vyap\}}$, $\tuple{\{\tbx,\tbxp\}, \{\tx,\txp\}, \{\vxb,\vxbp\}}$, and $\tuple{\{\tby,\tbyp\}, \{\ty,\typ\}, \{\vyb,\vybp\}}$, are symmetric.
Any $\hb$-path from some thread in $\{\tax,\taxp\}$ to some thread in $\{\tx,\txp\}$ must pass through an $\rf$-edge connecting a guesser thread to a \Revision{checker} thread.
Specifically, it must contain an edge $\wt_{\ell'}\LTo{\rf}\rd_{\ell'}$ with $\lloc(\wt_{\ell'})\in\{\vxa,\vxap\}$ (see \cref{fig:pcp_top,FIG:RAEX});
therefore, $\Path$ has the form $\wt_i\LPath{\hb^?}\wt_{\ell'}\LTo{\rf}\rd_{\ell'}\LPath{\hb^?}\rd_j$.

We consider the cases when $\lloc(\wt_i)=\lloc(\wt_{\ell'})$ and when $\lloc(\wt_i)\neq \lloc(\wt_{\ell'})$.

\begin{compactitem}

    \item \textit{$\lloc(\wt_i)=\lloc(\wt_{\ell'})$.}
    If $\wt_i=\wt_{\ell'}$ then $i=\ell'$;
    otherwise, we have that $\wt_i\LPath{\hb}\wt_{\ell'}$ and $\lloc(\wt_i)=\lloc(\wt_{\ell'})$, so \cref{lem:hb_conflicting_writes_strictly_monotonic} gives $i<\ell'$.
    In either case, $i\leq\ell'$.
    Similarly, if $\rd_{\ell'}=\rd_j$ then $\ell'=j$; otherwise $\lloc(\rd_{\ell'})=\lloc(\rd_j)$ and \cref{lem:rr_bridge_monotonic} gives $\ell'<j$.
    In either case, $\ell'\leq j$.
    Combining the two inequalities yields $i\leq j$.
    \item \textit{$\lloc(\wt_i) \neq \lloc(\wt_{\ell'})$.} 
    Recall that $\lloc(\wt_i), \lloc(\wt_{\ell'})\in \{\vxa,\vxap\}$, and
    $\tid(\wt_i)\in\{\tax,\taxp\}$, and $\wt_{\ell'}$ is written by a guesser thread.
    Since different guesser threads write each of $\{\vxa,\vxap\}$, it must be the case that $\tid(\wt_i)\neq\tid(\wt_{\ell'})$, so \cref{item:ww_nb_inc} of \cref{lem:write_to_write_no_bridge_increasing} gives $i<\ell'$.
    The path $\rd_{\ell'}\LPath{\hb}\rd_j$ has the form $\rd_{\ell'}\LPath{\hb^?}\rd_m\LTo{\po^?}\rd_j$, where the sub-path to $\rd_m$ consists only of non-decreasing edges ($\rf$-edges by construction, read-to-write $\po$-edges by \cref{lem:po_strict_monotonic}, as $\Path$ does not contain bridge events), giving $\ell'\leq m$ and thus $i<m$.
    By \cref{lem:decreasing_r_r_po}, the final edge $\rd_m\LTo{\po^?}\rd_j$ decreases the index by at most one, so $m\leq j+1$; 
    combined with $i < m$, this yields $i\leq j$.
\end{compactitem}
\end{compactitem}
\end{proof}

\section{Proofs of \cref{SEC:DECBCS}}\label{SEC:APP_CONTEXT_SWITCHES}

\lemmairreduciblesmall*
\begin{proof}

We let $\pi=\pi_{1}\cdot \pi_2 \cdots \pi_{\ell}$, and let $\SummaryDom_{\LinearTrace}$ denote the set of possible summaries $\tuple{\SummaryStates_{\LinearTrace}, \SummaryLastValue_{\LinearTrace}, \SummaryExternalRf_{\LinearTrace}}$. 
The number of possible values for $\SummaryStates_{\LinearTrace}$, $\SummaryLastValue_{\LinearTrace}$, and $\SummaryExternalRf_{\LinearTrace}$ are $2^{|\ConcProg.\States|}$, $(|\ValueDom|+1)^{|\LocationDom|}$, and $2^{|\LocationDom|}$, respectively. 
\Revision{Hence, $|\SummaryDom_{\LinearTrace}| = 2^{|\ConcProg.\States|+|\LocationDom|}\cdot(|\ValueDom|+1)^{|\LocationDom|}$.}
We first show that the size of each run $\pi_c$, with $c\in[\ell]$, is bounded. 
If $|\pi_c| \leq |\SummaryDom_{\LinearTrace}|$, we are done. Otherwise, we proceed as follows.

\SubParagraph{Defining impeding events.}
Let $\ReadDom_{\pi_c}$ be the set of events in $\pi_{c+1},\ldots ,\pi_{\ell}$ that read from some event in $\pi_c$ via an $\rf$-edge.
Define $\WriteDom_{\pi_c}$ as the image of $[\ReadDom_{\pi_c}];\rf^{-1}$, i.e., the writing events of $\pi_c$ that are read by some event in $\ReadDom_{\pi_c}$. 
We call $\WriteDom_{\pi_c}$ the impeding events of $\pi_c$.
Note that $|\WriteDom_{\pi_c}|\leq |\ReadDom_{\pi_c}|$. 

\paragraph{Partitioning $\pi_c$ into segments.}
Partition $\pi_c$ into $m+1$ sequences of events $S_1, \ldots, S_m, S_{m+1}$, where $|S_{\CSid}|=|\SummaryDom_{\LinearTrace}|+1$ for $\CSid\in[m]$ and $|S_{m+1}|\leq|\SummaryDom_{\LinearTrace}|$. 
By the pigeonhole principle, each segment $S_{\CSid}$ with $\CSid \in [m]$ contains a pair of events $\tuple{\CSei^{\CSid}, \CSej^{\CSid}}$ with $\CSei^{\CSid} <_{\pi} \CSej^{\CSid}$ and $\Summary_{\LinearTrace}(\CSei^{\CSid})=\Summary_{\LinearTrace}(\CSej^{\CSid})$. 
Since $\tuple{G,(\pi_1,\ldots,\pi_{\ell})}$ is irreducible, the pair $\tuple{\CSei^{\CSid},\CSej^{\CSid}}$ is not collapsible, which means at least one of the following must hold:
\begin{compactenum}
    \item There exists $\CSeb \in \WriteDom_{\pi_c}$ such that $\CSei^{\CSid} <_{\pi} \CSeb \leq_{\pi} \CSej^{\CSid}$. \label{item:impede_1}
    \item For some $x \in \LocationDom$ and $\event \in (G.\Events \setminus G.\Events^{\tid(\CSei^{\CSid})})$, we have $\LatestWrite_{\LinearTrace}(\CSei^{\CSid},x)\LTo{\hb} \event$ but $\LatestWrite_{\LinearTrace}(\CSej^{\CSid},x)\nLTo{\hb} \event$. 
    \Revision{This implies there exists $\CSeb \in \WriteDom_{\pi_c}$ such that $\LatestWrite_{\LinearTrace}(\CSei^{\CSid},x) \leq_{\pi} \CSeb <_{\pi} \LatestWrite_{\LinearTrace}(\CSej^{\CSid},x)$.} 
    \Revision{Indeed, consider an $\hb$-path from $\LatestWrite_{\LinearTrace}(\CSei^{\CSid},x)$ to $\event$, and let $\CSeb$ be the source of its first $\rf$-edge whose target is not in $\tid(\CSei^{\CSid})$.} 
    \Revision{Since $\rf$-edges within a thread agree with $\po$ (by irr-hb), we have $\LatestWrite_{\LinearTrace}(\CSei^{\CSid},x) \LTo{\po^?} \CSeb$.} 
    \Revision{Moreover, $\CSeb \LTo{\po} \LatestWrite_{\LinearTrace}(\CSej^{\CSid},x)$, as otherwise $\LatestWrite_{\LinearTrace}(\CSej^{\CSid},x)\LTo{\hb} \event$.} 
    \Revision{Since $\pi_c$ is a contiguous run and $\pi$ extends $\po$, we have $\CSeb \in \pi_c$, and since $\CSeb$ is read by an event of a different thread, which appears in a later run, $\CSeb \in \WriteDom_{\pi_c}$.}\label{item:impede_2}
\end{compactenum}
In either case, we say $\CSeb$ \emph{impedes} $\tuple{\CSei^{\CSid},\CSej^{\CSid}}$ from being reduced.

\paragraph{Bounding the number of impeded pairs per impeding event.}
We now determine how many pairs $\tuple{\CSei^{\CSid},\CSej^{\CSid}}$ can be impeded by a single event $\CSeb \in \WriteDom_{\pi_c}$.
As helpful notation, let $\EarliestWrite_{\LinearTrace}(\event,x)$ denote the minimal write $\wt(t, x)$ with respect to $\leq_{\pi_{c}}$ such that $\event \leq_{\pi_{c}} \wt(t,x)$, or $\bot$ if no such write exists.

\begin{compactenum}
    \item If $\CSeb$ impedes $\tuple{\CSei^{\CSid},\CSej^{\CSid}}$ via \ref{item:impede_1}, then $\CSei^{\CSid} <_{\pi} \CSeb \leq_{\pi} \CSej^{\CSid}$. 
    There is at most one $\CSid\in [m]$ satisfying this condition.
    
    \item \Revision{If $\CSeb$ impedes $\tuple{\CSei^{\CSid},\CSej^{\CSid}}$ via \ref{item:impede_2}, then for the location $x \in \LocationDom$ of \ref{item:impede_2}, we have $\LatestWrite_{\LinearTrace}(\CSei^{\CSid},x) \leq_{\pi} \CSeb <_{\pi} \LatestWrite_{\LinearTrace}(\CSej^{\CSid},x)$.}
    \Revision{Since $\LatestWrite_{\LinearTrace}(\CSei^{\CSid},x) \neq \LatestWrite_{\LinearTrace}(\CSej^{\CSid},x)$ by \ref{item:impede_2}, the maximality of $\LatestWrite_{\LinearTrace}(\CSei^{\CSid},x)$ yields $\LatestWrite_{\LinearTrace}(\CSei^{\CSid},x) \leq_{\pi} \CSei^{\CSid} <_{\pi} \LatestWrite_{\LinearTrace}(\CSej^{\CSid},x) \leq_{\pi} \CSej^{\CSid}$.} 
    \Revision{We consider three subcases:} 
    \begin{compactitem}
        \item \Revision{If $\CSei^{\CSid} <_{\pi} \CSeb$, then $\CSeb$ also impedes $\tuple{\CSei^{\CSid},\CSej^{\CSid}}$ via \ref{item:impede_1}, and this pair is already counted in the first case.} 
        \item \Revision{If $\LatestWrite_{\LinearTrace}(\CSei^{\CSid},x) <_{\pi} \CSeb \leq_{\pi} \CSei^{\CSid}$, then $\CSeb$ is not a write on $x$, by the maximality of $\LatestWrite_{\LinearTrace}(\CSei^{\CSid},x)$, and $\CSei^{\CSid} <_{\pi} \EarliestWrite_{\LinearTrace}(\CSeb,x) \leq_{\pi} \CSej^{\CSid}$. There is at most one $\CSid$ satisfying this condition.} 
        \item \Revision{If $\CSeb = \LatestWrite_{\LinearTrace}(\CSei^{\CSid},x)$, then $\CSeb$ is a write on $x$, and $\CSeb <_{\pi} \LatestWrite_{\LinearTrace}(\CSej^{\CSid},x)$, as the two writes are distinct by \ref{item:impede_2}.} 
        \Revision{Hence, $\LatestWrite_{\LinearTrace}(\CSei^{\CSidp},x) \neq \CSeb$ for all $\CSidp > \CSid$, since $\CSej^{\CSid} <_{\pi} \CSei^{\CSidp}$ implies $\CSeb <_{\pi} \LatestWrite_{\LinearTrace}(\CSej^{\CSid},x) \leq_{\pi} \LatestWrite_{\LinearTrace}(\CSei^{\CSidp},x)$, i.e., once $\CSeb$ is no longer the latest write on $x$, this cannot be recovered. There is at most one $\CSid$ satisfying this condition.}
    \end{compactitem}
    \Revision{Since the second and third subcases cannot both hold for the same location, for each location $x \in \LocationDom$, event $\CSeb$ impedes at most one pair via \ref{item:impede_2} besides the pair of the first case.}
    Across all locations, $\CSeb$ impedes at most $|\LocationDom|$ such pairs via \ref{item:impede_2}.
\end{compactenum}

Combining both cases, each $\CSeb \in \WriteDom_{\pi_c}$ impedes at most $|\LocationDom|+1$ pairs of the form $\tuple{\CSei^{\CSid},\CSej^{\CSid}}$ from being reduced.

\paragraph{Deriving the bound.}
Since every pair $\tuple{\CSei^{\CSid},\CSej^{\CSid}}$ with $\CSid\in [m]$ must be impeded by some event in $\WriteDom_{\pi_c}$, we have:
\[m \leq |\WriteDom_{\pi_c}|\cdot (|\LocationDom| +1)\]

Since $|\pi_c| = m \cdot (|\SummaryDom_{\LinearTrace}|+1) + |S_{m+1}| \leq m \cdot (|\SummaryDom_{\LinearTrace}|+1) + |\SummaryDom_{\LinearTrace}|$, we obtain:
\[|\pi_c|\leq |\WriteDom_{\pi_c}| \cdot(|\LocationDom| +1) \cdot (|\SummaryDom_{\LinearTrace}|+1) + |\SummaryDom_{\LinearTrace}|\]

Since $|\WriteDom_{\pi_c}| \leq |\ReadDom_{\pi_c}| \leq |\pi_{c+1}|+\ldots+|\pi_{\ell}|$, we have:
\[|\pi_c|\leq (|\pi_{c+1}| + \ldots + |\pi_{\ell}|) \cdot (|\LocationDom| +1) \cdot (|\SummaryDom_{\LinearTrace}|+1) + |\SummaryDom_{\LinearTrace}|\]

\Revision{Let $N_c=\sum_{j=c}^{\ell}|\pi_j|$ for $c\in[\ell]$, with $N_{\ell+1}=0$, and let $B=(|\SummaryDom_{\LinearTrace}|+1)\cdot(|\LocationDom|+1)+1$.}
\Revision{Adding $N_{c+1}$ to both sides of the above inequality, we obtain}
\[N_c\leq |\SummaryDom_{\LinearTrace}| + B\cdot N_{c+1}\]
\Revision{Unrolling this recurrence, and using that $B-1>|\SummaryDom_{\LinearTrace}|\cdot(|\LocationDom|+1)$ and $B\leq(|\SummaryDom_{\LinearTrace}|+2)\cdot(|\LocationDom|+1)$, we have} 
\[|G.\Events|=N_1\leq |\SummaryDom_{\LinearTrace}|\cdot\frac{B^{\ell}-1}{B-1} < \frac{B^{\ell}}{|\LocationDom|+1}\leq (|\SummaryDom_{\LinearTrace}|+2)^{\ell}\cdot(|\LocationDom|+1)^{\ell-1}\] 
\Revision{Finally, substituting $|\SummaryDom_{\LinearTrace}|$ and using that $\log(|\LocationDom|+1)\leq|\LocationDom|\leq|\LocationDom|\cdot\log(|\ValueDom|+1)$, we conclude that $|G.\Events|=2^{O(\ell\cdot (|\ConcProg.\States| + |\LocationDom|\cdot \log(|\ValueDom|+1)))}$.} 
\end{proof}

\lemmairreduciblesmallrmw*
\begin{proof}
This proof follows the same steps as \cref{lem:irreducible_small}. We detail here the key changes. 
\Revision{Let $k=|G.\Events\cap\RMWDom|$.}
\Revision{The summaries are unchanged, so $|\SummaryDom_{\LinearTrace}|$ is as in the proof of \cref{lem:irreducible_small}.}
\Revision{As with $\LatestWrite_{\LinearTrace}$, we let $\EarliestWrite_{\LinearTrace}(\event,x)$ range over writing events, i.e., plain writes and RMWs; accordingly, in the second subcase of the proof of \cref{lem:irreducible_small}, $\CSeb$ is not a writing event on $x$.}

\SubParagraph{Defining impeding events.} We add another set of impeding events, say, $\RMWDom_{\pi_c}$, which corresponds to the set of all RMW events in $\pi_c$. 

\paragraph{Partitioning $\pi_c$ into segments.} In this part, we must also incorporate condition \ref{item:cond_red_5}, which states that $\rmw \in \pi_c$ may block a pair $\tuple{\CSei^{\CSid},\CSej^{\CSid}}$ from being reduced if, for some location $x$, we have $\LatestWrite_{\LinearTrace}(\CSei^{\CSid}, x)\neq \LatestWrite_{\LinearTrace}(\CSej^{\CSid}, x)$ and $\op(\LatestWrite_{\LinearTrace}(\CSei^{\CSid}, x))=\rmw$. 

To do so, also consider the following condition for a pair  $\tuple{\CSei^{\CSid},\CSej^{\CSid}}$ to not be collapsible:
\begin{enumerate}[label=(\arabic*)]
    \addtocounter{enumi}{2}
    \item There exist \Revision{$x \in \LocationDom$ and} $\CSeb \in \RMWDom_{\pi_c}$ such that 
    $\LatestWrite_{\LinearTrace}(\CSei^{\CSid},x) = \CSeb$ and $\LatestWrite_{\LinearTrace}(\CSei^{\CSid},x) \neq \LatestWrite_{\LinearTrace}(\CSej^{\CSid},x)$.\label{item:impede_3}
\end{enumerate}

We can argue that if $\CSeb$ impedes $\tuple{\CSei^{\CSid},\CSej^{\CSid}}$ via condition \ref{item:impede_3}, then $\CSeb$ cannot simultaneously satisfy the same impeding condition for $\tuple{\CSei^{\CSidp},\CSej^{\CSidp}}$. The reason for this is that for any $\CSid < \CSidp$, $\LatestWrite_{\LinearTrace}(\CSei^{\CSid},x) = \CSeb$ and $\LatestWrite_{\LinearTrace}(\CSei^{\CSid},x) \neq \LatestWrite_{\LinearTrace}(\CSej^{\CSid},x)$ implies $\LatestWrite_{\LinearTrace}(\CSei^{\CSidp},x) \neq \CSeb$. Intuitively, after $\CSeb$ is not the latest writing event on $x$ anymore, that property cannot be recovered. 

\paragraph{Bounding the number of impeded pairs per impeding event.} Per the argument above, each $\CSeb \in \RMWDom_{\pi_c}$ impedes at most one pair of the form $\tuple{\CSei^{\CSid},\CSej^{\CSid}}$ from being reduced.

\paragraph{Deriving the bound.} We now initially have $m \leq |\WriteDom_{\pi_c}|\cdot (|\LocationDom| +1) + |\RMWDom_{\pi_c}|$.
\Revision{Following the same steps as in \cref{lem:irreducible_small}, with $N_c$ and $B$ defined as there, we obtain}
\[N_c\leq |\SummaryDom_{\LinearTrace}| + |\RMWDom_{\pi_c}|\cdot(|\SummaryDom_{\LinearTrace}|+1) + B\cdot N_{c+1}\] 
\Revision{Unrolling this recurrence, and using that $B^{c-1}\leq B^{\ell-1}$ for all $c\in[\ell]$ and $\sum_{c=1}^{\ell}|\RMWDom_{\pi_c}|=k$, we have}
\[|G.\Events|=N_1\leq B^{\ell-1}\cdot\left(\ell\cdot|\SummaryDom_{\LinearTrace}| + k\cdot(|\SummaryDom_{\LinearTrace}|+1)\right)\leq (|\SummaryDom_{\LinearTrace}|+2)^{\ell}\cdot(|\LocationDom|+1)^{\ell-1}\cdot(\ell+k)\] 
\Revision{Finally, substituting $|\SummaryDom_{\LinearTrace}|$ as in the proof of \cref{lem:irreducible_small}, and using that $\ell+k\leq 2^{\ell}\cdot(k+1)$, we conclude that $|G.\Events|=(k+1)\cdot 2^{O(\ell\cdot (|\ConcProg.\States| + |\LocationDom|\cdot \log(|\ValueDom|+1)))}$.}

\end{proof}
\end{document}